%% file: journal.tex
\documentclass[12pt]{elsarticle}
\usepackage[paper=a4paper,margin=2cm,centering]{geometry}

\pdfoutput=1
\biboptions{square,round, numbers}

\usepackage[latin1]{inputenc}
\usepackage[T1]{fontenc}
\usepackage{slantsc}

\usepackage{xspace}
\usepackage{paralist}
\usepackage{amsbsy,amssymb,amsmath,amsfonts,amsthm}
\usepackage{fancybox}
\usepackage{color}
\usepackage{graphicx}
\usepackage{wrapfig}
\usepackage{url}
\usepackage{subfig}
\usepackage{verbatim}

\usepackage{graphicx}

\usepackage{rotating}
\usepackage{microtype}

\usepackage{multirow}

\usepackage{algorithm}
\usepackage{algorithmic}

\usepackage{array,tabularx,multirow}

\usepackage{multibib}
\newcites{app}{Additional References}

\usepackage{units}
\usepackage{calc}

\makeatletter
\def\fakemarginpar{%
  \ifhmode
    \@bsphack
    \@floatpenalty -\@Mii
  \else
    \@floatpenalty-\@Miii
  \fi
  \ifinner
    \@parmoderr
    \@floatpenalty\z@
  \else
    \@next\@currbox\@freelist{}{}%
    \@next\@marbox\@freelist{\global\count\@marbox\m@ne}%
       {\@floatpenalty\z@
        \@fltovf\def\@currbox{\@tempboxa}\def\@marbox{\@tempboxa}}%
  \fi
  \@ifnextchar [\@xmpar\@ympar}
\makeatother

\makeatletter
\let\oldparagraph\paragraph
\renewcommand{\paragraph}[1]{\oldparagraph*{\textbf{#1}}}
\makeatother

\newcolumntype{L}{>{\raggedright\arraybackslash}X}
\newcolumntype{C}{>{\centering\arraybackslash}X}
\newcolumntype{R}{>{\raggedleft\arraybackslash}X}

\DeclareGraphicsExtensions{.png}
\DeclareGraphicsExtensions{.jpg}
\newcommand{\ignore}[1]{}

\newcommand{\Z}{\mathbb{Z}}

\newcommand{\N}{\mathbb{N}}
\newcommand{\RR}{\mathbb{R}}
\newcommand{\C}{\mathbb{C}}

\newcommand{\cgal}{\textsc{Cgal}\xspace}

\newcommand{\gmp}{\textsc{Gmp}\xspace}
\newcommand{\ntl}{\textsc{Ntl}\xspace}

\newcommand{\rs}{\textsc{Rs}\xspace}
\newcommand{\isolate}{\textsc{Isolate}\xspace}
\newcommand{\lgp}{\textsc{Lgp}\xspace}
\newcommand{\res}{\operatorname{res}}
\newcommand{\bdesc}{$\textsc{Bdc}$\xspace}

\newcommand{\bs}{\textsc{Bisolve}\xspace}
\newcommand{\bsproject}{\textsc{BiProject}\xspace}
\newcommand{\bsseparate}{\textsc{Separate}\xspace}
\newcommand{\bsvalidate}{\textsc{Validate}\xspace}

\newcommand{\ca}{\textsc{GeoTop}\xspace}
\newcommand{\slowca}{\textsc{GeoTop-BS}\xspace}
\newcommand{\fastca}{\textsc{Top-NT}\xspace}
\newcommand{\caproject}{\textsc{Project}\xspace}
\newcommand{\calift}{\textsc{Lift}\xspace}
\newcommand{\fastlift}{\textsc{Lift-NT}\xspace}
\newcommand{\slowlift}{\textsc{Lift-BS}\xspace}
\newcommand{\caconnect}{\textsc{Connect}\xspace}
\newcommand{\kernelnt}{\textsc{\ca{}AK\_2}\xspace}

\newcommand{\etal}{et~al.\xspace}

\numberwithin{equation}{section}
\numberwithin{figure}{section}

\providecommand{\DeclareCaptionType}[1]{\relax}
\DeclareCaptionType{copyrightbox}

\newtheorem{theorem}{Theorem}
\newtheorem{lemma}{Lemma}

\newtheorem{remark}{Remark}
\newtheorem*{remark*}{Remark}

\begin{document}

\begin{frontmatter}

\title{Exact Symbolic-Numeric Computation of Planar Algebraic Curves}

\author{Eric Berberich\corref{cor1}}
\ead{eric@mpi-inf.mpg.de}

\author{Pavel Emeliyanenko\corref{}}
\ead{asm@mpi-inf.mpg.de}

\author{Alexander Kobel\corref{}}
\ead{akobel@mpi-inf.mpg.de}

\author{Michael Sagraloff\corref{cor2}}
\ead{msagralo@mpi-inf.mpg.de}

\address{Max-Planck-Institut f\"ur Informatik, Campus E1 4, D-66123 Saarbr\"ucken, Germany}

\cortext[cor1]{Principal corresponding author: Tel +49~681~9325~1012, Fax +49~681~9325~1099} 
\cortext[cor2]{Corresponding author: Tel +49~681~9325~1006, Fax +49~681~9325~1099}

\begin{abstract}
We present a novel \emph{certified and complete algorithm to compute 
arrangements of real planar algebraic curves}. It provides a 
geometric-topological analysis of the decomposition of the plane induced 
by a finite number of algebraic curves in terms of a cylindrical algebraic 
decomposition. From a high-level perspective, the overall method splits into two main subroutines, namely an algorithm denoted \bs to isolate the real solutions of a zero-dimensional bivariate system, and an algorithm denoted \ca to analyze a single algebraic curve.

Compared to existing approaches based on elimination techniques, we considerably improve the corresponding lifting steps in both subroutines. As a result, generic position of the input system is never assumed, and thus our algorithm never demands for any change of coordinates.
In addition, we significantly limit the types of involved exact operations, 
that is, we only use resultant and $\gcd$ computations 
as purely symbolic operations. The latter results are achieved by combining techniques from different fields such as (modular) symbolic computation, numerical analysis and algebraic geometry.

We have implemented our algorithms as prototypical contributions to the 
C++-project \cgal. They exploit graphics hardware 
to expedite the symbolic computations. 
We have also compared our implementation with the current 
reference implementations, that is, 
\lgp and Maple's \isolate for polynomial system solving, and 
\cgal's bivariate algebraic kernel for analyses and arrangement 
computations of algebraic curves.
For various series of challenging instances, our exhaustive  
experiments show that the new implementations outperform the existing ones.
\end{abstract}


\end{frontmatter}


\section{Introduction}
\label{sec:introduction} 

Computing the topology of a planar algebraic curve 
\begin{align}
C=V(f)=\{(x,y)\in\mathbb{R}^2:f(x,y)=0\}\label{def:curve}
\end{align}
can be considered as one of the fundamental problems in real algebraic geometry with numerous applications in computational geometry, computer graphics and computer aided geometric design. Typically, the topology of~$C$ is given in terms of a planar graph $\mathcal{G}_{C}$ embedded in $\mathbb{R}^2$ that is isotopic to~$C$.\footnote{$\mathcal{G}_C$ is isotopic to $C$ if there
exists a continuous mapping $\phi:[0,1]\times C\mapsto \RR^2$ with
$\phi(0,C)=C$, $\phi(1,C)=\mathcal{G}_C$ and $\phi(t_0,.):C\mapsto
\phi(t_0,C)$ a homeomorphism for each $t_0\in[0,1]$.} For a geometric-topological analysis, we further require the vertices of $\mathcal{G}_{C}$ to be located on~$C$. In this paper, we study the more general problem of computing an arrangement of a finite set of algebraic curves, that is, the decomposition of the plane into cells of dimensions $0$, $1$ and $2$ induced by the given curves. The proposed algorithm is \emph{certified} and \emph{complete}, and the overall arrangement computation is exclusively carried out in the initial coordinate system. Efficiency of our approach is shown by implementing our algorithm based on the current reference implementation within \cgal\footnote{Computational Geometry Algorithms Library,
\url{www.cgal.org}; see also \url{http://exacus.mpi-inf.mpg.de/cgi-bin/xalci.cgi} for an online demo on arrangement computation.} (see also~\cite{eigenwilligk08,cgal:wfzh-a2-11b}) and comparing it to the most efficient implementations which are currently available.

From a high-level perspective, we follow the same approach as in~\cite{eigenwilligk08,cgal:wfzh-a2-11b}. That is, the arrangement computation is reduced to the geometric-topological analysis of single curves and of pairs of curves. The main contribution of this paper is to provide novel solutions for the basic subtasks needed by these analysis, that is, \emph{isolating the real solutions of a bivariate polynomial system} (\bs) and \emph{computing the topology of a single algebraic curve} (\ca).\\

\noindent\bs: For a given \emph{zero-dimensional} polynomial system $f(x,y)=g(x,y)=0$ (i.e.~there exist only finitely many solutions), with $f,g\in\Z[x,y]$, the algorithm computes disjoint boxes $B_1,\ldots,B_m\subset\RR^2$ for
\emph{all real
solutions}, where each box $B_i$ contains exactly one
solution (i.e.~$B_i$ is isolating). In addition, the boxes can be refined
to an arbitrary small size. \bs is a classical 
elimination method which follows the same basic idea as the \textsc{Grid} method from~\cite{det-asymptotic} for solving a bivariate polynomial system, or the \textsc{Insulate} method from~\cite{sw-exact} for computing the topology of a planar algebraic curve.\footnote{For the analysis of a planar curve $C=\{(x,y)\in\RR^2:f(x,y)=0$\}, it is crucial to find the solutions of $f=f_y=0$. The method in~\cite{sw-exact} uses several projection directions to find these solutions.} Namely, all of them consider several projection directions to derive a set of candidates of possible solutions and eventually identify those candidates which are actually solutions. 

More precisely, we separately eliminate the variables $x$ and $y$ by
means of resultant computations. Then, for each possible candidate (represented as a pair of projected solutions in $x$- and $y$-direction), we check whether it actually constitutes a solution of the given system or not. The
proposed method comes with a number of improvements compared to the aforementioned approaches and also to other existing elimination techniques~\cite{eigenwilligk08,abrw-zeros,khfta-solvingsystems,r-rur,r-rsicms-2010}. 
First, we considerably reduce the amount of purely symbolic computations, 
namely, our method only demands for resultant computation of bivariate polynomials and gcd computation
of univariate polynomials.
Second, our implementation profits from a novel
approach~\cite{emel-pasco-10,emel-ica3pp-10,emel-gcd-11} to compute
resultants and gcds exploiting the power of Graphics Processing Units (GPUs). Here, it is important to remark that, in comparison to the classical resultant computation on the
CPU, the GPU implementation is typically more than $100$-times faster. Our
experiments show that, for the huge variety of considered instances, the symbolic 
computations are no longer a ``global'' bottleneck of an elimination
approach.
Third, the proposed method never uses any kind of a coordinate 
transformation, even for non-generic input.\footnote{The 
system $f=g=0$ is non-generic if there exist two solutions sharing a 
common coordinate.} The latter fact is due to a novel inclusion predicate which 
combines information from the resultant computation and a homotopy 
argument to prove that a certain candidate box is isolating for a solution. 
Since we never apply any change of coordinates, our method
particularly profits in the case where $f$ and
$g$ are sparse, or where we are only interested in solutions
within a given ``local'' box. Finally, we integrated a series of additional
filtering techniques which allow us to considerably speed up the computation 
for the majority of instances.\\ 

\ca: There exist a number of certified and complete approaches to determine the topology of an algebraic curve; we refer the reader to~\cite{cheng.lazard.ea:on,eigenwilligkw07,gn-efficient,kerber-phd,LuisPhD2010} for recent work and further references. At present, only the method from~\cite{eigenwilligkw07} has been extended to arrangement computations of arbitrary algebraic curves~\cite{eigenwilligk08}. Common to all of these approaches is that, in essence, they consider the following three phases:\vspace{0.2cm}
\begin{compactenum}
\item \emph{Projection}: Elimination techniques (e.g.~resultants) are used
to project the \emph{$x$-critical points} (i.e.~points $p$ on the (complex) curve $C=\{(x,y)\in\C^2:f(x,y)=0\}$ with $f_y(p)=0$) of the curve into one dimension. The so obtained projections are called \emph{$x$-critical values}.
\item \emph{Lifting}: For all real $x$-critical values $\alpha$ (as well as for real values in between), we compute the \emph{fiber}, that is, all intersection points of $C$ with a corresponding vertical line $x=\alpha$. 
\item \emph{Connection} (in the analysis of a single curve): The so obtained points are connected by straight line edges in an appropriate manner.\vspace{0.2cm}
\end{compactenum}
In general, the lifting step at an $x$-critical value $\alpha$ has turned out to be the most time-consuming part because it amounts to determining the real roots of a non square-free univariate polynomial $f_{\alpha}(y):=f(\alpha,y)\in\mathbb{R}[y]$ with algebraic coefficients. In all existing approaches, the high computational cost for computing the roots of $f_{\alpha}$ is mainly due to a more comprehensive algebraic machinery such as the computation of subresultants (in~\cite{eigenwilligk08,eigenwilligkw07,gn-efficient}), Gr\"obner basis or a rational univariate representation (in~\cite{cheng.lazard.ea:on}) in order to obtain additional information on the number of distinct real (or complex) roots of $f_{\alpha}$, or the multiplicities of the multiple roots of $f_{\alpha}$. In addition, all except the method from~\cite{cheng.lazard.ea:on} consider a shearing of the curve which guarantees that the sheared curve has no two $x$-critical points sharing the same $x$-coordinate. This, in turn, simplifies the lifting as well as the connection step but for the price of giving up sparseness of the initial input. It turns out that considering such an initial coordinate transformation typically yields larger bitsizes of the coefficients and considerably increased running times; see also~\cite{LuisPhD2010} for extensive experiments. 

For \ca, we achieved several improvements in the lifting step. Namely, as in the algorithm \bs, we managed to reduce the amount of purely symbolic computations, that is, we only use resultants and $\gcd$s, where both computations are outsourced again to graphics hardware. 
Furthermore, based on a result from Teissier~\cite{Gwozdziewicz00formulaefor,Teissier} which relates the intersection multiplicities of the curves $f$, $f_x$ and $f_y$, and the multiplicity of a root of $f_{\alpha}$, we derive additional information about the number $n_{\alpha}$ of distinct complex roots of $f_{\alpha}$. In fact, we compute an upper bound $n_{\alpha}^+$ which matches $n_{\alpha}$ except in the case where the curve $C$ is in a very special geometric location. In the lifting phase, we then combine the information about the number of distinct roots of $f_{\alpha}$ with a certified numerical complex root solver~\cite{Kobel11} to isolate the roots of $f_{\alpha}$. The latter symbolic-numeric step applies as an efficient filter denoted \fastlift that is effective in 
almost all cases. In case of a rare failure (due to a special geometric configuration), we fall back to a complete method
\slowlift which is based on \bs. In addition, we also provide a simple test based on a single modular computation only to detect (in advance) special configurations, where \fastlift may fail.
Considering a generic coordinate transformation, it can be further proven that \fastlift generally succeeds. We remark that the latter result is more of theoretical interest since our experiments hint to the fact that combining \fastlift and \slowlift typically yields better running times than \fastlift on its own using an additional shearing.\\

\paragraph{Experiments} We implemented \ca in a topic branch of  
\cgal. Our implementation uses the combinatorial framework of the existing bivariate algebraic kernel (\textsc{Ak\_2} for short) which is based on the algorithms from~\cite{eigenwilligk08,eigenwilligkw07}.
Intensive benchmarks~\cite{eigenwilligkw07,LuisPhD2010} have shown that \textsc{Ak\_2} can be considered as the current reference implementation. In our experiments, we run \textsc{Ak\_2} against our new implementation on numerous challenging benchmark instances;
we also outsourced all resultant and gcd computations within \textsc{Ak\_2} to the GPU which allows a better comparison of both implementations.
Our experiments show that \ca outperforms \textsc{Ak\_2} for all instances. More 
precisely, our method is, on average, twice as fast for easy instances such as non-singular curves in generic position, whereas, for hard instances, 
we typically improve by large factors between $5$ and $50$. The latter is mainly due to the new symbolic-numeric filter \fastlift, the exclusive use of resultant and $\gcd$ computations as the only symbolic operations, and the abdication of shearing. Computing arrangements mainly benefit from the improved curve-analyses, the improved bivariate solver (see below), and from avoiding
subresultants and coordinate transformations for harder instances.

We also compared the bivariate solver \bs with two currently state-of-the-art 
implementations, that is, \isolate (based on \rs by Fabrice 
Rouillier with ideas from~\cite{r-rur}) and 
\lgp by Xiao-Shan~Gao~\etal~\cite{LGP-09}, both interfaced in Maple~14.
Again, our experiments show that our method is efficient as it 
outperforms both contestants for most instances. More 
precisely, it is comparable for all considered instances and
typically between $5$ and $10$-times faster.

\medskip

From our experiments, we conclude that the considerable gain in performance 
of \bs and \ca is due to the following reasons: Since our algorithms only use resultant and $\gcd$ computations as purely symbolic operations they beat by design other approaches that use more involved algebraic techniques.
As both symbolic computations are outsourced to the GPU, we even see 
tremendously reduced cost, eliminating a (previously) typical bottleneck. 
Moreover, our filters apply to many
input systems and, thus, allow a more adaptive treatment of algebraic curves. 
Our initial decision to avoid any coordinate transformation has turned out to be favorable, in particular, for sparse input and for computing arrangements.
In summary, from our experiments, we conclude that instances which have so far been considered to be difficult, such as singular curves or curves in non-generic position, can be handled at least as fast as seemingly easy instances such as randomly chosen, non-singular curves of the same input size.

We would like to remark that preliminary versions of this work have already been presented at 
ALENEX~2011~\cite{bes-bisolve-2011} and SNC~2011~\cite{beks:snc:2011}. A recent result~\cite{es-bisolvecomplexity-11} on the complexity of \bs further shows that it is also very efficient in theory, that is, the bound on its worst case bit complexity is by several magnitudes lower than the best bound known so far for this problem. 
In comparison to the above mentioned conference papers, this journal version comes along with a series of improvements: First, we consider a new filter for \bs which is 
based on a certified numerical complex root solver. It allows us to certify 
a box to be isolating for a solution $(\alpha,\beta)\in\RR^2$ in a generic situation, 
where no further solution with the same $x$-coordinate exists. Second, the test within \ca to decide in advance whether \fastlift applies, and the proof that \fastlift applies to any curve in a generic position have not been presented before. The latter two results yield a novel complete and certified method \fastca (i.e.~\ca with \fastlift only, where \slowlift is disabled) to compute the \emph{topology} of an algebraic curve.\\

\paragraph{Outline} The bivariate solver \bs is 
discussed in Section~\ref{sec:bs}. In Section~\ref{sec:ca}, we introduce \ca to analyze a single algebraic curve. The latter section particularly features two parts, that is, the presentation of a complete method \slowlift in Section~\ref{sssec:ca:alg:lift:slowlift} that is based on \bs, and the presentation of the symbolic-numeric method \fastlift in Section~\ref{sssec:ca:lift:fastlift}. \fastlift uses a numerical solver whose details are given in \ref{asec:numerical}. \bs and \ca 
are finally utilized in Section~\ref{sec:arr} in order to enable the 
computation of arrangements of algebraic curves.
The presented algorithms allow speedups, among other things, due to the use of graphics 
hardware for symbolic operations as described in
Section~\ref{sec:speedups}. Our algorithms are prototypically implemented
in the \cgal\ project. Section~\ref{sec:implex} gives necessary details
and also features many experiments that show the performance
of the new approach. We conclude in Section~\ref{sec:conclusion} and outline
further directions of research.\newpage

\section{\bs: Solving a Bivariate System}
\label{sec:bs}

The \emph{input} of our algorithm is the following polynomial system

\begin{equation}
f(x,y)=\sum_{i,j\in\N:i+j\le m} f_{ij}x^iy^j=0\text{\hspace{0.1cm} and\hspace{0.1cm} }g(x,y)=\sum_{i,j\in\N:i+j\le n} g_{ij}x^iy^j=0,\label{system}
\end{equation}
where $f$, $g\in\Z[x,y]$ are polynomials of total degrees $m$ and $n$, 
respectively. It is assumed that $f$ and 
$g$ have no common factors; otherwise, $f$ and $g$ have to be decomposed into common and non-common factors first, and then the finite-dimensional solution set has to be merged with the one-dimensional part defined by the common factor (not part of our algorithm). 
Hence, the set
$V_{\C}:=\{(x,y)\in\C^2|f(x,y)=g(x,y)=0\}$ of
(complex) solutions of~(\ref{system})
is zero-dimensional and consists, by B\'{e}zout's theorem, of at most
$m\cdot n$ distinct elements.

Our algorithm \emph{outputs} disjoint boxes
$B_k\subset\RR^2$ such that the union of all $B_k$ contains all \emph{real}
solutions
$$
V_{\RR}:=\{(x,y)\in\RR^2|f(x,y)=g(x,y)=0\}=V_{\C}\cap\RR^{2}
$$
of (\ref{system}) and each $B_k$ is \emph{isolating}, that is, it contains
exactly one solution.\\

\paragraph{Notation} We also write
\begin{equation*}
f(x,y)=\sum_{i=0}^{m_x} f_{i}^{(x)}(y) x^i=\sum_{i=0}^{m_y} f_{i}^{(y)}(x)
y^i\text{\hspace{0.1cm} and\hspace{0.1cm} }
g(x,y)=\sum_{i=0}^{n_x} g_{i}^{(x)}(y)x^i=\sum_{i=0}^{n_y}
g_{i}^{(y)}(x)y^i,
\end{equation*}
where $f_{i}^{(y)}$, $g_{i}^{(y)}\in\Z[x]$, $f_{i}^{(x)}$, 
$g_{i}^{(x)}\in\Z[y]$ and $m_x$, $n_x$ and $m_y$, 
$n_y$ denote the degrees of $f$ and $g$ considered as polynomials in $x$
and
$y$, respectively. For an interval 
$I=(a,b)\subset\RR$, $m_{I}:=(a+b)/2$ denotes the \emph{center} and 
$r_{I}:=(b-a)/2$ the \emph{radius} of $I$. For an arbitrary 
$m\in\C$ and $r\in\RR^{+}$, $\Delta_{r}(m)$ denotes the disc with 
center $m$ and radius $r$.

\paragraph{Resultants}
Our method is based on well known elimination techniques. 
We consider the projections
\begin{eqnarray*} 
V^{(x)}_{\C}&:=&\{x\in\C|\exists y\in\C\text{ with }f(x,y)=g(x,y)=0\},\\
V^{(y)}_{\C}&:=&\{y\in\C|\exists x\in\C\text{ with }f(x,y)=g(x,y)=0\}
\end{eqnarray*}
of all complex solutions $V_{\C}$ onto the $x$- and $y$-coordinate.
Resultant computation is a well studied tool to obtain an algebraic
description of these projection sets, that is, polynomials whose roots are
exactly the projections of the solution set $V_{\C}$. The resultant
$R^{(y)}=\res(f,g;y)\in\Z[x]$ of $f$ and $g$ with respect to the
variable $y$ is the determinant of the $(m_y+n_y)\times(m_y+n_y)$
\emph{Sylvester matrix}:
\[
S^{(y)}(f,g):=\left[
\begin{array}{ccccccc}
f_{m_y}^{(y)} & f_{m_y-1}^{(y)} & \ldots & f_0^{(y)} & 0\ \ldots & 0 \\
\vdots & \ddots & \ddots & & \ddots &  \vdots \\
0 & \ldots\ \ 0 & f_{m_y}^{(y)} & f_{m_y-1}^{(y)} & \ldots & f_0^{(y)} \\
g_{n_y}^{(y)} & g_{n_y-1}^{(y)} & \ldots & g_0^{(y)} & 0\  \ldots & 0 \\
\vdots & \ddots & \ddots & & \ddots &  \vdots \\
0 & \ldots\ \ 0 &  g_{n_y}^{(y)} & g_{n_y-1}^{(y)} & \ldots & g_0^{(y)}
\end{array}\right]
\]

\noindent From the definition, it follows that $R^{(y)}(x)$ is a polynomial in $x$ of degree
less than or equal to $m\cdot n$. The resultant
$R^{(x)}=\res(f,g;x)\in\Z[y]$ of $f$ and $g$ with respect to $x$ is
defined in completely analogous manner by considering $f$ and $g$ as
polynomials in $x$ instead of $y$.
As mentioned above, the resultant polynomials have the following important
property (see~\cite{bpr-algorithms} for a proof):
\begin{theorem}\label{thm:resultants}
The roots of $R^{(y)}(x)$ are exactly the projections of the solutions of (\ref{system}) onto the $x$-coordinate and the roots of the greatest common divisor $h^{(y)}(x):=\gcd(f_{m_y}(x),g_{n_y}(x))$ of the leading coefficients of $f$ and~$g$. 
More precisely, 
\begin{align*}
\{x\in\C|R^{(y)}(x)=0\}=V^{(x)}_{\C}\cup \{x\in\C|h^{(y)}(x)=0\}
\end{align*}
For $R^{(x)}(y)$, a corresponding result holds:
\begin{align*}
\{y\in\C|R^{(x)}(y)=0\}=V^{(y)}_{\C}\cup \{y\in\C|h^{(x)}(y)=0\},
\end{align*}
where $h^{(x)}(y):=\gcd(f_{m_x}(y),g_{n_x}(y))$.
The multiplicity of a root $\alpha$ of $R^{(y)}$ ($R^{(x)}$) is the sum\footnote{For a root $\alpha$ of $h^{(y)}(x)$ (or $h^{(x)}(y)$), the intersection multiplicity of $f$ and $g$ at the ``infinite point'' $(\alpha,\infty)$ (or $(\infty,\alpha)$) has also been taken into account. For simplicity, we decided not to consider the more general projective setting.} of
the intersection multiplicities\footnote{The multiplicity of a solution $(x_0,y_0)$ of ($\ref{system}$) is defined as the dimension of the localization of 
$\C[x,y]/(f,g)$ at $(x_0,y_0)$ considered as $\C$-vector space  (cf.~\cite[p.148]{bpr-algorithms})} of all solutions of
(\ref{system}) with $x$-coordinate ($y$-coordinate) $\alpha$.
\end{theorem}

\paragraph{Overview of the Algorithm}

We start with the following high level description of the proposed
algorithm which decomposes into three subroutines:
In the first phase (\bsproject, see Section~\ref{ssec:bs:biproject}), 
we project the complex solutions
$V_{\C}$ of (\ref{system})
onto the $x$- and onto the $y$-axis. More precisely, we compute the restrictions $
V^{(x)}_{\RR}:=V^{(x)}_{\C}\cap \RR$ and $V^{(y)}_{\RR}:=V^{(y)}_{\C}\cap \RR
$
of the complex projection sets $V^{(x)}_{\C}$ and $V^{(y)}_{\C}$ to
the real axes and isolating intervals for their elements.
Obviously, the real solutions $V_{\RR}$ are contained in the cross product $\mathcal{C}:=V^{(x)}_{\RR}\times 
V^{(y)}_{\RR}\subset\RR^2$. 
In the second phase (\bsseparate, see Section~\ref{ssec:bs:separate}), we compute isolating discs
which "well separate" the projected solutions from each other. The latter step prepares the third phase (\bsvalidate, see Section~\ref{ssec:bs:validate}) in which 
candidates of $\mathcal{C}$ are either discarded or certified 
to be a solution of (\ref{system}). 
Our \emph{main theoretical contribution} is the introduction of a 
novel predicate to ensure that a certain candidate 
$(\alpha,\beta)\in\mathcal{C}\cap V_{\RR}$ actually fulfills 
$f(\alpha,\beta)=g(\alpha,\beta)=0$ (cf.~Theorem~\ref{thm:inclusion}).
For candidates $(\alpha,\beta)\in\mathcal{C}\backslash V_{\RR}$,
interval arithmetic suffices to exclude $(\alpha,\beta)$ as
a solution of (\ref{system}).

We remark that,
in order to increase the efficiency of our implementation,
we also introduce additional filtering techniques to eliminate
many of the candidates in $\mathcal{C}$.
However, for the sake of clarity, we refrain from integrating our
filtering techniques into the following description
of the three subroutines.
Section~\ref{ssec:speedups:gpu_res} briefly discusses a
highly parallel algorithm on the graphics hardware
to accelerate computations of the resultants the $\gcd$s needed in the first step,
while the filtering techniques for \bsvalidate are covered in 
Section~\ref{ssec:speedups:bsfilters}.\\

\subsection{\bsproject}
\label{ssec:bs:biproject}

We compute the resultant
$R:=R^{(y)}=\res(f,g;y)
\in\Z[x]$ and a square-free
factorization of~$R$. More precisely, we determine square-free and pairwise 
coprime factors $r_{i}\in\Z[x]$, $i=1,\ldots,\deg(R)$, such that
$R(x)=\prod_{i=1}^{\deg
(R)}\left(r_{i}(x)\right)^{i}$. We remark that, for some
$i\in\{1,\ldots,\deg(R)\}$, $r_{i}(x)=1$.
Yun's algorithm~\cite[Alg.~14.21]{gathen} constructs such a square-free factorization by
essentially computing greatest common divisors of $R$ and its higher
derivatives in an iterative way.
Next, we isolate the real roots $\alpha_{i,j}$, $j=1,\ldots,\ell_{i}$, of
the polynomials $r_{i}$. That is, we determine disjoint isolating intervals
$I(\alpha_{i,j})\subset\RR$ such that each interval
$I(\alpha_{i,j})$ contains exactly one root (namely, $\alpha_{i,j}$) of
$r_{i}$, and the union
of all $I(\alpha_{i,j})$, $j=1,\ldots,\ell_{i}$, covers all real roots of
$r_{i}$.
For the real root isolation, we consider the Descartes method~\cite{vca,RS} as a suited
algorithm.
From the square-free factorization we know that $\alpha_{i,j},
j=1,\ldots,\ell_{i}$, is a root of $R$ with multiplicity~$i$.\\ 

\subsection{\bsseparate}
\label{ssec:bs:separate}

We separate the real roots
of $R=R^{(y)}$ from all other (complex) roots of $R$, an operation which is
crucial for the final validation. More precisely, let
$\alpha=\alpha_{i_{0},j_{0}}$ be the $j_{0}$-th real root of the polynomial
$r_{i_{0}}$, where $i_{0}\in\{1,\ldots,\deg(R)\}$ and
$j_{0}\in\{1,\ldots,\ell_{i_{0}}\}$ are arbitrary indices. We refine the
corresponding isolating interval $I=(a,b):=I(\alpha)$ such that the disc
$\Delta_{8r_{I}}(m_I)$ does not contain any root of $R$
except~$\alpha$.
For the refinement of $I$, we use quadratic interval 
refinement (QIR for short)~\cite{abbott-qir-06,qir-kerber-11} which constitutes a highly efficient 
method because of its simple tests and the fact that it eventually achieves
quadratic convergence.

In order to test whether the disc $\Delta_{8r_{I}}(m_I)$ isolates $\alpha$
from all other roots of $R$, we consider an approach which was also used in~\cite{sy-ceval}. It is based on the following test:

\begin{equation*}
    T^p_K(m,r):|p(m)| - K \sum_{k\ge 1}
        \left|\frac{p^{(k)}(m)}{k!} \right|r^k>0,
\end{equation*}
where $p\in\RR[x]$ denotes an arbitrary polynomial and $m$, $r$, $K$
arbitrary real values. Then, the following theorem holds:\footnote{For a similar result, the reader may also consider~\cite{skh-ccri-2009}, where a corresponding test based on interval arithmetic only has been introduced.}

\begin{theorem}
\label{thm:test}
Consider a disk $\Delta=\Delta_m(r)\subset\C$ with center $m$ and radius
$r$.
\begin{enumerate}
\item If $T_{K}^p(m,r)$ holds for some $K\geq 1$, then the closure
$\overline{\Delta}$ of $\Delta$ contains no root of $p$.
\item   If $T^{p'}_{K}(m,r)$ holds for a $K\geq \sqrt{2}$, then
$\overline{\Delta}$ contains at most one root of $p$.
\end{enumerate}
\end{theorem}

\begin{proof}
(1) follows from a straight-forward computation: For each
$z\in\overline{\Delta}$, we have $$p(z)=p(m+(z-m))=p(m)+\sum_{k\ge 1}
        \frac{p^{(k)}(m)}{k!}(z-m)^k,$$ and thus
        $$\frac{|p(z)|}{|p(m)|}\geq 1-\frac{1}{|p(m)|}\cdot\sum_{k\ge 1}
        \frac{|p^{(k)}(m)|}{k!}|z-m|^k>\left(1-\frac{1}{K}\right)$$
since $|z-m|\le r$ and $T_{K}^p(m,r)$ holds. In particular, for $K\geq 1$,
the above inequality implies $|p(z)|>0$ and, thus, $p$ has no root in
$\overline{\Delta}$.

It remains to show (2): If $T_K^{p'}(m,r)$ holds, then, for any point
$z\in\overline{\Delta}$, the derivative $p'(z)$ differs from $p'(m)$ by a
complex number of absolute value less than $|p'(m)|/K$. Consider the
triangle spanned by the points $0$, $p'(m)$ and $p'(z)$, and let $\alpha$
and $\beta$ denote the angles at the points $0$ and $p'(z)$, respectively.
From the Sine Theorem, it follows that $$|\sin\alpha|=|p'(m)-p'(z)|\cdot
\frac{|\sin \gamma|}{|p'(m)|}< \frac{1}{K}.$$ Thus, the arguments of
$p'(m)$ and $p'(z)$ differ by less than $\operatorname*{arcsin}(1/K)$ which
is smaller than or equal to $\pi/4$ for $K\geq\sqrt{2}$.
Assume that there exist two roots $a, b\in\overline{\Delta}$ of $p$. Since
$a=b$ implies $p'(a)=0$, which is not possible as $T^{p'}_1(m,r)$ holds, we
can assume that $a\neq b$.
We split $p$ into its real and imaginary part, that is, we consider
$p(x+iy)=u(x,y)+iv(x,y)$ where $u,v:\RR^2\rightarrow\RR$ are two bivariate
polynomials. Then, $p(a)=p(b)=0$ and so $u(a)=v(a)=u(b)=v(b)=0$.
But $u(a)=u(b)=0$ implies, due to the Mean Value Theorem in several real
variables, that there exists a $\phi\in [a,b]$ such that 
  $$\nabla u(\phi)\perp (b-a).$$
Similarly, $v(a)=v(b)=0$ implies that there exists a $\xi\in [a,b]$
such that $\nabla v(\xi)\perp (b-a).$ But $\nabla
v(\xi)=(v_x(\xi),v_y(\xi))=(-u_y(\xi), u_x(\xi))$, thus,
it follows that $\nabla u(\xi) ~\|~ (b-a).$
Therefore, $\nabla u(\psi)$ and $\nabla u(\xi)$ must be perpendicular.
Since $p'=u_x+iv_x=u_x-iu_y$, the arguments of $p'(\psi)$ and $p'(\xi)$
must differ by $\pi/2$. This contradicts our above result that both differ
from the argument of $p'(m)$ by less than $\pi/4$, thus, (2) follows.
\end{proof}

Theorem~\ref{thm:test} now directly applies to the above scenario, where $p=r_{i_0}$ and $r=8r_I$. More
precisely, $I$ is refined until $T^{(r_{i_{0}})'}_{3/2}(m_I,8r_I)$ and
$T^{r_{i}}_1(m_I,8r_I)$ holds for all $i\neq i_0$.
If the latter two conditions are fulfilled, $\Delta_{8r_I}(m_I)$ isolates $\alpha$ from all
other roots of $R$. In this situation, we obtain a lower bound $L(\alpha)$ for $|R(z)|$ on the boundary of
$\Delta(\alpha):=\Delta_{2r_{I}}(m_{I})$:

\begin{lemma}\label{lem:lowerbound}
Let $I$ be an interval which contains a root $\alpha$ of $r_{i_0}$. If $T^{(r_{i_{0}})'}_{3/2}(m_I,8r_I)$ and
$T^{r_{i}}_1(m_I,8r_I)$ holds for all $i\neq i_0$, then the disc $\Delta(\alpha)=\Delta_{2r_{I}}(m_{I})$ isolates $\alpha$ from all
other (complex) roots of $R$ and, for any $z$ on the boundary $\partial
\Delta(\alpha)$ of $\Delta(\alpha)$, it holds that $$|R(z)|>L(\alpha):= 2^{-i_0-\deg(R)}|R(m_I-2r_I)|.$$
\end{lemma}
 
\begin{proof}
$\Delta(\alpha)$ is isolating as already $\Delta_{8r_{I}}(m_{I})$ is
isolating. Then, let $\beta\neq \alpha$ be an arbitrary root of $R$ and
$d:=|\beta-m_I|>8r_I$ the distance between $\beta$ and $m_I$. Then, for any
point $z\in\partial\Delta(\alpha)$, it holds that
\[\frac{|z-\beta|}{|(m_I-2r_I)-\beta|} >\frac{d-2r_I}{d+2r_I}=1-\frac{4r_I}{
d+2r_I}>\frac{1}{2}\hspace{0.2cm}\text{ and }\hspace{0.2cm}\frac{|z-\alpha|}{|(m_I-2r_I)-\alpha|} >\frac{r_I}{3r_I}>\frac{1}{4}.\]
Hence, it follows that
\begin{eqnarray*}
  \frac{|R(z)|}{|R(m_I-2r_I)|} & > & \left(\frac{|z-\alpha|}{|(m_I-2r_I)-\alpha|}\right)^{i_{0}} \cdot\prod_{\beta\neq\alpha:\;R(\beta)=0}\frac{|z-\beta|}{|(m_I-2r_I)-\beta|} > 4^{-i_{0}} 2^{-\deg(R)+i_{0}},
\end{eqnarray*}
where each root $\beta$ occurs as many times in the product as its
multiplicity as a root of $R$.
\end{proof}

We compute $L(\alpha)=2^{-i_{0}-\deg(R)}|R(m_I-2r_I)|$
and store the interval $I(\alpha)$, the disc $\Delta(\alpha)$, and the lower
bound $L(\alpha)$ for $|R(z)|$ on the boundary $\partial\Delta(\alpha)$ of
$\Delta(\alpha)$.\\

Proceeding in exactly the same manner for each real
root $\alpha$ of $R^{(y)}$, we get an isolating interval $I(\alpha)$,
an isolating disc $\Delta(\alpha)=\Delta_{2r_{I}}(m_{I})$, and a lower bound
$L(\alpha)$ for $|R^{(y)}|$ on $\partial\Delta(\alpha)$.
For the resultant polynomial $R^{(x)}=\res(f,g;x)$, \bsproject and
\bsseparate are processed in exactly the same manner: We compute
$R^{(x)}$ and a corresponding square-free factorization. Then, for each
real root $\beta$ of $R^{(x)}$, we compute a corresponding isolating
interval $I(\beta)$, a disc $\Delta(\beta)$ and a lower bound $L(\beta)$
for $|R^{(x)}|$ on $\partial\Delta(\beta)$.\\

\subsection{\bsvalidate}
\label{ssec:bs:validate}

We start with the following theorem:
\begin{theorem}\label{thm:boxproperties}
Let $\alpha$ and $\beta$ be arbitrary real roots of $R^{(y)}$ and
$R^{(x)}$, respectively. Then,
\begin{enumerate}
\item the polydisc $\Delta(\alpha,\beta):=\Delta(\alpha)
\times\Delta(\beta)\subset\C^2$ contains at most one solution of
(\ref{system}). If $\Delta(\alpha,\beta)$ contains a 
solution of (\ref{system}), then this solution is real valued and equals
$(\alpha,\beta)$.
\item For an arbitrary point $(z_1,z_2)\in\C^2$ on the boundary of
$\Delta(\alpha,\beta)$, it holds that
\begin{align*}
|R^{(y)}(z_1)|>L(\alpha)\text{ if }z_1\in\partial\Delta(\alpha)\text{, and }
|R^{(x)}(z_2)|>L(\beta) \text{ if }z_2\in\partial\Delta(\beta).
\end{align*}
\end{enumerate}
\end{theorem}

\begin{proof}
(1) is an easy consequence from the construction of the discs
$\Delta(\alpha)$ and $\Delta(\beta)$. Namely, if $\Delta(\alpha,\beta)$
contains two distinct solutions of (\ref{system}), then they would differ
in at least one coordinate. Thus, one of the discs $\Delta(\alpha)$ or
$\Delta(\beta)$ would contain two roots of $R^{(y)}$ or $R^{(x)}$. Since
both discs are isolating for a root of the corresponding resultant
polynomial, it follows that $\Delta(\alpha,\beta)$ contains at most one
solution. In the case, where $\Delta(\alpha,\beta)$ contains a solution of
(\ref{system}), this solution must be real since, otherwise,
$\Delta(\alpha,\beta)$ would also contain a corresponding complex conjugate
solution ($f$ and $g$ have real valued coefficients).
(2) follows directly from the definition of $\Delta(\alpha,\beta)$, the
definition of $L(\alpha)$, $L(\beta)$ and Lemma~\ref{lem:lowerbound}.
\end{proof}

We denote $B(\alpha,\beta)=I(\alpha)\times I(\beta)\subset\RR^2$ a
\emph {candidate box} for a real solution of (\ref{system}), where $\alpha$
and $\beta$ are real roots of $R^{(y)}$ and $R^{(x)}$, respectively. Due to
Theorem~\ref{thm:boxproperties}, the corresponding ``container polydisc''
$\Delta(\alpha,\beta)\subset \C^2$ either contains no solution of
(\ref{system}), or $(\alpha,\beta)$ is the only solution contained in
$\Delta(\alpha,\beta)$. Hence, for each candidate pair
$(\alpha,\beta)\in\mathcal{C}$, it suffices to show that either
$(\alpha,\beta)$ is no solution of (\ref{system}), or the corresponding
polydisc $\Delta(\alpha,\beta)$ contains at least one solution.
In the following steps, we fix the polydiscs $\Delta(\alpha,\beta)$, whereas
the boxes $B(\alpha,\beta)$ are further refined (by further refining the
isolating intervals $I(\alpha)$ and $I(\beta)$). We further introduce
exclusion and inclusion predicates such that, for sufficiently small
$B(\alpha,\beta)$, either $(\alpha,\beta)$ can be discarded or certified as
a solution of (\ref{system}).\\

\label{pref:interval_exclusion}
In order to \emph{exclude} a candidate box, we use simple interval arithmetic.
More precisely, we evaluate $\Box f(B(\alpha,\beta))$ and $\Box
g(B(\alpha,\beta))$, where $\Box f$ and $\Box g$ constitute box functions
for $f$ and $g$, respectively: If either $\Box
f(B(\alpha,\beta))$ or $\Box g(B(\alpha,\beta))$ does not contain zero,
then $(\alpha,\beta)$ cannot be a solution of (\ref{system}). Vice versa,
if $(\alpha,\beta)$ is not a solution and $B(\alpha,\beta)$ becomes
sufficiently small, then either $0\notin\Box f(B(\alpha,\beta))$ or
$0\notin\Box g(B(\alpha,\beta))$, and thus our exclusion predicate applies.

It remains to provide an \emph{inclusion predicate}, that is, a method 
that approves that a certain candidate $(\alpha,\beta)\in\mathcal{C}$ is
actually a solution of (\ref{system}). We first rewrite the resultant polynomial $R^{(y)}$ as
\begin{align*}
R^{(y)}(x)=u^{(y)}(x,y)\cdot f(x,y)+v^{(y)}(x,y)\cdot g(x,y),
\end{align*}
where $u^{(y)}$, $v^{(y)}\in\Z[x,y]$ are cofactor polynomials which can be
expressed as determinants of corresponding ``Sylvester-like'' matrices:
\[\scriptsize
U^{(y)}=\left|
\begin{array}{ccccccc}
f_{m_y}^{(y)} & f_{m_y-1,y}^{(y)} & \ldots & f_{0}^{(y)} & 0\ \ldots &
y^{n_y-1} \\
\vdots & \ddots & \ddots & & \ddots &  \vdots \\
0 & \ldots\ \ 0 & f_{m_y}^{(y)} & f_{m_y-1}^{(y)} & \ldots & 1 \\
g_{n_y}^{(y)} & g_{n_y-1}^{(y)} & \ldots & g_{0}^{(y)} & 0\  \ldots & 0 \\
\vdots & \ddots & \ddots & & \ddots &  \vdots \\
0 & \ldots\ \ 0 &  g_{n_y}^{(y)} & g_{n_y-1}^{(y)} & \ldots & 0
\end{array}\right|,\hspace{0.25cm}
V^{(y)} =\left|
\begin{array}{ccccccc}
f_{m_y}^{(y)} & f_{m_y-1}^{(y)} & \ldots & f_{0}^{(y)} & 0\ \ldots & 0 \\
\vdots & \ddots & \ddots & & \ddots &  \vdots \\
0 & \ldots\ \ 0 & f_{m_y}^{(y)} & f_{m_y-1}^{(y)} & \ldots & 0 \\
g_{n_y}^{(y)} & g_{n_y-1}^{(y)} & \ldots & g_{0}^{(y)} & 0\  \ldots &
y^{m_y-1} \\
\vdots & \ddots & \ddots & & \ddots &  \vdots \\
0 & \ldots\ \ 0 &  g_{n_y}^{(y)} & g_{n_y-1}^{(y)} & \ldots & 1
\end{array}\right|
\]
The matrices $U^{(y)}$ and
$V^{(y)}$ are obtained from $S^{(y)}(f,g)$ by replacing the last
column with vectors $(y^{n_y-1}\dots 1\ 0 \dots 0)^T$ and $(0 \dots 0\
y^{m_y-1}\dots 1)^T$ of appropriate size, respectively~\cite[p.~287]{algs-92}. 
Both matrices have size
$(n_y+m_y)\times(n_y+m_y)$ and univariate polynomials in
$x$ (the first $n_{y}+m_{y}-1$ columns), or powers of $y$ (only the last
column), or zeros as entries. We now aim for upper bounds for $|u^{(y)}|$ and
$|v^{(y)}|$ on the polydisc $\Delta(\alpha,\beta)$. The polynomials $u^{(y)}$ 
and $v^{(y)}$ have huge coefficients and their computation, either via a 
signed remainder sequence or via determinant evaluation, is very costly. 
Hence, we directly derive 
such upper bounds from the corresponding matrix representations 
\textbf{without computing} $u^{(y)}$ and
$v^{(y)}$: Due to 
Hadamard's
bound, $|u^{(y)}|$ is smaller than the product of the $2$-norms of
the column vectors of $U^{(y)}$. The absolute value of each of the entries
of $U^{(y)}$ can be easily upper bounded by using interval arithmetic on a
box in $\C^2$ that contains the polydisc $\Delta(\alpha,\beta)$. Hence, we get
an upper bound on the $2-$norm of each column vector and, thus, an upper
bound $U(\alpha,\beta,u^{(y)})$ for $|u^{(y)}|$ on $\Delta(\alpha,\beta)$
by multiplying the bounds for the column vectors. In the same manner, we also
derive an upper bound $U(\alpha,\beta,v^{(y)})$ for $|v^{(y)}|$ on
$\Delta(\alpha,\beta)$. With respect to our second projection direction, we
write $R^{(x)}=u^{(x)}\cdot f+v^{(x)}\cdot g$
with corresponding polynomials $u^{(x)}$, $v^{(x)}\in\Z[x,y]$. In exactly
the same manner as done for $R^{(y)}$, we compute corresponding upper
bounds $U(\alpha,\beta,u^{(x)})$ and $U(\alpha,\beta,v^{(x)})$ for
$|u^{(x)}|$ and $|v^{(x)}|$ on $\Delta(\alpha,\beta)$, respectively. 

\begin{theorem}\label{thm:inclusion}
If there exists an $(x_0,y_0)\in\Delta(\alpha,\beta)$ with
\begin{equation}\begin{aligned}
U(\alpha,\beta,u^{(y)}) \cdot |f(x_0,y_0)| + U(\alpha,\beta,v^{(y)}) \cdot |g(x_0,y_0)|< L(\alpha)\label{ineqA}
\end{aligned}
\end{equation}
and
\begin{equation}
\begin{aligned}
U(\alpha,\beta,u^{(x)}) \cdot |f(x_0,y_0)|+U(\alpha,\beta,v^{(x)}) \cdot |g(x_0,y_0)|<  L(\beta)\label{ineqB},
\end{aligned}
\end{equation}
then $\Delta(\alpha,\beta)$ contains a solution of (\ref{system}), and thus $f(\alpha,\beta)=0$.
\end{theorem}

\begin{proof}
The proof uses a homotopy argument. Namely, we consider the parameterized
system
\begin{equation}
\begin{aligned}
f(x,y)-(1-t)\cdot f(x_0,y_0)  = g(x,y)-(1-t)\cdot g(x_0,y_0)  = 0,\label{systemt}
\end{aligned}
\end{equation}
where $t$ is an arbitrary real value in $[0,1]$. For $t=1$, (\ref{systemt})
is equivalent to our initial system (\ref{system}). For $t=0$,
(\ref{systemt}) has a solution in $\Delta(\alpha,\beta)$, namely,
$(x_0,y_0)$. The complex solutions of (\ref{systemt}) continuously depend
on the parameter $t$. Hence, there exists a ``solution path''
$\Gamma:[0,1]\mapsto\C^2$ which connects $\Gamma(0)=(x_0,y_0)$ with a
solution $\Gamma(1)\in\C^2$ of (\ref{system}). We show that $\Gamma(t)$
does not leave the polydisc $\Delta(\alpha,\beta)$ and, thus,
(\ref{system}) has a solution in $\Delta(\alpha,\beta)$: Assume that the path
$\Gamma(t)$ leaves the polydisc, then there exists a $t'\in[0,1]$ with
$(x',y')=\Gamma(t')\in\partial\Delta(\alpha,\beta)$. We assume that
$x'\in\partial\Delta(\alpha)$ (the case $y'\in\partial\Delta(\beta)$ is
treated in analogous manner). Since $(x',y')$ is a solution of
(\ref{systemt}) for $t=t'$, we must have $|f(x',y')|\le|f(x_0,y_0)|$ and
$|g(x',y')|\le|g(x_0,y_0)|$. Hence, it follows that
\begin{align*}
|R^{(y)}(x')| &= 
    |u^{(y)}(x',y') f(x',y')+v^{(y)}(x',y')g(x',y')|\\
    &\le |u^{(y)}(x',y')|\cdot |f(x',y')| + |v^{(y)}(x',y')|\cdot |g(x',y')|\\
&\le U(\alpha,\beta,u^{(y)})\cdot|f(x_0,y_0)|+ U(\alpha,\beta,v^{(y)})\cdot |g(x_0,y_0)|<L(\alpha).
\end{align*}
This contradicts the fact that $|R^{(y)}(x')|$ is lower bounded by
$L(\alpha)$. It follows that $\Delta(\alpha,\beta)$ contains a solution of
(\ref{system}) and, according to Theorem~\ref{thm:boxproperties}, this
solution must be $(\alpha,\beta)$.
\end{proof}

Theorem~\ref{thm:inclusion} now directly applies as an inclusion predicate.
Namely, in each refinement step of $B(\alpha,\beta)$, we choose an arbitrary
$(x_0,y_0)\in B(\alpha,\beta)$ (e.g.~the center
$(m_{I(\alpha)},m_{I(\beta)})$ of the candidate box $B(\alpha,\beta)$) and
check whether both inequalities (\ref{ineqA}) and (\ref{ineqB}) are
fulfilled. If $(\alpha,\beta)$ is a solution of (\ref{system}), then both
inequalities eventually hold and, thus, we have shown that $(\alpha,\beta)$
is a solution.\\

We want to remark that the upper bounds $U(\alpha,\beta,u^{(y)})$, 
$U(\alpha,\beta,v^{(y)})$, $U(\alpha,\beta,u^{(x)})$ and 
$U(\alpha,\beta,v^{(y)})$ are far from being optimal. Nevertheless, our 
inclusion predicate is still efficient since we can approximate the potential 
solution $(\alpha,\beta)$ with quadratic convergence due to the QIR method. Hence, the 
values $f(x_0,y_0)$ and $g(x_0,y_0)$ become very small after a few iterations. 
In order to improve the above upper bounds, we propose to consider more 
sophisticated methods from numerical analysis and matrix perturbation 
theory~\cite{IR-perturbation,rump-verifiedbounds}.
Finally, we would like to emphasize that our method applies particularly well 
to the situation where we are only interested in the solutions of (\ref{system}) 
within a given box $\mathcal{B}=[A,B]\times[C,D]\subset\RR^{2}$. 
Though $R^{(y)}$ ($R^{(x)}$) capture all (real and complex) projections of the solutions of the system, we only have to search for the real ones contained 
within the interval $[A,B]$ ($[C,D]$). Then, 
only candidate boxes within $\mathcal{B}$ have to be 
considered in \bsseparate and \bsvalidate. Hence, since the computation 
of the resultants is relatively cheap due to our fast implementation on the 
GPU (see Section~\ref{ssec:speedups:gpu_res}), 
our method is particularly well suited to search for local solutions.

\section{\ca: Analysing an Algebraic Curve}
\label{sec:ca}

\label{ssec:ca:algorithm} 

The input of \ca is a planar algebraic curve
$C$ as defined in (\ref{def:curve}), where $f\in\Z[x,y]$ is a \emph{square-free}, bivariate polynomial with
integer coefficients. If $f$ is considered as polynomial in $y$ with
coefficients $f_i(x)\in\Z[x]$, its coefficients typically share a \emph{trivial} content $h:=\gcd(f_0,f_1,\ldots)$, that is, $h\in\Z$.
A non-trivial content $h\in\Z[x]\backslash\Z$ defines vertical lines at the real roots of $h$. Our algorithm handles this situation by dividing out $h$ first and finally merging the vertical lines defined by $h=0$ and the analysis of the curve $C':=V(f/h)$ at the end of the algorithm; see \cite{kerber-phd} for details. Hence, throughout the following considerations, we can assume that $h$ is trivial, thus $C$ contains no vertical line.

The algorithm returns a planar graph $\mathcal{G}_C$
that is isotopic to $C$,
where the set $V$ of all vertices of $\mathcal{G}_C$ is located on $C$. From a high-level
perspective our algorithm follows a classical cylindrical algebraic decomposition approach
consisting of three phases that we overview next:

\paragraph{Overview of the Algorithm} In the first phase (\caproject, see Section~\ref{ssec:ca:alg:project}), we project all
\emph{$x$-critical points} $(\alpha,\beta)\in C$ (i.e.~$f(\alpha,\beta)=f_y(\alpha,\beta)=0$) onto the $x$-axis by means of a
resultant computation and root isolation for the elimination polynomial.
The set of $x$-critical points comprises exactly the points where $C$ has a
vertical tangent or is singular. It is well known (e.g.~see~\cite[Theorem 2.2.10]{kerber-phd} for a short proof) that, for any two
consecutive real $x$-critical values $\alpha$ and $\alpha'$, $C$ is
\emph{delineable} over $I=(\alpha,\alpha')$, that is, $C|_{I\times\RR}$
decomposes into a certain number $m_I$ of disjoint function graphs
$C_{I,1},\ldots,C_{I,m_I}$. In the second phase (\calift, see Section~\ref{ssec:ca:alg:lift}), we first isolate the roots of the (square-free)
\emph{intermediate polynomial} $f(q_I,y)\in\mathbb{Q}[y]$, where $q_I$
constitutes an arbitrary chosen but fixed rational value in $I$. This computation yields the number $m_I$ ($=$ number of real roots of $f(q_I,y)$) of arcs above $I$ and corresponding representatives
$(q_I,y_{I,i})\in C_{I,i}$ on each arc. We further compute all points on $C$ that are
located above an $x$-critical value $\alpha$, that is, we determine the
real roots $y_{\alpha,1},\ldots,y_{\alpha,m_{\alpha}}$ of each (non
square-free) \emph{fiber polynomial} $f(\alpha,y)\in\RR[y]$. For this task, we propose two different novel methods, and we show that both of them can be combined in a way to improve the overall efficiency.
From the latter computations we obtain the vertex set $V$ of
$\mathcal{G}_C$ as the union of all points $(q_I,y_{I,i})$ and
$(\alpha,y_{\alpha,i})$. In the third and final phase (\caconnect, see Section~\ref{ssec:ca:alg:connect}),
which concludes the geometric-topological analysis, we determine which of the above
vertices are connected via an arc of $C$. For each connected pair
$(v_1,v_2)\in V$, we insert a line segment connecting $v_1$ and $v_2$. It
is then straight-forward to prove that $\mathcal{G}_C$ is isotopic
to $C$; see also~\cite[Theorem 6.4.4]{kerber-phd}. We remark that we never consider any kind of coordinate
transformation, even in the case where $C$ contains two or more $x$-critical
points sharing the same $x$-coordinate.\\

\begin{figure*}[t]
  \centering
  \includegraphics[width=\linewidth]{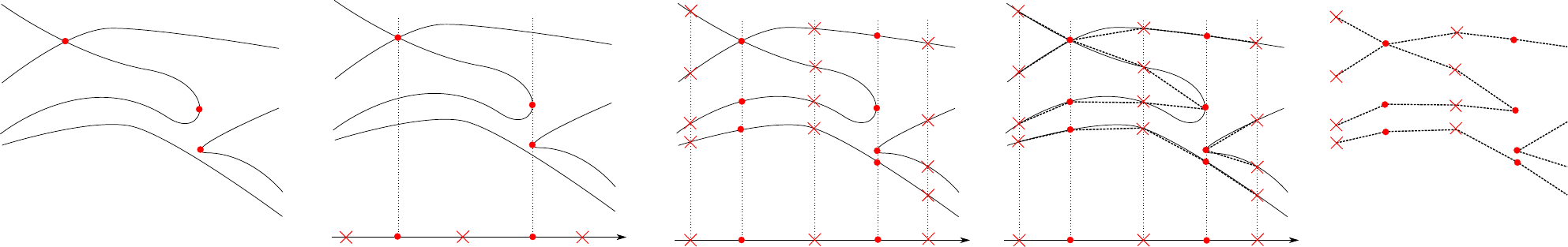}
\caption{\label{fig:cad}The figure on the left shows a curve $C$ with two
$x$-extremal points and one singular point (red dots). In the
\emph{projection phase}, these points are projected onto the $x$-axis and rational points separating the $x$-critical values are
inserted (red crosses). In the \emph{lifting phase}, the fibers at the
critical values (red dots) and at the points in between (red
crosses) are computed. In the \emph{connection phase}, each pair of points connected by an arc of $C$ is determined, and a corresponding line
segment is inserted. Finally, we obtain a graph that is
isotopic to $C$.\protect\\[-2em]\strut}
\end{figure*}

\subsection{\caproject}
\label{ssec:ca:alg:project}

We follow a similar approach as in \bsproject, that is, we compute the resultant $R(x):=\res(f,f_y;y) \in\Z[x]$
and a square-free factorization of $R$. In other words, we first determine
square-free and pairwise coprime factors\footnote{Either by square-free factorization, or
full factorization} $r_{i}\in\Z[x]$,
$i=1,\ldots,\deg(R)$, such that $R(x)=\prod_{i=1}^{\deg
(R)}\left(r_{i}(x)\right)^{i}$, and then isolate the real roots
$\alpha_{i,j}$, $j=1,\ldots,\ell_{i}$, of the polynomials $r_{i}$ which in
turn are $i$-fold roots of $R$. The so-obtained isolating intervals have rational endpoints, and we denote $I(\alpha_{i,j})\subset\RR$ the interval which contains $\alpha_{i,j}$ but no other root of $r_{i}$. Similar as in \bs, we further refine the intervals $I(\alpha_{i,j})$, $i=1,\ldots,\deg(R)$ and $j=1,\ldots,\ell_i$, such that all of them are pairwise disjoint. Then, for each pair $\alpha$
and $\alpha'$ of consecutive roots of $R$
defining an open interval $I = (\alpha,\alpha')$, we choose a separating
rational value $q_I$ in between the corresponding isolating intervals.

\subsection{\calift}
\label{ssec:ca:alg:lift}

Isolating the roots of the intermediate polynomials $f(q_I,y)$ is straight-forward because each $f(q_I,y)$ is a square-free polynomial with rational
coefficients, and thus the Descartes method directly applies. 

Determining the
roots of $f_{\alpha}(y):=f(\alpha,y)\in\mathbb{R}[y]$ at an $x$-critical value $\alpha$ is
considerably more complicated because $f_{\alpha}$ has multiple roots and, in general,
irrational coefficients.
One of the main contributions of this paper is to provide novel methods to compute the fiber at an $x$-critical value $x=\alpha$. More precisely, we first present a \emph{complete and certified} method \slowlift which is based on \bs (taken from Section~\ref{sec:bs}). It applies to any input curve (without assuming generic position) and any corresponding $x$-critical value; see Section~\ref{sssec:ca:alg:lift:slowlift}.
In Section~\ref{sssec:ca:lift:fastlift}, we further present a \emph{certified} symbolic-numeric method denoted \fastlift. Compared to \slowlift, it shows better efficiency in practice, but it may fail for a few fibers if the input curve is in a special geometric situation.
We further provide a method in order to easily check in advance whether \fastlift will succeed, and we also prove that this can always be achieved by means of a random coordinate transformation.
As already mentioned in the introduction, we aim to avoid such a transformation for efficiency reasons. Hence, we propose to combine both lifting methods in way such that \fastlift runs by default, and, only in case of its failure, we fall back to \slowlift.

\subsubsection{\slowlift{} --- a complete method for fiber computation}
\label{sssec:ca:alg:lift:slowlift}

\slowlift is based on the algorithm \bs to isolate the real solutions of a system of two bivariate polynomials $f,g\in\mathbb{Z}[x,y]$. Recall that \bs returns a set
of disjoint boxes $B_1,\ldots,B_m\subset\RR^2$ such that each box $B_i$
contains exactly one real solution $\xi=(x_0,y_0)$ of $f(x,y)=g(x,y)=0$,
and the union of all $B_i$ covers all solutions. Furthermore, for each
solution $\xi$, \bs provides square-free polynomials
$p,q\in\Z[x]$ with $p(x_0)=q(y_0)=0$ and corresponding isolating (and
refineable) intervals $I(x_0)$ and $I(y_0)$ for $x_0$ and $y_0$,
respectively. Comparing $\xi$ with another point $\tilde{\xi}=(x_1,y_1)\in\RR^2$
given by a similar representation is rather straight-forward. Namely, let
$\tilde{p},\tilde{q}\in\Z[x]$ be corresponding defining square-free
polynomials and $I(x_1)$ and $I(y_1)$ isolating intervals for $x_1$ and
$y_1$, respectively, then we can compare the $x$- and $y$-coordinates of the points
$\xi$ and $\tilde{\xi}$ via $\gcd$-computation of the defining univariate
polynomials and sign evaluation at the endpoints of the isolating intervals
(see~\cite[Algorithm 10.44]{bpr-algorithms} for more details).\\ 

In order to compute the fiber at a specific real $x$-critical value $\alpha$ of $C$, we
proceed as follows: We first use \bs to determine all
solutions $p_i=(\alpha,\beta_i)$, $i=1,\ldots,l$, of the system $f=f_y=0$
with $x$-coordinate $\alpha$. Then, for each $p_i$, we compute 
\begin{align*}
k_i:=\min\{k:f_{y^k}(\alpha,\beta_i)=\frac{\partial^k f}{\partial
y^k}(\alpha,\beta_i)\neq 0\}\ge 2.
\end{align*}
The latter computation is done by iteratively calling \bs for
$f_y=f_{y^2}=0$, $f_{y^2}=f_{y^3}=0$,  and so on, and, finally, by restricting and
sorting the solutions along
the vertical line $x=\alpha$. We eventually obtain disjoint intervals
$I_1,\ldots,I_l$ and corresponding multiplicities $k_1,\ldots,k_l$ such
that $\beta_j$ is a $k_j$-fold root of $f_{\alpha}$ which is contained in
$I_j$. The intervals $I_j$ already separate the roots $\beta_j$ from any
other multiple root of $f_{\alpha}$, however, $I_j$ might still contain
ordinary roots of $f_{\alpha}$. Hence, we further refine
each $I_j$ until we can guarantee via interval arithmetic that
$\frac{\partial^{k_{j}} f}{\partial y^{k_j}}(\alpha,y)$ does not vanish on $I_j$.
If the latter condition is fulfilled, then $I_j$ cannot contain any root of
$f_{\alpha}$ except $\beta_j$ due to the Mean Value Theorem. Thus, after refining $I_j$, we can guarantee that $I_j$ is isolating. It remains to isolate the ordinary roots of $f_{\alpha}$:

We consider the so-called
\emph{Bitstream Descartes} isolator~\cite{ek+-descartes} (\bdesc for
short) which constitutes a variant of the Descartes method working
on polynomials with interval coefficients. This method can be used to get
arbitrary good approximations of the real roots of a polynomial with
``bitstream'' coefficients, that is, coefficients that can be approximated
to arbitrary precision. \bdesc starts from an interval guaranteed to
contain all real roots of a polynomial and proceeds with interval
subdivisions giving rise to a \emph{subdivision tree}. Accordingly, the
approximation precision for the coefficients is increased in each step of
the algorithm. Each leaf of the tree is associated with an interval $I$ and
stores a lower bound $l(I)$ and an upper bound $u(I)$ for the number of
real roots of $f_{\alpha}$ within this interval based on Descartes' Rule of Signs. Hence, $u(I)=0$ implies that $I$ contains no root and thus can be discarded. If $l(I)=u(I)=1$, then $I$ is an \emph{isolating interval} for a simple root. Intervals with $u(I)>1$ are further subdivided. We remark that,
after a number of iterations, \bdesc isolates all simple roots of a bitstream polynomial, and intervals not containing any root are eventually discarded. For a multiple root $\xi$, \bdesc determines an interval $I$ which approximates $\xi$ to an arbitrary good precision but never certifies such an interval $I$ to be isolating.

Now, in order to isolate the ordinary roots of $f_{\alpha}$, we modify \bdesc in the following way: We discard an interval $I$ if one of following three cases applies:
\begin{inparaenum}[i)] \item $u(I)=0$, or \item $I$ is completely
contained in one of the intervals $I_j$, or \item $I$ contains an interval $I_j$
and $u(I)\le k_j$. \end{inparaenum}
Namely, in each of these situations, $I$ cannot contain
an ordinary root of $f_{\alpha}$. An interval $I$ is stored as isolating
for an ordinary root of $f_{\alpha}$ if $l(I)=u(I)=1$, and $I$ intersects
no interval $I_j$. All intervals which do not fulfill one of the above
conditions are further subdivided. In a last step, we sort the intervals $I_j$ (isolating the multiple roots) and the newly obtained isolating intervals 
for the ordinary roots along the vertical line.

We remark that, in our implementation, \bs applied in \slowlift reuses 
the resultant $\res(f,f_y;y)$ which has already been computed in the projection phase of the algorithm.
Furthermore, it is a \emph{local approach} in the sense that
its cost is almost proportional to the number of $x$-critical fibers that have to be considered.
This will turn out to be beneficial in the overall approach, where most
fibers can successfully be treated by \fastlift; see Section~\ref{sssec:ca:lift:algo}.\\

\subsubsection{\fastlift --- a symbolic-numeric approach for fiber computation}
\label{sssec:ca:lift:fastlift}

Many of the existing algorithms to isolate the roots of $f_{\alpha}(y)=f(\alpha,y)$ are based on the computation of additional (combinatorial) information about $f_{\alpha}$ such as the degree $k=k_{\alpha}$ of $\gcd(f_{\alpha},f_{\alpha}')$, or the number $m=m_{\alpha}$ of distinct real roots of $f_{\alpha}$; for instance, in~\cite{eigenwilligkw07}, the values $m$ and $k$ are determined by means of computing a subresultant sequence before using a variant of the \bdesc method (denoted $m$-$k$-Descartes) to eventually isolate the roots of $f_{\alpha}$. Unfortunately, the additional symbolic operations for computing the entire subresultant sequence have turned out to be very costly in practice. The following consideration will show that \emph{the number $n_{\alpha}$} ($=\deg(f_{\alpha})-k_{\alpha}$)\emph{ of distinct complex roots of $f_{\alpha}$} can be computed by means of resultant and gcd computations, and a \emph{single} modular subresultant computation only. In order to do so, we first compute an upper bound $n_{\alpha}^+$ for each $n_{\alpha}$, where $n_{\alpha}^+$ has the following property: 
\begin{align}\nonumber
&\text{If }C\text{ has no vertical asymptote at }x=\alpha,\text{ and each critical point }(\alpha,\beta)\text{ (i.e.~}f_x(\alpha,\beta)\\
&=f_y(\alpha,\beta)=0\text{) on the vertical line }x=\alpha\text{ is also located on }C\text{, then }n_{\alpha}=n_{\alpha}^+.\label{specialgeo} 
\end{align}
We will later see that the condition in (\ref{specialgeo}) is always fulfilled if $C$ is in a generic location. From our experiments, we report that, for almost all considered instances, the condition is fulfilled for all fibers. Only for a very few instances, we observed that $n_{\alpha}\neq n_{\alpha}^+$ for a small number of fibers. 
In order to check in advance whether $n_{\alpha}=n_{\alpha}^+$ for all $x$-critical values $\alpha$, we will later introduce an additional test that uses a single modular computation and a semi-continuity argument.\\

\paragraph{Computation of $\mathbf{n_{\alpha}^+}$}
The following result due to Teissier~\cite{Gwozdziewicz00formulaefor,Teissier} is crucial for our approach:
\begin{lemma}[Teissier]
For an $x$-critical point $p=(\alpha,\beta)$ of $C$, it holds that
\begin{align}
\operatorname{mult}(f(\alpha,y),\beta)=\operatorname{Int}(f,f_y,
p)-\operatorname{Int}(f_x,f_y,p)+1,\label{Teissier}
\end{align}
where $\operatorname{mult}(f(\alpha,y),\beta)$ denotes the multiplicity of
$\beta$ as a root of $f(\alpha,y)\in\RR[y]$, $\operatorname{Int}(f,f_y,p)$
the intersection multiplicity\footnote{The intersection multiplicity of
two curves $f=0$ and $g=0$ at a point $p$ is defined as the dimension of
the localization of $\C[x,y]/(f,g)$ at $p$, considered as a $\C$-vector
space.} of the curves implicitly defined by $f=0$ and $f_y=0$ at $p$, and
$\operatorname{Int}(f_x,f_y,p)$ the intersection multiplicity of $f_x=0$
and $f_y=0$ at $p$. 
\end{lemma}

\begin{remark}\label{remark:Teissier}
In the case, where $f_x$ and $f_y$ share a
common non-trivial factor $h=\gcd(f_x,f_y)\in\Z[x,y]\backslash\Z$, $h$ does not vanish
on any $x$-critical point $p$ of $C$, that is, the curves $h=0$ and $f=0$ only intersect at infinity. Namely, $h(p)=0$ for some $p\in\C^2$ would imply that
$\operatorname{Int}(f_x,f_y,p)=\infty$ and, thus,
$\operatorname{Int}(f,f_y,p)=\infty$ as well, a contradiction to our
assumption on $f$ to be square-free. Hence, we have
$\operatorname{Int}(f_x,f_y,p)=\operatorname{Int}(f_x^*,f_y^*,p)$ with
$f_x^*:=f_x/h$ and $f_y^*:=f_y/h$. Hence, the following more general formula (which
is equivalent to (\ref{Teissier}) for trivial $h$) applies: 
\begin{align}
\operatorname{mult}(f(\alpha,y),\beta)=\operatorname{Int}(f,f_y,
p)-\operatorname{Int}(f_x^*,f_y^*,p)+1.\label{Teissier2}
\end{align}
\end{remark}

We now turn to the computation of the upper bound $n_{\alpha}^+$. We distinguish the cases $\deg f_{\alpha}\neq \deg_y f$ and $\deg f_{\alpha}= \deg_y f$. In the first case, where $C$ has a vertical asymptote at $\alpha$, we define $n_{\alpha}^+:=\deg f_{\alpha}$ which is obviously an upper bound for $n_{\alpha}$.
In the case $\deg f_{\alpha}=\deg_y f$, the formula (\ref{Teissier2}) yields:
\begin{align}\nonumber
n_{\alpha}&=\#\{\text{distinct complex roots of }f_{\alpha}\}
= \deg_y f-\deg\gcd(f(\alpha,y),f_y(\alpha,y))\\ \nonumber
&= \deg_y f - \sum_{\substack{\beta\in\C:\\
    \makebox[2em][c]{\scriptsize $f(\alpha,\beta)=0$}}}(\operatorname{mult}(f(\alpha,y),
\beta)-1)\\ \nonumber
&= \deg_y f - \sum_{\substack{\beta\in\C:\\
    \makebox[2em][c]{\scriptsize $(\alpha,\beta)$ is $x$-critical}}}
\left(\operatorname{Int}(f,f_y,(\alpha,\beta))-\operatorname{
Int}(f_x^*,f_y^*,(\alpha,\beta))\right)\\
&= \deg_y f-\operatorname{mult}(R,\alpha)+ \sum_{\substack{\beta\in\C:\\
    \makebox[2.5em][c]{\scriptsize $(\alpha,\beta)$ is $x$-critical}}}
\operatorname{Int}(f_x^*,f_y^*,(\alpha,\beta))\label{eq1}\\
&\le  \deg_y f-\operatorname{mult}(R,\alpha)+\sum_{\beta\in\C}\operatorname{Int}(f_x^*,
f_y^*,(\alpha,\beta)) \label{ineq1}\\
&=\deg_y
f-\operatorname{mult}(R,\alpha)+\operatorname{mult}(Q,\alpha)=:n_{\alpha}^+\label{eq2}
\end{align}
where $R(x)=\res(f,f_y;y)$ and
$Q(x):=\res(f_x^*,f_y^*;y)$. The equality (\ref{eq1}) is due
to the fact that $f$ has no vertical asymptote at $\alpha$ and, thus, the
multiplicity $\operatorname{mult}(R,\alpha)$ equals the sum
$\sum_{\beta\in\C}\operatorname{Int}((f,f_y,(\alpha,\beta))$ of the
intersection multiplicities of $f$ and $f_y$ in the fiber at $\alpha$. (\ref{eq2}) follows by an analogous argument for the
intersection multiplicities of $f_x^*$ and $f_y^*$ along the vertical line
at $\alpha$. From the square-free factorization of $R$, the value
$\operatorname{mult}(R,\alpha)$ is already computed, and $\operatorname{mult}(Q,\alpha)$ can be
determined, for instance, by computing $Q$, its square-free factorization and checking
whether $\alpha$ is a root of one of the factors. 
The following theorem shows that, if the curve $C$ is in generic position, then $C$ has no vertical asymptote or a vertical line, and $f_x^*$ and $f_y^*$ do not intersect at
any point above $\alpha$ which is not located on $C$.\footnote{The reader may notice that generic position is used in a different context here. It is required that all intersection points of $f_x^*$ and $f_y^*$ above an $x$-critical value $\alpha$ are located on the curve $C$.} In the latter case, the inequality (\ref{ineq1})
becomes an equality, and thus $n_{\alpha}=n^+_{\alpha}$.\\

\begin{theorem}
\label{thm:genericposition}
For a generic $s\in\RR$ (i.e.~for all but finitely many), the sheared curve $$C_s:=\{(x,y)\in\RR^2:f(x+s\cdot y,y)=0\}$$
yields $n_{\alpha}^+=n_{\alpha}$ for all $x$-critical values $\alpha$ of $C_s$.
\end{theorem}

\begin{proof}
For a generic $s$, the leading coefficient of $f(x+sy,y)$ (considered as a polynomial in $y$) is a constant, hence we can assume that $C$ has no vertical asymptote and contains no vertical line. We can further assume that $f_x$ and $f_y$ do not share a common non-trivial factor $h$. Otherwise, we have to remove $h$ first; see also Remark~\ref{remark:Teissier}.
Let $g(x,y)=f(x+sy,y)\in\RR[x,y]$ denote the defining equation of the sheared curve $C_s$, then the critical points of $C_s$ are the common solutions of $$g_x(x,y)=f_x(x+sy,y)=0\quad\text{and}\quad g_y(x,y)=f_x(x+sy)\cdot s+f_y(x+sy,y)=0.$$
Hence, the critical points of $C_s$ are exactly the points $(\alpha',\beta')=(\alpha-s\beta,\beta)$, where $(\alpha,\beta)$ is a critical point of $C$.
We now consider a specific $(\alpha,\beta)$ and show that, for a generic $s$, the polynomial $g(\alpha',y)$ has either no multiple root or exactly one multiple root at $y=\beta'=\beta$, where $(\alpha',\beta')=(\alpha-s\beta,\beta)$ denotes the corresponding critical point of $C_s$. Then, the same holds for all critical values $(\alpha',\beta')$ in parallel because there are only finitely many critical $(\alpha,\beta)$ for $C$.
Hence, from the definition of $n_{\alpha'}^+$, it then follows that $n_{\alpha'}^+=n_{\alpha'}$ for all $x$-critical values $\alpha'$ of $C_s$.
W.l.o.g., we can assume that $(\alpha,\beta)=(0,0)$, and thus $(\alpha',\beta')=(0,0)$ for the corresponding critical point of $C_s$. Let $y^m$ be the highest power of $y$ that divides $g(0,y)=f(sy,y)$, and define $f^*(s,y):=f(sy,y)/y^m$. 
If there exists an $s_0\in\RR$ such that $f^*(s_0,y)$ has no multiple root, then we are done. Otherwise, for each $s$, $f^*(s,y)$ has a multiple root $y_0$ that is different from $0$. It follows that $f^*(s,y)$ is not square-free, that is, there exist polynomials $p_1,p_2\in\C[s,y]$ with 
\begin{align}
f^*(s,y)=\frac{f(sy,y)}{y^m}=p_1^2(s,y)\cdot p_2(s,y)\label{eq:factors}
\end{align}
We remark that, for each $s\in\C$, there exists a $y_s\in\C\backslash\{0\}$ such that $p_1(s,y_s)=0$. Hence, for $x_s:=s/y_s$, we have $p_1(x_s/y_s,y_s)=0$, and thus $p_1(x/y,y)$ cannot be a power of $y$. Now plugging $s=x/y$, with $y\neq 0$, into (\ref{eq:factors}) yields
\[
f(x,y)=y^m\cdot p_1^2(x/y,y)\cdot p_2(x/y,y)=y^m\cdot\left(\frac{\tilde{p}_1(x,y)}{y^{m_1}}\right)^2\cdot \frac{\tilde{p}_2(x,y)}{y^{m_2}}=y^{m-2m_1-m_2}\cdot \tilde{p}_1^2(x,y)\cdot \tilde{p}_2(x,y),
\]
where $\tilde{p}_1,\tilde{p}_2\in\C[x,y]$, and $m_1,m_2\in\N$. Since $f(x,y)$ is square-free, this is only possible if $\tilde{p}_1(x,y)$ is a power of $y$. This implies that $p_1(x/y,y)=\tilde{p}_1(x,y)/y^{m_1}$ is also a power of $y$, a contradiction.
\end{proof}

We remark that, in the context of computing the topology of a planar algebraic curve, Teissier's formula has already been used in~\cite{cheng.lazard.ea:on,LuisPhD2010}. There, the authors apply (\ref{Teissier}) in its simplified form (i.e.~$\operatorname{Int}(f_x,f_y,p)=0$) to compute $\operatorname*{mult}(\beta,f(\alpha,y))$ for a non-singular point $p=(\alpha,\beta)$. In contrast, we use the formula in its general form and sum up the information along the entire fiber which eventually leads to the upper bound $n_{\alpha}^+$ on the number of distinct complex roots of $f_{\alpha}$.\\

\newcommand{\LB}{\ensuremath{N^-}}
\newcommand{\MM}{\ensuremath{N}}
\newcommand{\UB}{\ensuremath{N^+}}

\newcommand{\ZZ}{\ensuremath{\mathbb{Z}}}
\newcommand{\ZZp}{\ensuremath{\ZZ_p}}
\newcommand{\Sres}{\ensuremath{\operatorname{Sres}}}
\newcommand{\sres}{\ensuremath{\operatorname{sres}}}
\newcommand{\coeff}{\ensuremath{\operatorname{coeff}}}
\newcommand{\lcoeff}{\ensuremath{\operatorname{lcoeff}}}
\newcommand{\mult}{\ensuremath{\operatorname{mult}}}

In the next step, we provide a method to check in advance whether the curve $C$ is in a generic position in the sense of Theorem~\ref{thm:genericposition}.  
Unfortunately, we see no cheap way to check generic position with respect to a \emph{specific} $x$-critical fiber $x=\alpha$, that is, whether $n_\alpha^+$ matches $n_\alpha$ for a specific $\alpha$.
However, we can derive a global test to decide whether the upper bound $n^+_{\alpha}$ matches $n_{\alpha}$ for \emph{all} fibers.
While the evaluation of the corresponding test with exact integer arithmetic 
is expensive, we can use the same argument to derive a conservative modular 
test which returns the same answer with very high probability. The test relies 
on the comparison of an \emph{upper bound $N^+$ for $\sum_{\alpha} 
n^+_{\alpha}$} (i.e.~$N^+\ge\sum_{\alpha} 
n^+_{\alpha}\ge N:=\sum_{\alpha}n_{\alpha}$) and a \emph{lower bound $N^-$ for $N$} (i.e.~$N^-\le N=\sum_{\alpha}n_{\alpha}$), where we sum over all (complex) $x$-critical 
values $\alpha$.
Then, $N^-=N^+$ implies that $n_{\alpha}=n_{\alpha}^+$ for all $\alpha$. We 
now turn to the computation of $N^-$ and $N^+$. Here, we assume that $f$ has 
no vertical asymptote and no vertical component (in particular, $\deg_y 
f(\alpha, y) = \deg_y f(x,y) =: n_y$ for all values $\alpha.$).

\paragraph{Computation of $\mathbf{\UB}$}
\begin{lemma}
The sum over all $n^+_\alpha$, $\alpha$ a complex $x$-critical value of $C$, yields:
  \begin{align*}
    (\deg_x R^* \cdot \deg_y f) - \deg_x R + \deg_x \gcd (R^\infty, Q),
  \end{align*}
  where $Q = \res (f_x^*, f^*_y; y)$, and $\gcd (R^\infty, Q)$ is defined as the product of all common factors of $R$ and $Q$ with multiplicity according to their occurrence in $Q.$
\end{lemma}
\begin{proof}
  For the first term, note that $\deg_x R^*$ is the number of distinct complex $x$-critical values for $f$ and, thus, the number of summands in $\sum_{\alpha}n_{\alpha}$.
  The sum over all multiplicities $\mult(R,\alpha)$ for the roots $\alpha$ of $R$ simply yields the degree of $R.$
  Finally, the summation over $\mult(Q,\alpha)$ amounts to removing the factors of $Q$ that do not share a root with $R.$
\end{proof}

We remark that the square-free part $R^*$ of the resultant $R$ is already computed in the projection phase of the curve analysis, and thus we already know $\deg_x R^*$. The additional computation of $Q$ and $\gcd(R^\infty, Q)$ can be performed over a modular prime field $\ZZp$ for some randomly chosen prime $p.$ Then, $\deg_x (\gcd (R^\infty\bmod p, Q\bmod p))\ge\deg_x \gcd(R^\infty, Q)$, and thus
\begin{align}
\UB:=(\deg_x R^* \cdot \deg_y f) - \deg_x R + \deg_x (\gcd (R^\infty\bmod p, Q\bmod p))\label{eq:up-bound}
\end{align}
constitutes an upper bound for $\sum_{\alpha} n_{\alpha}^+$. We remark that the result obtained by the modular computation matches $\sum_{\alpha}n_{\alpha}^+$ with very high probability. That is, up to the choice of finitely many ``unlucky'' primes, we have $\UB=\sum_{\alpha}n_{\alpha}^+$.\\

In the next step, we show how to compute a lower bound $\LB$ for $N$. In order to understand its construction, we first explain how to exactly compute $N$. We stress that our algorithm never performs this computation.

\paragraph{(Exact) Computation of $\mathbf{\MM}$}
Consider a decomposition of the square-free part $R^*$ of the resultant $R = \res (f, f_y; y)$:
\begin{align}\label{eq:res-decomposition-exact}
  R^* = R_1 R_2 \cdots R_s, \quad R_i \in \ZZ[x],
\end{align}
such that $R_i(\alpha) = 0$ if and only if $f(\alpha, y)$ has exactly $n_y-i$ distinct complex roots.
Note that all $R_i$ are square-free and pairwise coprime.
With $d_i := \deg R_i$ the degree of the factor $R_i$, it follows that
\begin{align*}
  \MM = \sum_{1 \le i \le r} (n_y-i) \cdot d_i.
\end{align*}

\noindent The computation of the decomposition in (\ref{eq:res-decomposition-exact}) uses \emph{subresultants.} The \emph{$i$-th subresultant} polynomial $\Sres_i(f,g;y)\allowbreak \in \ZZ[x,y]$ of two bivariate polynomials $f$ and $g$ with $y$-degrees $m_y$ and $n_y,$ respectively, is defined as the determinant of a Sylvester-like matrix.
\begin{align*}
  \Sres_i (f,f_y;y) := \left\lvert\;
    \begin{matrix}
      f_{m_y}^{(y)} & f_{m_y-1}^{(y)} & \cdots & \cdots & f_{2i-n_y+2}^{(y)} & y^{n_y-i-1} f\\
      & \ddots & \ddots & \ddots & \vdots & \vdots\\
      & & f_{m_y}^{(y)} & \cdots & f_{i+1}^{(y)} & f\\
      g_{n_y}^{(y)} & g_{n_y-1}^{(y)} & \cdots & \cdots & g_{2i-m_y+2}^{(y)} & y^{m_y-i-1} g\\
      & \ddots & \ddots & \ddots & \vdots & \vdots\\
      & & g_{n_y}^{(y)} & \cdots & g_{i+1}^{(y)} & g
    \end{matrix}
    \;\right\rvert\!\!
  \begin{array}{l}
    \left.
    \begin{matrix}
      \vphantom{fg_{m_y}^{(y)}}\\\vphantom{fg_{m_y}^{(y)}}\\\vphantom{fg_{m_y}^{(y)}}
    \end{matrix}
    \right\}\,\text{\small $n_y-i$ rows}\\
    \vspace*{-0.95em}\\
    \left.
    \begin{matrix}
      \vphantom{fg_{m_y}^{(y)}}\\\vphantom{fg_{m_y}^{(y)}}\\\vphantom{fg_{m_y}^{(y)}}
    \end{matrix}
    \right\}\,\text{\small $m_y-i$ rows}
  \end{array}
\end{align*}
The subresultants exhibit a direct relation to the number and multiplicities of common roots of $f$ and $g.$
More specifically, it holds that $\deg\gcd(f(\alpha,y),g(\alpha,y))=k$ if and only if the \emph{$i$-th principal subresultant coefficient (psc)} $\operatorname{sr}_i(x):=\sres_i(f,g;y) := \coeff_i(\Sres_i(f, g; y);y) \in \ZZ[x]$ vanishes at $\alpha$ for all $i=0,\ldots,k-1$, and $\operatorname{sr}_k(\alpha)\neq 0$ (e.g.~see \cite{KerberPhD2009,alg-geom-06} for a proof).

Thus, the decomposition in (\ref{eq:res-decomposition-exact}) can be derived as
\begin{alignat}{3}\label{def:Sis}
  S_0 &:= R^*,
  &\qquad
  S_i &:= \gcd (S_{i-1}, \operatorname{sr}_i) &&\quad\text{ for } i=1,\dots,s,\\ 
  R_1 &:= \frac{S_0}{\gcd(S_0, S_1)},
  &\qquad
  R_i &:= \frac{\gcd(S_0, \dots, S_{i-1})}{\gcd(S_0, \dots, S_{i-1}, S_i)} &&\quad\text{  for } i=1,\dots,s,\label{def:Ris}
\end{alignat}
where $s$ is the number of non-trivial entries in the subresultant sequence of $f$ and $f_y.$ The computation of $\MM$ as described here requires the exact computation of all psc's, a very costly operation which would affect the overall runtime considerably. Instead, we consider the following modular approach:

\paragraph{Computation of $\mathbf{\LB}$}
The main idea of our approach is to perform the above subresultant computation 
over $\ZZp$ for a \emph{single}, randomly chosen prime $p$. More precisely, we denote 
$$\operatorname{sr}_i^{(p)}(x):=\sres_i^{(p)}(f^{(p)},g^{(p)};y) := 
\coeff_i(\Sres_i^{(p)}(f^{(p)},g^{(p)}; y);y) \in \ZZp[x]$$ the $i$-th principle subresultant coefficient in 
the subresultant sequence of $f^{(p)}:=f\bmod p\in\ZZp[x,y]$ and $g^{(p)}:=g\bmod 
p\in\ZZp[x,y]$. The polynomials $S_i^{(p)}\in\ZZp[x]$ and $R_i^{(p)}\in\ZZp[x]$ are then defined in completely 
analogous manner as the polynomials $S_i\in\Z[x]$ and $R_i\in\Z[x]$ in (\ref{def:Sis}) and (\ref{def:Ris}), respectively. 
The following lemma shows that this yields a lower bound for $N$ if $p$ does not divide the leading coefficient of $f$ and $f_y$: 
\begin{lemma}
  Let $p$ be a prime that does not divide the leading coefficient of $f$ and $f_y$, and let $d_i^{(p)} := \deg R_i^{(p)}$ denote the degree of $R_i^{(p)}$. 
Then,
  \begin{align}
    \LB := \sum_{i \ge 1} (n_y-i) \cdot d_i^{(p)} \label{eq:lower-bound}
  \end{align}
  constitutes a lower bound for the total number $\MM$ of distinct points on $C$ in $x$-critical fibers.
\end{lemma}
\begin{proof}
  It suffices to show that $\sum_{i\ge 1} i \cdot d_i\le \sum_{i\ge 1} i\cdot d_i^{(p)}.$
  Namely, using $d := \deg R^* = \sum_{i\ge 1} d_i = \sum_{i\ge 1} d_i^{(p)},$ we obtain
  \begin{align*}
    \LB = \sum_{i \ge 1} (n_y-i) \cdot d_i^{(p)}
    = n_y d - \sum_{i \ge 1} i \cdot d_i^{(p)}
    \le n_y d - \sum_{i \ge 1} i \cdot d_i
    = \sum_{i \ge 1} (n_y-i) \cdot d_i
    = N.
  \end{align*}
  %
  %
Since $p$ does not divide the leading coefficient of $f$ and $f_y$, we have $\operatorname{sr}_i^{(p)}=\operatorname{sr}_i\bmod p$ due to the specialization property of subresultants. Hence, $n_i^{(p)}:=\deg S_i^{(p)}\ge n_i:=\deg S_i$ which implies the following diagram (with some $t$ such that $s\le t\le n$)
  \begin{align*}
    \newcommand{\LE}{\parbox[c]{\widthof{\rotatebox{-90}{$\le$}}}{\rotatebox{-90}{$\le$}}}
    \newcommand{\GE}{\parbox[c]{\widthof{\rotatebox{-90}{$\ge$}}}{\rotatebox{-90}{$\ge$}}}
    \newcommand{\E}{\parbox[c]{\widthof{\rotatebox{-90}{$=$}}}{\rotatebox{-90}{$=$}}}
    \begin{array}{r@{\;}*{13}{c@{\;}}l}
      d = &n_0^{(p)} &\ge &n_1^{(p)} &\ge &\cdots &\ge &n_s^{(p)} &\ge &n_{s+1}^{(p)} &\ge &\cdots &\ge &n_t^{(p)} &= 0\\[0.5ex]
      &\E & &\GE & &  & &\GE & &\GE & &  & &\E \\[0.5ex]
      d = &n_0 &\ge &n_1 &\ge &\cdots &\ge &n_s &= &n_{s+1} &= &\cdots &= &n_t &= 0
    \end{array}
  \end{align*}
  Furthermore, we have $d_i = n_{i-1} - n_i$ and $d_i^{(p)} = n_{i-1}^{(p)} - n_i^{(p)}.$ Thus,
  \begin{align*}
    \sum_{i\ge 1} i \cdot d_i
    = \sum_{i\ge 1} \sum_{j\ge i}^s d_j
    = \sum_{i\ge 1} \sum_{j\ge i}^s (n_{j-1} - n_j)
    = \sum_{i\ge 1} n_{i-1}
    = \sum_{i\ge 0} n_i \quad(\text{since } n_i = 0\text{ for }i\ge s)
  \end{align*}
  and, analogously, $\sum_{i\ge 1} i \cdot d_i^{(p)} = \sum_{i\ge 0} n_i^{(p)}.$ This shows $\sum_{i\ge 1} i \cdot d_i \le \sum_{i\ge 1} i\cdot d_i^{(p)}$.
\end{proof}

We remark that, for all but finitely many (unlucky) choices of $p,$ all polynomials $R_i$ and $R^{(p)}_i$ have the same degree. 
Thus, with high probability, $\LB$ as defined in (\ref{eq:lower-bound}) matches $\MM$. In addition, also with very high probability, we have $N^+=\sum_{\alpha}n^{+}_{\alpha}$. Hence, if the curve $C$ is in generic position and our choice of $p$ is not unlucky, then $N^-=N^+=N$, and thus we can certify in advance that $n_\alpha = n_\alpha^+$ for all $x$-critical values $\alpha$. We would like to emphasize that the only exact computation (over $\Z$) that is needed for this test is that of the square-free part of the resultant $R$ (more precisely, only that of its degree). All other operations can be performed over $\ZZp$ for a single, randomly chosen prime $p$. Putting everything together now yields our method \fastlift to compute the fiber at an $x$-critical value:

\paragraph{\fastlift} We consider a hybrid method to isolate all complex roots and, thus, also the real roots
of $f_{\alpha}(y)=f(\alpha,y)\in\RR[y]$, where $\alpha$ is a real valued $x$-critical value of the curve $C$. It combines 
\begin{inparaenum}[(a)]
  \item a numerical solver to compute arbitrary good approximations (i.e.~complex discs in $\mathbb{C}$) of the roots of $f_{\alpha}$,
  \item an exact certification step to
certify the existence of roots within the computed discs, and 
  \item additional knowledge on the number $n_{\alpha}$ of distinct (complex) roots of $f_{\alpha}$.
\end{inparaenum}
\fastlift starts with computing the upper bound $n_{\alpha}^+$ for $n_{\alpha}$ and the values $N^-$ and $N^+$ as defined in (\ref{eq2}), (\ref{eq:lower-bound}), and (\ref{eq:up-bound}), respectively. We distinguish two cases:

\begin{itemize}
\item $N^-=N^+$: In this case, we know that $n_{\alpha}=n_{\alpha}^+$. We now use a numerical solver to determine disjoint discs $D_1,\ldots,D_{m}\subset\C$ and an exact certification step to certify the existence of a
certain number $m_i\ge 1$ of roots (counted with multiplicity) of
$f_{\alpha}$ within each $D_i$; see \ref{asec:numerical} for
details. Increasing the working precision and the number of iterations within
the numerical solver eventually leads to arbitrary well refined discs $D_i$
-- but without a guarantee that these discs are actually isolating! However, from a certain iteration on, the number of discs certified to contain at least one root matches $n_{\alpha}$. When this happens, we know for sure that the $D_i$'s are isolating. We can then further refine these discs until, for all
$i=1,\ldots,m$, 
\begin{align}
D_i\cap\mathbb{R}= \emptyset\text{ or }\bar{D}_i\cap D_j=\emptyset\text{
for all }j\neq i,\label{condition1} 
\end{align}
where $\bar{D}_i:=\{\bar{z}:z\in D_i\}$ denotes the complex conjugate of
$D_i$. The latter condition guarantees that each disc $D_i$ which
intersects the real axis actually isolates a real root of $f_{\alpha}$. In
addition, for each real root isolated by some $D_i$, we further
obtain its multiplicity $m_i$ as a root of $f_{\alpha}$.

\item $N^-<N^+$: In this case, we have either chosen an unlucky prime in some of the modular computations, or the curve $C$ is located in a special geometric situation; see (\ref{specialgeo}) and Theorem~\ref{thm:genericposition}. However, despite the fact that there might exist a few critical fibers where $n_{\alpha}<n_{\alpha}^+$, there is still a good chance that equality holds for most $\alpha$. Hence, we propose to use the numerical solver \emph{as a filter} in a similar manner as in the case, where $N^- = N^+$. More precisely, we run the numerical solver on $f_{\alpha}$ for a certain number of iterations.\footnote{The threshold for the number of iterations 
should be chosen based on the degree of $f$ and its coefficient's bitlengths. For the instances
considered in our experiments, we stop when reaching 2048 bits of precision.} Since $n_{\alpha}^+$ constitutes an upper bound on the number of distinct complex roots of $f_{\alpha}$, we must have $m\le n_{\alpha}\le n_{\alpha}^+$ at any time. Hence, if the number $m$ equals $n_{\alpha}^+$, we know for sure that all complex roots of $f_{\alpha}$ are isolated and can then proceed as above. If, after a number of iterations, it still holds that $m<n_{\alpha}^+$, \fastlift reports
a failure.
\end{itemize}

\fastlift is a certified method, that is, in case of success, 
it returns the mathematical correct result. However, in comparison to the 
complete method \slowlift, \fastlift may not apply to all critical fibers
if the curve $C$ is in a special geometric situation.
We would like to remark that, for computing the topology of the curve $C$ only, we can exclusively use \fastlift as the lifting method.
Namely, when considering, as indicated earlier,
an initial shearing $x\mapsto x+s\cdot y$, with $s$ a randomly 
chosen integer, the sheared curve $$C_s:=\{(x,y)\in\RR^2:f(x+s\cdot 
y,y)=0\}$$
is in generic situation (with high probability) due to Theorem~\ref{thm:genericposition}. Then, 
up to an unlucky choice 
of prime numbers in the modular computations, we obtain bounds $N^-$ and $N^+$ 
for $N$ which are equal. Hence, up to an unlucky choice of finitely many "bad" 
shearing parameters $s$ and primes $p$, the curve $C_s$ is in a generic situation, and, in addition, we can actually prove this.
It follows that $n_{\alpha}^+=n_{\alpha}$ for all 
$x$-critical values of the sheared curve $C_s$, and thus \fastlift is successful for all fibers. Since the sheared curve is 
isotopic to $C$, this shows that we can always compute the topology of $C$ by 
exclusively using \fastlift during the lifting phase.

\subsubsection{\calift --- Combining \slowlift and \fastlift}
\label{sssec:ca:lift:algo}

We have introduced two different methods to compute the fibers at the $x$-critical values of a curve $C$. \slowlift is certified and complete, but turns out to be less efficient than \fastlift which, in turn, may fail for a few fibers for curves in a special geometric situation. Hence, in the lifting step, we propose to combine the two methods. That is, we run \fastlift by default, and fall back to \slowlift only if \fastlift fails. In practice, as observed in our experiments presented in Section~\ref{ssec:implex:ca}, 
the failure conditions for \fastlift are almost negligible, that is, the method only fails for a few critical fibers for some curves in a special geometric situation.
In addition, in case of a failure, we profit from the fact that our backup method \slowlift applies very well to a specific fiber. That is, its computational cost is almost proportional to the number of fibers that are considered.

We also remark that, for the modular computations of $N^-$ and $N^+$, we never observed any failure when choosing a reasonable large prime. 
However, it should not be concealed that we only performed these computations off-line in Maple.
Our C++-implementation still employs a more naive approach, where we always use \fastlift as a filter as described in the case $N^-<N^+$ above.

In the last section, we mentioned that \fastlift can be turned into a complete method when considering an initial coordinate transformation. Hence, one might ask why we do not consider such a transformation to compute the topology of $C$. There are several reasons to not follow this approach. Namely, when considering 
a shearing, the algorithm computes the topology of $C$, but does not directly 
yield a geometric-topological analysis of the curve since the vertices of the so-obtained graph are not located on $C$.
In order to achieve the latter as well, we still have to "shear back" the information for the sheared curve, an operation which is non-trivial at all; see~\cite{eigenwilligkw07} for details. Even though the latter approach seems manageable for a single curve, it considerably complicates the arrangement computation (see Section~\ref{sec:arr}) because the majority of the input curves can be treated in the initial coordinates.
Furthermore, in particular for sparse input, a coordinate transformation induces considerably higher computational costs in all subsequent operations.\\

\subsection{\caconnect}
\label{ssec:ca:alg:connect}

Let us consider a fixed $x$-critical value $\alpha$, the corresponding
isolating interval $I(\alpha)=(a,b)$ computed in the projection phase and
the points $p_i:=(\alpha,y_{\alpha,i})\in C$, $i=1,\ldots,m_{\alpha}$,
located on $C$ above $\alpha$. Furthermore, let $I=(\alpha,\alpha')$ be the
interval connecting $\alpha$ with the nearest $x$-critical value to the
right of $\alpha$ (or $+\infty$ if none exists) and $A_j$, $j=1,\ldots,m_I$, the $j$-th arc of $C$ above
$I$ with respect to vertical ordering. $A_j$ is represented by a point
$a_j:=(q_I,y_{I,j})\in C$, where $y_{I,j}$ denotes the $j$-th real root of
$f(q_I,y)$ and $q_I$ an arbitrary but fixed rational value in $I$. 
To its left, $A_j$ is either
connected to $(\alpha,\pm \infty)$ (in case of a vertical asymptote) or to
one of the points $p_i$. In order to determine the point to which an arc
$A_j$ is connected, we consider the following two distinct cases:

\begin{figure*}[t]
  \centering
  \includegraphics[width=0.9\linewidth]{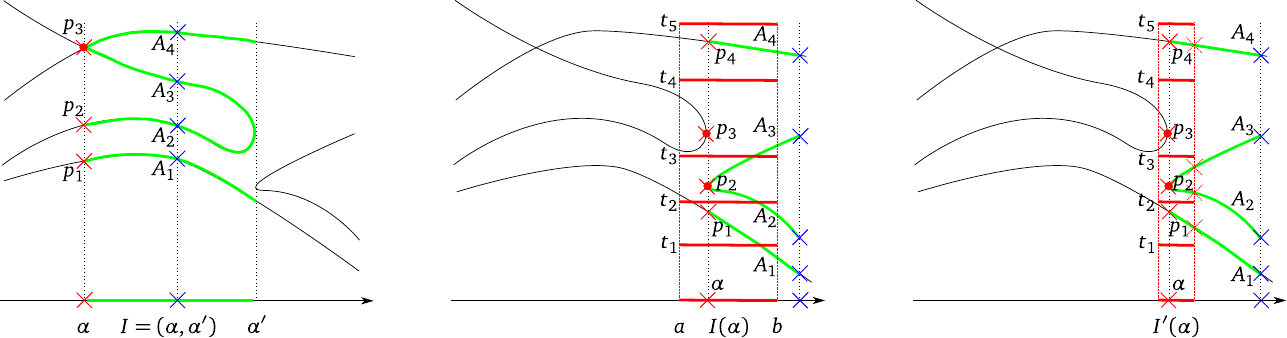}
\caption{\label{fig:connection}The left figure shows the generic case, where
exactly one $x$-critical point ($p_3$) above $\alpha$ exists.
  The bottom-up method connects $A_1$ to $p_1$ and $A_2$ to $p_2$; the
remaining arcs have to pass $p_3$. In the second figure, the fiber at
$\alpha$ contains two critical points $p_2$ and $p_3.$
  The red horizontal line segments pass through arbitrary chosen points
$(\alpha,t_i)$ separating $p_{i-1}$ and $p_i$.
  The initial isolating interval $I(\alpha) = (a,b)$ for $\alpha$ is not
sufficient to determine the connections for all arcs since $A_1, A_2, A_3$
intersect the segments $I \times \{ t_i \}.$ On the right, the refined
isolating interval $I'(\alpha)$ induces boxes $I'(\alpha) \times (t_i,
t_{i+1})$ small enough such that no arc crosses the horizontal boundaries.
By examination of the $y$-coordinates of the intersections between the arcs
and the fiber over the right-hand boundary of $I'(\alpha)$ (red crosses),
we can match arcs and critical points.\protect\\[-2em]\strut}
\end{figure*}

\begin{itemize}
\item The \emph{generic case}, that is, there exists exactly one real
$x$-critical point $p_{i_0}$ above $\alpha$ and $\deg f(\alpha,y)=\deg_y
f$. The latter condition implies that $C$ has no vertical asymptote at
$\alpha$. Then, the points $p_1,\ldots,p_{i_0-1}$ must be connected with
$A_1,\ldots,A_{i_0-1}$ in bottom-up fashion, respectively, since, for each of these points,
there exists a single arc of $C$ passing this point. The same argument
shows that $p_{i_0+1},\ldots,p_{m_{\alpha}}$ must be connected to
$A_{m_I-m_{\alpha}+i_0+1},\ldots,A_{m_I}$ in top-down fashion, respectively. Finally, the remaining arcs
in between must all be connected to the $x$-critical point $p_{i_0}$.

\item The \emph{non-generic case}: We choose
arbitrary rational values $t_1,\ldots,t_{m_{\alpha}+1}$ with
$t_1<y_{\alpha,1}<t_2<\ldots<y_{\alpha,m_{\alpha}}<t_{m_{\alpha}+1}$. Then,
the points $\tilde{p}_i:=(\alpha,t_i)$ separate the $p_i$'s from each
other. Computing such $\tilde{p}_i$ is easy since we have isolating
intervals with rational endpoints for each of the roots $y_{\alpha,i}$ of
$f(\alpha,y)$. In a second step, we use interval arithmetic to obtain
intervals $\mathfrak{B}f(I(\alpha)\times t_i)\subset\RR$ with $f(I(\alpha)\times
t_i)\subset \mathfrak{B}f(I(\alpha)\times t_i)$. As long as there exists an
$i$ with $0\in \mathfrak{B}f(I(\alpha)\times t_i)$, we refine $I(\alpha)$.
Since none of the $\tilde{p}_i$ is located on $C$, we eventually obtain a
sufficiently refined interval $I(\alpha)$ with $0\notin
\mathfrak{B}f(I(\alpha)\times t_i)$ for all $i$. It follows that none of
the arcs $A_j$ intersects any line segment $I(\alpha)\times t_i$. Hence,
above $I(\alpha)$, each $A_j$ stays within the rectangle bounded by the
two segments $I(\alpha)\times t_{i_0}$ and $I(\alpha)\times t_{i_0+1}$ and
is thus connected to $p_{i_0}$. In order to determine $i_0$, we compute the
$j$-th real root $\gamma_j$ of $f(b,y)\in\mathbb{Q}[y]$ and the largest
$i_0$ such that $\gamma_j>t_{i_0}$. In the special case where
$\gamma_j<t_i$ or $\gamma_j>t_i$ for all $i$, it follows that $A_j$ is
connected to $(\alpha,-\infty)$ or $(\alpha,+\infty)$, respectively.
\end{itemize} 

For the arcs located to the left of $\alpha$, we proceed in exactly the
same manner. This concludes the connection phase and, thus, the description
of our algorithm.

\begin{figure*}[t]
  \centering
  \includegraphics[width=\linewidth]{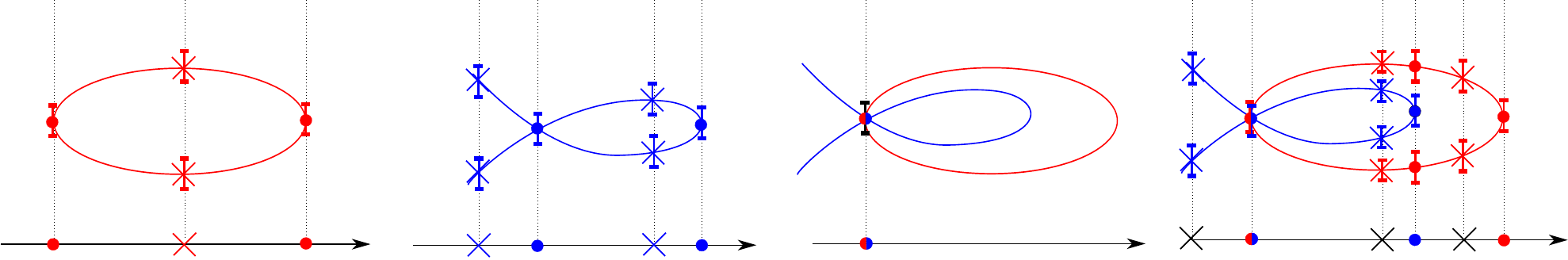}
\caption{The two figures on the left show the topology analyses for the curves $C=V(f)$ and $D=V(g)$. The second figure from the right shows the intersection of the two curves. For the curve pair analysis, critical event lines (at dots) are sorted and non-critical event lines (at crosses) in between are inserted. Finally, for each event line $x=\alpha$, the roots of $f(\alpha,y)$ and $g(\alpha,y)$ are sorted. The latter task is done by further refining corresponding isolating intervals (blue or red intervals) and using the combinatorial information from the curve analyses and the computation of the intersection points.\protect\\[-2em]\strut}
\label{fig:arrangement}
\end{figure*}

\section{Arrangement computation}
\label{sec:arr}

\cgal{}'s prevailing implementation for computing arrangements
of planar algebraic curves reduces all required geometric constructions
(as intersections) and predicates (as comparisons of points
and $x$-monotone curves) to the geo\-me\-tric-to\-po\-lo\-gi\-cal
analysis of a single curve~\cite{eigenwilligkw07} and pairs of 
curves~\cite{eigenwilligk08}; see also~\cite{bhk-ak2-2011} and 
\cgal{}'s documentation~\cite{cgal:wfzh-a2-11b}.

In Section~\ref{sec:ca}, we have already seen how to improve
the curve-analysis. In a similar way, we want to increase the performance
of the analyses of a pair of curves $C=V(f)$ and $D=V(g)$,
(see illustration in Figure~\ref{fig:arrangement}). In general, the algorithm from~\cite{eigenwilligk08} had to compute the entire subresultant sequence, an operation that we are aiming to avoid.
Using the new analyses of each single curve and combining the so-obtained information with the information on the intersection points of the two curves $C$ and $D$ as returned by \bs, it is straight-forward 
to achieve this goal. We mainly have to compute the common intersection 
points of the two curves:

Let $C=V(f)$ and $D=V(g)$ be two planar algebraic curves implicitly defined
by square-free polynomials $f$, $g\in\mathbb{Z}[x,y]$. The curve
analysis for $C$ provides a set of \emph{$x$-critical event lines} $x=\alpha$. Each $\alpha$ is represented as the
root of a square-free polynomial $r_i$, with $r_i$ a factor of
$R_C:=\res(f,f_y;y)$, together with an isolating interval
$I(\alpha)$. In addition, we have isolating intervals for the roots of
$f(\alpha,y)$. A corresponding result also holds for the curve $D$ with
$R_D:=\res(g,g_y;y)$. For the common intersection points of
$C$ and $D$, a similar representation is known. That is, we have critical
event lines $x=\alpha'$, where $\alpha'$ is a root of a square-free factor
of $R_{CD}:=\res(f,g;y)$ and, thus, $f(\alpha,y)$ and
$g(\alpha,y)$ share at least one common root (or the their leading
coefficients both vanish for $x=\alpha$). In addition, isolating intervals
for each of these roots have been computed. The curve-pair analysis now
essentially follows from merging this information. More precisely, we first
compute merged \emph{critical event lines} (via sorting the roots of $R_C$,
$R_D$ and $R_{CD}$) and, then, insert merged \emph{non-critical event
lines} at rational values $q_I$ in between. The intersections of $C$ and
$D$ with a non-critical event line at $x=q_I$ are easily computed via
isolating the roots of $f(q_I,y)$ and $g(q_I,y)$ and further refining the
isolating intervals until all isolating intervals are pairwise disjoint.
For a critical event line $x=\alpha$, we refine the already computed
isolating intervals for $f(\alpha,y)$ and $g(\alpha,y)$ until the number of
pairs of overlapping intervals matches the number $m$ of intersection
points of $C$ and $D$ above $\alpha$. This number is obtained from the
output of \bs applied to $f$ and $g$, restricted to $x=\alpha$. The
information on how to connect the lifted points is provided by the curve
analyses for~$C$ and~$D$. Note that efficiency is achieved by the fact that
\bs constitute (in its expensive parts) a local algorithm.

We remark that, in the previous approach by Eigenwillig and
Kerber~\cite{eigenwilligk08}, $m$ is also determined via efficient filter
methods, while, in general, a subresultant computation is needed if the
filters fail. This is, for instance, the case when two covertical
intersections of~$C$ and~$D$ occur. For our proposed lifting
algorithms, such situations are not more difficult, and thus do not particularly influence the runtime. 

\section{Speedups}
\label{sec:speedups}

\subsection{GPU-accelerated symbolic computations}
\label{ssec:speedups:gpu_res}

As mentioned in the introduction, one of the notable advantages of 
\emph{all} our new algorithms
over similar approaches is that it is not based on sophisticated symbolic
computations (such as, for example, evaluating signed remainder sequences)
restricting the latter ones to only computing
bivariate resultants and gcds of univariate polynomials.
In turn, these operations can be outsourced to the graphics hardware to 
dramatically reduce the overhead of symbolic arithmetic.
In this section, we overview the proposed GPU\footnote{Graphics Processing Unit} algorithms and refer to~\cite{emel-ica3pp-10,emel-pasco-10,emel-gcd-11} for further details. 

At the highest level, the resultant and gcd algorithms are 
based on a \emph{modular} or homomorphism approach, first exploited in the works
of Brown~\cite{brown-71} and Collins~\cite{collins-71}.
The modular approach is a traditional way to avoid computational problems, such as
\emph{expression swell}, shared by all symbolic algorithms.
In addition, it enables us to distribute the computation of one symbolic expression over a large number of processor cores of the graphics card.
Our choice of the target realization platform is not surprising because, 
with the released CUDA framework~\cite{CUDA-10}, the GPU has become a 
new standard for high-performance scientific computing.

To understand the main principles of GPU computing, we first need to have a look
at the GPU architecture. Observe that the parallelism on the graphics processor 
is supported on \emph{two} levels. 
At the upper level, there are numerous \emph{thread blocks}
executing concurrently without any synchronization between them.
There is a potentially unlimited number of thread blocks that 
can be scheduled for execution on the GPU. 
These blocks are then processed in a queued fashion by the hardware.
This realizes \emph{block-level} parallelism.
For its part, each thread block contains a limited number of 
parallel threads (up to $1024$ threads on the latest GPUs)
which can cooperate using on-chip shared memory and 
synchronize the execution with barriers. This is referred to as \emph{thread-level} parallelism.
An important point is that individual threads running on the GPU 
are ``lite-weight'' in a sense that they do not possess
large private memory spaces, neither they can execute disjoint code paths without penalties. 
The conclusion is that an algorithm to be realized on the graphics card
must exhibit a high homogeneity of computations such that individual
threads can perform the \emph{same} operations but on different data elements.
We start our overview with the resultant algorithm.

\paragraph{Computing resultants in $\Z[x,y]$}
Given two bivariate polynomials $f,g\in\Z[x,y]$, the modular resultant algorithm
of Collins can be summarized in the following steps:
\begin{compactenum}[(a)]
 \item apply a modular homomorphism to map the coefficients of $f$ and $g$ 
to a finite field for sufficiently many primes $p$: $\Z[x,y]\to \Z_p[x,y]$;
 \item for each modular image, choose a set of points $\alpha_p^{(i)}\in \Z_p$, $i\in I$, and evaluate the polynomials at $x=\alpha_p^{(i)}$ (evaluation homomorphism): 
$\Z_p[x,y]\to \Z_p[x,y]/(x-\alpha_p^{(i)})$;
\item compute a set of univariate resultants in $\Z_p[x]$
in parallel: $\res_y(f,g)|_{\alpha_p^{(i)}}: \Z_m[x,y]/(x-\alpha_p^{(i)})\to\Z_p[x]/(x-\alpha_p^{(i)})$;
\item interpolate the resultant polynomial for each prime $p$ 
in parallel: $\Z_p[x]/(x-\alpha_m{(i)})\to\Z_p[x]$;
\item lift the resultant coefficients by means of Chinese remaindering:
$\Z_p[x]\to\Z[x]$.
\end{compactenum}
Steps (a)--(d) and partly~(e) are outsourced to the graphics processor, thereby minimizing the amount of work on the host machine. In essence, what remains to be done on the CPU, is to convert the resultant coefficients in the mixed-radix representation
(computed by the GPU) to the standard form.

Suppose we have applied modular and evaluation homomorphisms to reduce the 
resultant of $f$ and $g$ to $N$ univariate resultants in $\Z_p[x]$ 
for each of $M$ moduli. Thus, provided that the modular images can be processed independently, we can launch a grid of $N\times M$ thread blocks with each block
computing the resultant of one modular image.
Next, to compute the univariate resultants, we employ a \emph{matrix-based}
approach instead of the classical PRS (polynomial remainder sequences)
used by Collins' algorithm. 
One of the advantages of this approach is that, when a problem is expressed in terms of linear algebra, all data dependencies are usually made explicit, thereby enabling thread-level parallelism which is a key factor in achieving high performance.

More precisely, the resultants of the modular images are computed by direct factorization of the Sylvester matrix using the so-called Schur algorithm which exploits the special structure of the matrix. In order to give an idea how this algorithm works, let 
$\tilde f,\tilde g\in\Z[x]$ be polynomials of degrees $m$ and $n$, respectively.
Then, for the associated Sylvester matrix $S\in\Z^{r\times r}$ ($r =
m+n$), one can write the following \emph{displacement} equation~\cite{displ-95}:
\begin{equation}\label{eq:displ-sylvester}
S - Z_r S (Z_m\oplus Z_n)^T = GB^T\mbox{,}
\end{equation}
\noindent  where $Z_s\in\Z^{s\times s}$ is a down-shift matrix zeroed everywhere
except for 1's on the first subdiagonal, $\oplus$ denotes the Kronecker sum, and 
$G, B\in\Z^{r\times 2}$ are the \emph{generator matrices} whose entries can be deduced from $S$ by inspection. For illustration, we can write~(\ref{eq:displ-sylvester}) in explicit form setting $m=4$ and $n=3$:
\[\footnotesize
\underbrace{\left[
\begin{array}{ccccccc}
f_4 & 0   & 0   & g_3 & 0 & 0 & 0 \\
f_3 & f_4 & 0   & g_2 & g_3 & 0 & 0\\
f_2 & f_3 & f_4 & g_1 & g_2 & g_3 & 0\\
f_1 & f_2 & f_3 & g_0 & g_1 & g_2 & g_3\\
f_0 & f_1 & f_2 & 0   & g_0 & g_1 & g_2\\
0   & f_0 & f_1 & 0   & 0 & g_0 & g_1\\
0   & 0   & f_0 & 0   & 0 & 0 & g_0\\
\end{array}\right]}_{S}
-\underbrace{\left[
\begin{array}{ccccccc}
0 & 0   & 0   & 0 & 0 & 0 & 0 \\
0 & f_4 & 0   & 0 & g_3 & 0 & 0\\
0 & f_3 & f_4 & 0 & g_2 & g_3 & 0\\
0 & f_2 & f_3 & 0 & g_1 & g_2 & g_3\\
0 & f_1 & f_2 & 0   & g_0 & g_1 & g_2\\
0   & f_0 & f_1 & 0   & 0 & g_0 & g_1\\
0   & 0   & f_0 & 0   & 0 & 0 & g_0\\
\end{array}\right]}_{Z_r S (Z_m\oplus Z_n)^T}=
\underbrace{\left[
\begin{array}{ccccccc}
f_4 & 0   & 0   & g_3 & 0 & 0 & 0 \\
f_3 & 0 & 0   & g_2 & 0 & 0 & 0\\
f_2 & 0 & 0 & g_1 & 0 & 0 & 0\\
f_1 & 0 & 0 & g_0 & 0 & 0 & 0\\
f_0 & 0 & 0 & 0   & 0 & 0 & 0\\
0   & 0 & 0 & 0   & 0 & 0 & 0\\
0   & 0   & 0 & 0   & 0 & 0 & 0\\
\end{array}\right]}_{GB^T}.
\]
The matrix on the right-hand side has rank $2$, and hence it can be
decomposed as a product of $r\times 2$ and $2\times r$ matrices $G$ and $B^T$.
The idea of the Schur algorithm is to rely on this low-rank displacement representation
of a matrix to compute its factorization in an asymptotically fast way.
Particularly, to factorize the matrix $S$, this algorithm only demands for $\mathcal{O}(r^2)$ 
operations in $\Z$; see~\cite[p.~323]{displ-95}.
In short, the Schur algorithm is an iterative procedure: In each step, it
transforms the matrix generators into a ``special form'' from which triangular
factors can easily be deduced based on the displacement equation~(\ref{eq:displ-sylvester}). Using
division-free modifications, this procedure 
can be performed efficiently in a prime field giving rise to the resultant
algorithm in $\Z_p[x]$; its pseudocode (serial version) can be found in~\cite[Section~4.2]{emel-pasco-10}. 
Now, to port this to the GPU, we assign one thread to one row of each of the generator matrices, that is, to four 
elements (because $G, B\in\Z^{r\times 2}$). 
In each iteration of the Schur algorithm, each thread updates its associated generator rows and multiplies them by a $2\times 2$ transformation matrix.
Altogether, a univariate resultant can be computed in $\mathcal{O}(r)$
finite field operations using $r$ processors (threads).
This explains the basic routine of the resultant algorithm.

The next step of the algorithm, namely polynomial interpolation in $\Z_p$,
can also be performed efficiently on the graphics card.
Here, we exploit the fact that interpolation is equivalent to solving a Vandermonde
system, where the Vandermonde matrix has a special structure.
Hence, we can again employ the Schur algorithm to solve the system in a small parallel time, see~\cite[Section~4.3]{emel-pasco-10}.
Finally, in order to obtain a solution in $\Z[x]$, we apply the Mixed-Radix Conversion (MRC) algorithm~\cite{yassine-91} which reconstructs the integer coefficients of the resultant in the form of mixed-radix (MR) digits. The key feature of this algorithm
is that it decouples operations in a finite field $\Z_p$ from those 
in the integer domain. In addition, the computation of MR digits 
can be arranged in a very structured way allowing for data-level parallelism
which can be readily exploited to compute the digits on the GPU.\\

\paragraph{Computing gcds in $\Z[x]$}
The modular gcd algorithm proposed by Brown follows a similar outline
as Collins' algorithm discussed above. For $f,g\in\Z[x]$, it consists of three steps:
\begin{compactenum}[(a)]
 \item apply modular homomorphism reducing the coefficients of $f$ and $g$ modulo sufficiently many primes: $\Z[x]\to \Z_p[x]$;
 \item compute a set of univariate gcds in $\Z_p[x]$:
$\gcd(f,g)\bmod p: \Z_p[x]\to \Z_p[x]$;
\item lift the coefficients of a gcd using Chinese remaindering: $\Z_m[x]\to\Z[x]$.
\end{compactenum}
Again, we augment the original Brown's algorithm by replacing the Euclidean scheme
(used to compute a gcd of each homomorphic image) with a matrix-based approach.
The univariate gcd computation is based on the following theorem.

\begin{theorem}~\cite{gcd-69}\label{thm:sylvester-gcd}
Let $S$ be the Sylvester matrix for polynomials $f,g\in\mathbb{F}[x]$
with coefficients over some field $\mathbb{F}$.
If $S$ is put in echelon form\footnote{A matrix is in echelon form if all nonzero rows are above any rows of all zeroes, and the leading coefficient of a nonzero row is always strictly to the right of the leading coefficient of the row above it.}, using row transformations only, then the last
non-zero row gives the coefficients of $\gcd(f,g)\in\mathbb{F}[x]$.
\end{theorem}
\noindent 
Suppose $f$ and $g$ have degrees $m$ and $n$, respectively.
Theorem~\ref{thm:sylvester-gcd} asserts that if we triangulate the Sylvester matrix $S\in\Z^{r\times r}$ ($r=n+m$), for instance, by means of Gaussian elimination, we obtain $\gcd(f,g)$ in the last nonzero row of the triangular factor. 
In order to achieve the latter, we apply the Schur algorithm to the positive-definite matrix $W=S^TS$ to obtain the orthogonal (QR) factorization
of $S$.\footnote{The reason why we do not triangularize $S$ directly is elaborated upon in~\cite{emel-gcd-11}.} In terms of displacements, 
$W$ can be written as follows~\cite{displ-95}:
\begin{equation}\label{eq:W-gcd-displ}
 W - Z_r W Z_r^T = GJG^T\mbox{ with }G\in\Z^{r\times 4},\ J=I_2\oplus-I_2.
\end{equation}
Here, $I_s$ denotes an $s\times s$ identity matrix. Remark that it is not necessary
to compute the entries of $W$ explicitly because the generator matrix $G$ is easily 
expressible in terms of the coefficients of $f$ and $g$, see~\cite[Section~2.2]{emel-gcd-11}.
Similarly to the resultants, we can run the Schur algorithm for $W$ in $\mathcal{O}(r)$ time
on the GPU using $r$ processors (threads). 
That is, one thread is assigned to process one row of the generator matrix $G$ ($4$ elements). The source code of a sequential algorithm can be found in~\cite[Algorithm~1]{emel-gcd-11}. 

From the theoretical perspective, the rest of the GPU algorithm essentially 
follows the same outline as the one for resultants, with the exception that there is no need for an interpolation step anymore since the polynomials are univariate. 
Certainly, there is also a number of practical difficulties
that need to be clarified. One of them is computing tight upper bounds on the height of a polynomial divisor which is needed to estimate the number of moduli used by the algorithm.\footnote{The height of a polynomial is defined as the maximal magnitude of its coefficients.} The existing theoretical bounds are very pessimistic,
and the original algorithm by Brown relies on trial division
to reduce the number of homomorphic images. However, this solution is
incompatible with parallel processing because the algorithm must be applied \emph{incrementally}. That is why, in the implementation, 
we use a number of heuristics to shrink the theoretical worst-case bounds.

Another challenge relates to the fact that it is not always possible to 
compute the gcd of a modular image by a single thread block (recall that the number of threads per block is limited) while threads from different blocks cannot work cooperatively. Thus, we needed to introduce some ``data redundancy'' 
to be able to distribute the computation of a single modular gcd (factorization of the Sylvester matrix) across numerous thread blocks. The details can be found in the paper cited above.

\subsection{Filters for \bs}
\label{ssec:speedups:bsfilters}

Besides the parallel computation of resultants and gcds,
the algorithm \bs to solve bivariate polynomial systems from 
Section~\ref{sec:bs} can be elaborated with a number of filtering techniques to early validate a majority of the candidates:

As first step, we group candidates along the same vertical line 
(a \emph{fiber})
at an $x$-coordinate~$\alpha$ (a root of $R^{(y)}$) to process them
together. This allows us to use extra information on the real roots 
of $f(\alpha,y)\in\RR[y]$ and $g(\alpha,y) \in \RR[y]$ for the validation of candidates.
\begin{figure*}[tb]
\centering\footnotesize\sffamily
\resizebox{!}{!}{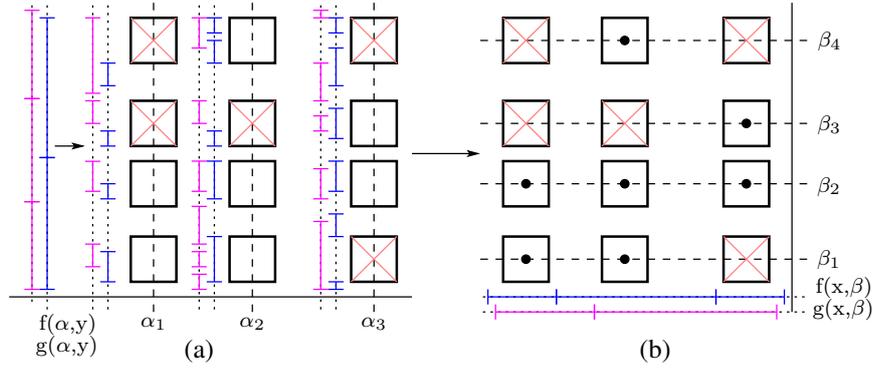} 
\caption{\textbf{(a)} Intervals containing the roots of $f(\alpha,y)$ and 
$g(\alpha,y)$ are refined until they either do not overlap or are fully 
included in candidate boxes. In the former case, the boxes can be 
discarded. \textbf{(b)} Unvalidated candidates are passed to 
\emph{bidirectional} filter
which runs bitstream isolation in another direction}
\label{fig:BS}
\end{figure*}
We replace the tests based on interval evaluation 
(see page~\pageref{pref:interval_exclusion}) by 
a test based on the \emph{bitstream Descartes}  
isolator~\cite{ek+-descartes} (\bdesc) (which has already been used in \slowlift; 
see Section~\ref{sssec:ca:alg:lift:slowlift}). To do so, we apply \bdesc to both polynomials
$f(\alpha,y)$ and $g(\alpha,y)$ in parallel, which eventually reports intervals
that do not share common roots.
This property is essential for our filtered version of \bsvalidate:
a candidate box $B(\alpha,\beta)$ can be rejected as soon as
the associated $y$-interval $I(\beta)$ fully overlaps with intervals rejected
by \bdesc for $f(\alpha,y)$ or $g(\alpha,y)$; see~Figure~\ref{fig:BS}~(a).

As alternative we could also deploy the numerical solver that is utilized in \fastlift; see
\ref{asec:numerical} for details. Namely, it can be modified in way to report \emph{active intervals}, and thus
allows us to discard candidates in non-active intervals. Even more, as
the numerical solver reports all (complex) solutions, we can use it
as inclusion predicate, too: If we see \emph{exactly one} overlap of reported
discs $\Delta_f$ and $\Delta_g$ (one for $f(\alpha,y)$, the other for $g(\alpha, y)$, respectively), and this overlap is completely contained in the projection $\Delta(\beta)$ of a candidate polydisc $\Delta(\alpha)\times\Delta(\beta)$, then $(\alpha,\beta)$ must be a solution. Namely, $f(\alpha,y)$ and $g(\alpha,y)$ share at least one common root, and each of these roots must be contained in $\Delta_f\cap\Delta_g$. By construction, $\Delta(\beta)$ contains at most one root, and thus $\beta$ must be the unique common root of the two polynomials.

Grouping candidates along a fiber $x=\alpha$ also enables us to
use \emph{combinatorial} tests to discard or to certify them.
First, when the number of certified solutions reaches
$\mathrm{mult}(\alpha)$,
the remaining candidates are automatically discarded because 
each real solution contributes at least once to the multiplicity of $\alpha$
as a root of $R^{(y)}$ (see Theorem~\ref{thm:resultants}). Second, if
$\alpha$ is not a root of the greatest common divisor $h^{(y)}(x)$ of the 
leading coefficients of $f$ and $g$, $\mathrm{mult}(\alpha)$ is odd, and all except one candidate along the fiber
are discarded, then the remaining candidate must be a real solution.
This is because complex roots come in conjugate pairs and, thus, 
do not change the parity of $\mathrm{mult}(\alpha)$.
We remark that, in case where the system (\ref{system}) is in
\emph{generic position} and the multiplicities of all roots of $R$ are odd,
the combinatorial test already suffices to certify all solutions 
without the need to apply the inclusion predicate based on Theorem~\ref{thm:inclusion}.

Now, suppose that, after the combinatorial test, there are several 
candidates left within a fiber. For instance, the latter can indicate 
the presence of \emph{covertical} solutions. 
In this case, before using the new inclusion predicate, 
we can apply the aforementioned filters in \emph{horizontal} direction as well.
More precisely, we construct the lists of unvalidated
candidates sharing the same $y$-coordinate $\beta$ and process them 
along a horizontal fiber. 
For this step, we initialize the bitstream trees (or the numerical solvers) 
for $f(x,\beta)\in\RR[x]$ and
$g(x,\beta)\in\RR[y]$ and proceed in exactly the same way as done for vertical 
fibers; see~Figure~\ref{fig:BS}~(b).
We will refer to this procedure as the \emph{bidirectional} filter, 
especially in Section~\ref{ssec:implex:bs}, where we examine the 
efficiency of all filters. The (few) candidates that still remain undecided after 
all filters are applied will be processed by
considering the new inclusion predicate.

\section{Implementation and experiments}
\label{sec:implex}

\paragraph{Setup}
We have implemented our algorithms in a branch of the bivariate
algebraic kernel first released with 
\cgal\footnote{The Computational Geometry Algorithms Library,
\url{www.cgal.org}.} version~3.7 in 
October~2010~\cite{cgal:bht-ak-11b,cgal:wfzh-a2-11b}.
\bs is a completely new implementation, whereas, for \ca and the analyses of pairs, we only replaced the lifting 
algorithms in \cgal's original curve- and curve-pair analyses\footnote{Note that those
and our algorithms have \caproject and \caconnect in common.} with our new 
methods based on \fastlift, \slowlift\footnote{
We remark, that the implementation of \slowlift can be improved:
each iteration of \bs can benefit from common factors that occur in the 
intermediate resultants, that is, for later iterations polynomials with smaller degree can be
considered.} and \bs.
As throughout \cgal, we follow the
\emph{generic programming paradigm} which allows us to choose among
various number types for polynomials' coefficients or intervals' boundaries 
and to choose the method used to isolate the real roots of univariate polynomials. For our setup, we rely on the integer and rational number types provided
by \gmp~5.0.1\footnote{\gmp:~\url{http://gmplib.org}}
and the highly efficient univariate 
solver based on the Descartes method contained in \rs\footnote{\rs:~\url{http://www.loria.fr/equipes/vegas/rs}} 
(by Fabrice~Rouillier~\cite{RS}),
which is also the basis for \isolate in Maple~13 and later versions.

All experiments have been conducted on a 2.8~GHz $8$-Core Intel Xeon W3530 
with 8~MB of L2~cache on a Linux platform. For the GPU-part of the 
algorithm, we have used the GeForce GTX580 graphics card (Fermi Core).

\paragraph{Symbolic Speedups} Our algorithms exclusively rely, as indicated,
on two symbolic operations, that is, resultant and $\gcd$ computation. We
outsource \emph{both} computations to the graphics hardware to reduce the
overhead of symbolic arithmetic which typically constitutes the main
bottleneck in previous approaches. Details about this have been covered in
Section~\ref{ssec:speedups:gpu_res}. Beyond that, it is worth noting that our
implementation of univariate $\gcd$s on the graphics card is comparable in
speed with the one from \ntl\footnote{A Library for Doing Number Theory, 
\url{http://www.shoup.net/ntl/}} running on the host machine. Our explanation
for this observation is that, in contrast to bivariate resultants, computing a
$\gcd$ of moderate degree univariate polynomials does not provide a
sufficient amount of parallelism, and \ntl's implementation is nearly
optimal. Moreover, the time for the initial modular reduction of polynomials,
still performed on the CPU, can become noticeably large, thereby neglecting
the efficiency of the GPU algorithm. Yet, we find it very promising to
perform the modular reduction directly on the GPU which should further
speed-up our algorithm.

\paragraph{Contestants}
For solving bivariate systems (Section~\ref{ssec:implex:bs}), 
we compared \bs to the bivariate
version of \isolate (based on \rs) and \lgp by Xiao-Shan Gao~\etal\footnote{\lgp:
\url{http://www.mmrc.iss.ac.cn/~xgao/software.html}} Both are interfaced
using Maple~14. We remark that, for the important substep of isolating the
real roots of the elimination polynomial, all three contestants in the ring
(including our implementation) use the highly efficient implementation provided by~\rs.

When analyzing algebraic curves (Section~\ref{ssec:implex:ca})
and computing arrangements of algebraic curves 
(Section~\ref{ssec:implex:arr}), we compared our new implementation 
with \cgal's bivariate algebraic kernel
(see \cite{bhk-ak2-2011} and \cite{cgal:bht-ak-11b}) that has shown 
excellent performance in exhaustive experiments 
over existing approaches, namely 
\textsc{cad2d}\footnote{\url{http://www.usna.edu/Users/cs/qepcad/B/QEPCAD.html}} and 
\textsc{Isotop}~\cite{gn-efficient} which is based on \rs.
These two other contestants were, except for few example instances, less efficient than 
\cgal's implementation, so that we omit further tests with them. 
Two further reasons can be given: Firstly, we enhanced \cgal{}'s kernel with 
GPU-supported resultants and $\gcd$s which makes it 
more competitive to existing software, but also to \ca. Still, 
slowdowns are observable for singular curves or curves in non-generic
position due to its need of subresultants sequences performed on the CPU. For such hard instances, our new algorithms particularly profit from the algorithmic design which avoids costly symbolic operations that can only be performed on the CPU. At this point, we also remark that, even if no GPU is available and all symbolic operations would be carried out solely on the CPU, \ca is still much faster for hard instances. 
Secondly, the contestants based on \rs require as subtask \rs to solve
the bivariate polynomial system $f = f_y = 0$ in the curve-analysis.
However, our experiments on bivariate system solving that we report in Section~\ref{ssec:implex:bs} show that \bs is at least competitive 
to the current version of \rs and even show in most cases
an excellent speed gain over \rs. However, it should not be concealed that \rs is currently getting a 
very promising polish which uses the computations of a rational univariate representations and 
modular arithmetic~\cite{blpr-rs3-eurocg-11}. Yacine Bouzidi \etal 
are about to submit a bivariate kernel based on the updated \rs to 
\cgal{} in the spirit of the existing univariate kernel based on \rs; 
see~\cite{penaranda-phd}.
We are looking forward to compare our analysis and the arrangement
computation with this upcoming approach.

\medskip

All test data sets that we consider in our experiments are available for download.\footnote{
\url{http://www.mpi-inf.mpg.de/departments/d1/projects/Geometry/TCS-SNC.zip}}

\subsection{Bivariate system solving}
\label{ssec:implex:bs}

\begin{table*}[p]
\centering\footnotesize\sffamily
\begin{minipage}[t]{\linewidth}
\centering
\begin{tabularx}{\linewidth}{ |>{\hsize=1.6\hsize}L| 
|>{\hsize=0.70\hsize}R 
|>{\hsize=0.9\hsize}R
|>{\hsize=1.5\hsize}R|
|>{\hsize=0.9\hsize\bfseries}R
|>{\hsize=0.9\hsize}R|
|>{\hsize=0.75\hsize}R
|>{\hsize=0.75\hsize}R|}
\hline\multicolumn{8}{|l|}{\multirow{2}{*}{(X)~special~curves (see Table \ref{tbl:app:special-desc} in~\ref{asec:instances} for descriptions)}}\\
\multicolumn{8}{|c|}{} \\
\hline
 & \textsc{BS} & \textsc{BS{\scriptsize +bstr}} & \textsc{BS{\scriptsize +bstr+comb}} & \textsc{BS{\scriptsize +all}} & \textsc{BS{\scriptsize +all}} & \textsc{Isolate} & \textsc{LGP} \\
 curve & \multicolumn{3}{c||}{GPU} & GPU & CPU & Maple & Maple \\
\hline
13\_sings\_9 & 2.13 & 1.84 & 1.48 & 0.97 & 1.65 & 341.93 & 2.81\\
FTT\_5\_4\_4 & 48.03 & 9.20 & 9.00 & 20.51 & 52.21 & 256.37 & 195.65\\
L4\_circles & 0.92 & 1.31 & 1.62 & 0.74 & 1.72 & 1.31 & 7.58\\
L6\_circles & 3.91 & 4.23 & 3.68 & 2.60 & 16.16 & 21.37 & 51.60\\
SA\_2\_4\_eps & 0.97 & 0.38 & 0.32 & 0.44 & 4.45 & 3.31 & 4.69\\
SA\_4\_4\_eps & 4.77 & 2.07 & 1.84 & 2.01 & 91.90 & 158.63 & 54.51\\
challenge\_12 & 21.54 & 5.33 & 5.44 & 7.35 & 18.90 & 44.02 & 37.07\\
challenge\_12\_1 & 84.63 & 12.50 & 12.50 & 19.17 & 72.57 & 351.62 & 277.68\\
compact\_surf & 12.42 & 3.45 & 3.29 & 4.06 & 12.18 & 871.95 & 12.00\\
cov\_sol\_20 & 28.18 & 24.05 & 18.82 & 5.77 & 16.57 & 532.41 & 171.62\\
curve24 & 85.91 & 87.92 & 13.93 & 8.22 & 25.36 & 86.04 & 37.94\\
curve\_issac & 2.39 & 2.72 & 2.25 & 0.88 & 1.82 & 29.80 & 3.29\\
cusps\_and\_flexes & 1.17 & 1.09 & 0.86 & 0.63 & 1.27 & 381.51 & 2.43\\
degree\_7\_surf & 29.92 & 13.14 & 11.92 & 7.74 & 90.50 & timeout & 131.25\\
dfold\_10\_6 & 3.30 & 2.68 & 2.73 & 1.55 & 17.85 & 3.35 & 3.76\\
grid\_deg\_10 & 2.49 & 2.37 & 1.30 & 1.20 & 2.49 & 111.20 & 2.64\\
huge\_cusp & 9.64 & 9.81 & 6.96 & 6.44 & 13.67 & timeout & 116.67\\
mignotte\_xy & t>600 & 584.75 & 252.94 & 243.16 & 310.13 & 564.05 & timeout\\
spider & 167.30 & 77.86 & 50.61 & 46.47 & 216.86 & timeout & timeout\\
swinnerton\_dyer & 28.39 & 19.70 & 18.92 & 5.28 & 24.38 & 71.14 & 27.92\\
ten\_circles & 4.62 & 4.19 & 4.13 & 1.33 & 3.74 & 5.77 & 4.96\\
\hline
\multicolumn{8}{|l|}{\multirow{2}{*}{(X)~pairs of special~curves (see Table \ref{tbl:app:special-desc} in~\ref{asec:instances} for descriptions)}}\\
\multicolumn{8}{|c|}{} \\
\hline
 & \textsc{BS} & \textsc{BS{\scriptsize +bstr}} & \textsc{BS{\scriptsize +bstr+comb}} & \textsc{BS{\scriptsize +all}} & \textsc{BS{\scriptsize +all}} & \textsc{Isolate} & \textsc{LGP} \\
 pair & \multicolumn{3}{c||}{GPU} & GPU & CPU & Maple & Maple \\
\hline
deg18\_7\_curves & 2.19 & 2.33 & 1.74 & 0.97 & 2.01 & 3.50 & 4.37\\
hard\_one & 11.34 & 10.13 & 6.46 & 4.29 & 82.53 & 64.50 & 17.45\\
large\_curves & 286.32 & 260.35 & 72.50 & 43.12 & 35.37 & 311.61 & 98.07\\
spiral29\_24 & 207.47 & 206.62 & 30.35 & 18.57 & 35.53 & 215.35 & 76.50\\
tryme & 64.77 & 65.55 & 22.67 & 18.61 & 48.21 & 397.41 & 107.80\\
vert\_lines & 0.60 & 0.61 & 0.63 & 0.47 & 0.69 & 5.79 & 1.20\\
\hline
\end{tabularx}
\end{minipage}
\caption{Running times (in seconds, including resultant computations) 
for solving bivariate system defined by special curves. 
\bs-GPU: our 
approach with GPU-resultants; \bs-CPU: our approach with \cgal's 
CPU-resultants; \isolate and \lgp use Maple's implementation for the resultant computation.
Bold face indicates the default setup for \bs; \textbf{timeout:} algorithm timed out ($>$~600~sec)}
\label{tbl:test1}
\end{table*}

\begin{table*}[p]
\centering\footnotesize\sffamily
\begin{minipage}[t]{\linewidth}
\centering
\begin{tabularx}{\linewidth}{ |>{\hsize=1.6\hsize}R| 
|>{\hsize=0.70\hsize}R 
|>{\hsize=0.9\hsize}R
|>{\hsize=1.5\hsize}R|
|>{\hsize=0.9\hsize\bfseries}R
|>{\hsize=0.9\hsize}R|
|>{\hsize=0.75\hsize}R
|>{\hsize=0.75\hsize}R|}
\hline
\multicolumn{8}{|l|}{\multirow{2}{*}{(R)~sets of five random dense curves}}\\
\multicolumn{8}{|c|}{} \\
\hline
 & \textsc{BS} & \textsc{BS{\scriptsize +bstr}} & \textsc{BS{\scriptsize +bstr+comb}} & \textsc{BS{\scriptsize +all}} & \textsc{BS{\scriptsize +all}} & \textsc{Isolate} & \textsc{LGP} \\
 degree, bits & \multicolumn{3}{c||}{GPU} & GPU & CPU & Maple & Maple \\
\hline
6, \hphantom{00}10 & 0.29 & 0.33 & 0.31 & 0.20 & 0.38 & 0.54 & 0.41\\
6, \hphantom{0}128 & 0.47 & 0.29 & 0.34 & 0.26 & 0.31 & 0.64 & 0.66\\
6, \hphantom{0}512 & 0.99 & 0.69 & 0.56 & 0.43 & 0.54 & 1.76 & 1.91\\
6, 2048 & 5.92 & 3.18 & 1.99 & 1.50 & 1.85 & 9.31 & 9.92\\
\hline
9, \hphantom{00}10 & 2.06 & 0.88 & 0.74 & 0.36 & 0.78 & 1.24 & 0.88\\
9, \hphantom{0}128  & 3.31 & 1.85 & 1.04 & 0.45 & 0.54 & 1.50 & 1.66\\
9, \hphantom{0}512 & 7.98 & 4.81 & 2.39 & 0.88 & 1.07 & 3.62 & 4.58\\
9, 2048 & 34.12 & 19.87 & 11.15 & 3.75 & 4.19 & 19.24 & 24.66\\
\hline
12, \hphantom{00}10 & 14.85 & 4.82 & 2.46 & 1.07 & 2.11 & 3.96 & 3.32\\
12, \hphantom{0}128  & 20.08 & 7.90 & 3.78 & 1.32 & 1.59 & 5.77 & 6.39\\
12, \hphantom{0}512 & 42.73 & 18.22 & 10.10 & 2.45 & 2.80 & 19.12 & 23.17\\
12, 2048 & 162.11 & 68.28 & 53.03 & 11.14 & 11.97 & 109.67 & 138.06\\
\hline
15, \hphantom{00}10 & 56.40 & 10.64 & 5.69 & 1.55 & 2.66 & 6.09 & 5.65\\
15, \hphantom{0}128  & 95.35 & 17.00 & 10.61 & 2.01 & 2.30 & 8.96 & 10.46\\
15, \hphantom{0}512 & 195.01 & 41.42 & 31.16 & 3.95 & 4.22 & 26.06 & 33.87\\
15, 2048 & timeout & 161.00 & 169.77 & 19.89 & 20.45 & 140.68 & 190.86\\
\hline
\multicolumn{8}{|l|}{\multirow{2}{*}{(R)~sets of five random sparse curves}}\\
\multicolumn{8}{|c|}{} \\
\hline
 & \textsc{BS} & \textsc{BS{\scriptsize +bstr}} & \textsc{BS{\scriptsize +bstr+comb}} & \textsc{BS{\scriptsize +all}} & \textsc{BS{\scriptsize +all}} & \textsc{Isolate} & \textsc{LGP} \\
 degree, bits & \multicolumn{3}{c||}{GPU} & GPU & CPU & Maple & Maple \\
\hline
6, \hphantom{00}10 & 0.11 & 0.10 & 0.16 & 0.10 & 0.13 & 0.19 & 0.14\\
6, \hphantom{0}128  & 0.28 & 0.12 & 0.14 & 0.11 & 0.15 & 0.23 & 0.21\\
6, \hphantom{0}512 & 0.50 & 0.32 & 0.24 & 0.20 & 0.21 & 0.48 & 0.47\\
6, 2048 & 3.32 & 1.28 & 0.65 & 0.58 & 0.68 & 2.12 & 2.15\\
\hline
9, \hphantom{00}10 & 0.20 & 0.52 & 0.27 & 0.18 & 0.24 & 0.39 & 0.31\\
9, \hphantom{0}128  & 0.45 & 0.92 & 0.33 & 0.22 & 0.25 & 0.51 & 0.52\\
9, \hphantom{0}512 & 1.21 & 1.82 & 0.54 & 0.37 & 0.40 & 1.44 & 1.49\\
9, 2048 & 7.52 & 11.02 & 1.96 & 1.21 & 1.38 & 7.44 & 8.42\\
\hline
12, \hphantom{00}10 & 0.51 & 0.72 & 0.55 & 0.28 & 0.38 & 0.65 & 0.53\\
12, \hphantom{0}128  & 1.49 & 1.61 & 0.75 & 0.36 & 0.36 & 1.08 & 1.11\\
12, \hphantom{0}512 & 5.17 & 5.75 & 1.67 & 0.66 & 0.69 & 3.61 & 3.83\\
12, 2048 & 47.19 & 42.35 & 7.98 & 2.70 & 2.75 & 21.25 & 23.89\\
\hline
15, \hphantom{00}10 & 3.66 & 3.33 & 2.11 & 1.00 & 1.39 & 2.48 & 2.25\\
15, \hphantom{0}128  & 12.14 & 6.37 & 3.35 & 1.25 & 1.35 & 4.17 & 4.27\\
15, \hphantom{0}512 & 43.36 & 19.93 & 8.52 & 2.40 & 2.54 & 13.95 & 15.48\\
15, 2048 & 408.90 & 150.49 & 44.34 & 10.97 & 10.98 & 78.65 & 89.35\\
\hline
\end{tabularx}
\end{minipage}
\caption{Total running times for solving five systems defined by random curves of increasing degree and with increasing bit-lengths. For description of configurations, see~Table~\ref{tbl:test1}.}
\label{tbl:test2}
\end{table*}

Our experiments for this task consist of two parts: 
In the first part, we consider ``special''
curves $C=V(f)$, and compute the $x$-critical points of $C$ (i.e.~the solutions of $f=f_y=0$). The curves are selected in order to challenge different parts of our algorithm (and also other algorithms), and in order to show the efficiency of the considered filtering techniques as given in
Section~\ref{ssec:speedups:bsfilters}. For instance, we considered curves with many singularities or high-curvature points which requires many candidates to 
be tested along each vertical line, or prohibit the use of special filters.
Table~\ref{tbl:test1} lists timings for various curves (described in Table~\ref{tbl:app:special-desc}).
In the second part of our experiments, we study the performance of
\bs on random polynomials with increasing total degrees and coefficient bit-lengths. We refer the reader to Table~\ref{tbl:test2} for the corresponding timings.

In columns $2$--$6$ of Table~\ref{tbl:test1} we see the performance of
\bs with all filters set off 
(\textsc{BS}), with bitstream filter enabled only (\textsc{BS+bstr}),
with bitstream and combinatorial filter (\textsc{BS+bstr+comb}) and
with all filters enabled (\textsc{BS+all}); the latter configuration
comes with and without the computation of symbolic operations on the GPU.
For the remaining configurations, we only show the timings using the GPU. 
The corresponding CPU-based timings can easily be obtained
by adding the (absolute) difference of the \textsc{BS+all}-columns.

One can observe that our algorithm is, in general, 
superior to \isolate and \lgp, even if the filters
are not used.
By comparing columns $2$--$6$ of Table~\ref{tbl:test1}, one can see that
filtering sometimes results in a
significant performance improvement. The \emph{combinatorial} test is
particularly useful when the defining polynomials of the system
(\ref{system})
have large degrees and/or large coefficient bit-length while, at the same
time, the number of
covertical or singular solutions is small compared to the total
number of candidates being checked.
The \emph{bidirectional} filter is advantageous when the system has
covertical solutions in one direction
(say along $y$-axis) which are \emph{not} cohorizontal.
This is essentially the case for \textsf{cov\_sol\_20},
\textsf{swinnerton\_dyer}, \textsf{ten\_circles} and \textsf{curve\_issac}.

Another strength of our approach relates to the fact that the amount of
symbolic operations is crucially reduced.
Hence, when the time for computing resultants is dominating, 
the GPU-based algorithm offers a speed-up by a typical factor of 
\textbf{$2$-$5$} (sometimes even more; see, in particular, \textsf{SA\_4\_4\_eps}, \textsf{degree\_7\_surf}, \textsf{hard\_one}) 
over the version with default resultant implementation. 
It is also worth mentioning that both \isolate and \lgp benefit from 
the \emph{fast resultant computation} available in Maple while 
\cgal's default resultant computation\footnote{Authors are 
indebted to \cgal developers working on resultants.} is generally 
\emph{much slower} than that of Maple. 

Table~\ref{tbl:test2} lists timings for experiments with random curves.
Each instance consists of five curves of the same degree 
(dense or sparse) and we report the total time to 
compute the solutions of five systems of the form $f = f_y = 0$.
In order to analyze the influence of the coefficients' bit-lengths, 
we multiplied each curve by $2^k$ with $k \in \{128, 512, 2048\}$ and 
increased the constant coefficient by one.
Since the latter operation constitutes only a small 
perturbation of the vanishing set of the input system, the number of solutions 
remains constant while the content of the polynomials' 
coefficients also stays trivial. 
We see that the bidirectional filtering is not of any advantage because the
system defined by random polynomials is unlikely to have covertical
solutions. However, in this case, most candidates are rejected by the
combinatorial check, thereby omitting a (more expensive) test 
based on Theorem~\ref{thm:inclusion}. This results in a clear speed-up 
over a ``non-filtered'' version.
Also, observe that, compared to its contestants, GPU-\bs is less vulnerable to increasing the bit-length of 
coefficients. We have also observed that, for our filtered versions, the time 
for the validation step is almost independent 
of the bit-lengths.

Further experiments on solving bivariate systems of interpolated, 
parameterized, translated or projected curves are listed
in~\ref{asec:bs}. In all these tests \bs outperforms \lgp and \isolate;
the CPU-only version of \bs is at least as efficient as the 
contestants, and often even faster.
We omit experiments to refine the solution boxes to certain precision
as this matches the efficiency of QIR due to the fact that we have 
algebraic descriptions for the solutions' $x$- and $y$-coordinates.

\subsection{Analysing curves}
\label{ssec:implex:ca}

We next present the experiments comparing the analyses of single algebraic
curves for different families of curves: 
\begin{inparaenum}
\item[(R)] random curves of various degree and bit-lengths of their coefficients,
\item[(I)] curves interpolated through points on a grid,
\item[(S)] curves in the two-dimensional parameter space of a sphere,
\item[(T)] curves that were constructed by multiplying a curve $f(x,y)$ with $f(x,y+1)$, 
such that each fiber has more than one critical point,
\item[(P)] projections of intersections of algebraic surfaces in 3D and, finally,
\item[(G)] sets of three generated curves of same degree: \begin{inparaenum}[(G.1)] \item bivariate polynomials
with random uniform coefficients (\textsf{non-singular}), \item projected
intersection curves of a random surface and its $z$-derivative
(\textsf{singular-$f$-$f_z$}), and \item projected intersection curves of two
independently chosen surfaces (\textsf{singular-$f$-$g$}) \end{inparaenum} 
\item[(X)] ``special'' curves of degrees up to 42 with many singularities or high-curvature points.
\end{inparaenum}
The random and special curves were already under consideration in Section~\ref{ssec:implex:bs} where we only computed their $x$-critical points.
All other curves are taken from~\cite[4.3]{kerber-phd}.
For the curve topology analysis, we consider five different
setups: 
\begin{compactenum}[(a)]
\item \textsc{BS+all} (i.e.~\bs with all filters enabled) which is, strictly speaking, not comparable with 
the curve-analysis as it only computes the solutions
of the system $f = f_y = 0$. 
Still, it is interesting to see that, for most instances, \ca outperforms
\bs though \bs one only solves a subproblem of the curve-analysis.
\item \textsc{Ak\_2} is the bivariate algebraic kernel shipped with \cgal~3.7
but with GPU-supported resultants and $\gcd$s.
\item \slowca that exclusively uses \slowlift for the fiber liftings.
\item \fastca that first applies a random shearing (with a low-bit shearing factor), and, then, exclusively uses \fastlift for lifting step.
\item \ca combines \fastlift and \slowlift in the fiber computations as discussed
in Section~\ref{sssec:ca:lift:algo}: It uses \fastlift first, and if it fails for a certain fiber after a certain number of iterations, \slowlift is considered for this fiber instead.\\
We remark that the global modular filter that checks whether \fastlift is successful for all fibers, is not yet in action. So far, this test has only been implemented within Maple. As expected, it performs very well, that is, the run-times are considerably less than that for the majority of steps in the curve analysis. 
\end{compactenum}
\ca is our default setting, and its running time also includes the timing for the fiber computations where \fastlift fails and \slowlift is applied instead.

\begin{table}[tp]
\sffamily\small
\centering
\begin{tabularx}{\linewidth}{ |>{\hsize=1.6\hsize}R| |>{\hsize=0.7\hsize}R| |>{\hsize=0.7\hsize}R||>{\hsize=1\hsize}R|>{\hsize=1\hsize}R| |>{\hsize=1\hsize\bfseries}R| } 

\hline
\multicolumn{6}{|l|}{\multirow{2}{*}{(R)~sets of five random curves}}\\
\multicolumn{6}{|c|}{} \\
\hline
type,\hfill degree, bits & \normalfont\textsc{BS{\scriptsize +all}} & \normalfont\textsc{Ak\_2} & \normalfont\slowca & \normalfont\fastca & \normalfont\ca \\
\hline
dense,\hfill 09, \hphantom{00}10 & 0.36 & 0.66 & 1.50 & 0.29 & 0.23\\
dense,\hfill 09, 2048 & 3.75 & 3.48 & 10.61 & 2.03 & 2.16\\
\hline
dense,\hfill 15, \hphantom{00}10 & 1.55 & 2.15 & 5.81 & 0.96 & 0.92\\
dense,\hfill 15, 2048 & 19.89 & 16.86 & 54.58 & 7.74 & 13.24\\
\hline
sparse,\hfill 09, \hphantom{00}10 & 0.18 & 1.05 & 0.54 & 0.20 & 0.11\\
sparse,\hfill 09, 2048 & 1.21 & 4.46 & 2.79 & 1.38 & 0.68\\
\hline
sparse,\hfill 15, \hphantom{00}10 & 1.00 & 3.37 & 3.03 & 0.71 & 0.59\\
sparse,\hfill 15, 2048 & 10.97 & 22.78 & 24.85 & 5.47 & 5.46\\

\hline
\multicolumn{6}{|l|}{\multirow{2}{*}{(I)~sets of five interpolated curves through points on a grid}}\\
\multicolumn{6}{|c|}{} \\
\hline
degree & \normalfont\textsc{BS{\scriptsize +all}} & \normalfont\textsc{Ak\_2} & \normalfont\slowca & \normalfont\fastca & \normalfont\ca \\
\hline

 9 & 3.70 & 4.98 & 9.49 & 1.59 & 2.37\\
12 & 23.09 & 27.56 & 57.91 & 12.37 & 13.61\\
15 & 214.54 & 160.36 & 451.29 & 69.20 & 114.63\\

\hline
\multicolumn{6}{|l|}{\multirow{2}{*}{(S)~sets of five parameterized curves on a sphere with 16bit-coefficients}}\\
\multicolumn{6}{|c|}{} \\
\hline
degree & \normalfont\textsc{BS{\scriptsize +all}} & \normalfont\textsc{Ak\_2} & \normalfont\slowca & \normalfont\fastca & \normalfont\ca \\
\hline
 6 & 3.00 & 12.62 & 16.12 & 1.97 & 1.98\\
 9 & 30.87 & 39.74 & 119.61 & 27.49 & 21.37\\

\hline
\multicolumn{6}{|l|}{\multirow{2}{*}{(T)~sets of five curves with a vertically translated copy}}\\
\multicolumn{6}{|c|}{} \\
\hline
degree & \normalfont\textsc{BS{\scriptsize +all}} & \normalfont\textsc{Ak\_2} & \normalfont\slowca & \normalfont\fastca & \normalfont\ca \\
\hline
 6 & 1.32 & 12.69 & 8.59 & 0.77 & 0.67\\
 9 & 5.05 & 134.75 & 27.93 & 5.39 & 2.23\\

\hline
\multicolumn{6}{|l|}{\multirow{2}{*}{(P)~projected~intersection~curve of
surfaces with 8bit-coefficients}}\\
\multicolumn{6}{|c|}{} \\
\hline
degree(s) & \normalfont\textsc{BS{\scriptsize +all}} & \normalfont\textsc{Ak\_2} & \normalfont\slowca & \normalfont\fastca & \normalfont\ca \\
\hline

$6\cdot6$ & 1.40 & 220.02 & 383.45 & 2.57 & 0.68\\
$8\cdot8$ & 21.86 & timeout & 117.57 & 19.56 & 6.17\\

\hline
\multicolumn{6}{|l|}{\multirow{2}{*}{(G)~random~singular~and~non-singular~curves}}\\
\multicolumn{6}{|c|}{} \\
\hline
type \hfill degree, bits & \normalfont\textsc{BS{\scriptsize +all}} & \normalfont\textsc{Ak\_2} & \normalfont\slowca & \normalfont\fastca & \normalfont\ca \\
\hline
non-singular            \hfill 42,  237 &56.57  &40.66  &133.12 &23.27  &35.80\\        
singular-$f$-$f_z$      \hfill 42,  238 &64.24  &timeout    &372.99 &52.27  &25.50\\      
singular-$f$-$g$        \hfill 42,  237 &122.20 &timeout    &419.16 &39.55  &18.77\\     
\hline

\hline
\multicolumn{6}{|l|}{\multirow{2}{*}{(X)~special~curves (see Table \ref{tbl:app:special-desc} in~\ref{asec:instances} for descriptions)}}\\
\multicolumn{6}{|c|}{} \\
\hline
curve & \normalfont\textsc{BS{\scriptsize +all}} & \normalfont\textsc{Ak\_2} & \normalfont\slowca & \normalfont\fastca & \normalfont\ca \\
\hline

L6\_circles & 2.60 & 171.86 & 108.46 & 1.61 & 1.62\\
SA\_4\_4\_eps & 2.01 & 122.30 & 11.96 & 3.92 & 2.00\\
challenge\_12 & 7.35 & timeout & 16.11 & 64.75 & 12.50\\
compact\_surf & 4.06 & 81.56 & 19.66 & 7.43 & 5.31\\
cov\_sol\_20 & 5.77 & 43.40 & 14.06 & 4.22 & 2.41\\
degree\_7\_surf & 7.74 & timeout & 57.41 & 6.23 & 4.19\\
dfold\_10\_6 & 1.55 & 35.40 & 10.74 & 8.97 & 0.90\\
mignotte\_xy & 243.16 & timeout & 276.89 & 199.59 & 128.05\\
spider & 46.47 & timeout & 200.61 & 22.34 & 21.03\\
swinnerton\_dyer & 5.28 & 347.28 & 43.78 & 13.04 & 6.97\\
ten\_circles & 1.33 & 22.77 & 11.84 & 4.26 & 0.86\\

\hline
\end{tabularx}
\caption{Running times (in sec) for analyses of algebraic curves of various families; \textbf{timeout:} algorithm timed out ($>$~600~sec)}
\label{tbl:cana}
\end{table}

Table~\ref{tbl:cana} lists the running times for single-curve analyses.
We only give the results for representative examples; full tables
are listed in \ref{asec:ca}.
From our experiments, we conclude that \ca is, in general,
superior to the existing kernel, even though \cgal's original implementation 
now profits from GPU-accelerated resultants and gcds.
Moreover, while the speed-up for curves in generic position is already considerable 
(about half of the time), it becomes even more impressive for projected intersection curves
of surfaces and ``special'' curves with many singularities.
The reason for this tremendous speed-up is that, for singular 
curves, \textsc{Ak\_2}'s performance drops significantly with the degree
of the curve when the time to compute subresultants on the CPU becomes dominating. In addition, for curves in 
non-generic position, the efficiency of \textsc{Ak\_2} is affected because 
a coordinate transformation has to be considered in these cases.

Recall that \fastlift in \ca fails for very few instances, where \slowlift is locally used to treat some of the $x$-critical fibers instead. 
The switch to the backup method is observable in timings; see
for instance, \textsf{challenge\_12}.
Namely, the difference of the running times between \ca and \slowca are considerably less than the difference which can usually be observed for instances where the filter method succeeds for all fibers. 
In these cases, the numerical solver cannot isolate the roots
within a given number of iterations, 
or we indeed have $n_{\alpha} < n^+_{\alpha}$ for some fibers $x=\alpha$; see~Section~\ref{sssec:ca:lift:fastlift}.
Nevertheless, the running times are still very promising
and yet perform much better than \textsc{Ak\_2} for non-generic input,
even though \slowlift's implementation is not yet optimized,
and we anticipate a further performance improvement.

Similar as \textsc{Ak\_2} has improved on previous approaches when it was presented in 2008, our
new methods improve on \textsc{Ak\_2} now. That is, for random, interpolated 
and parameterized curves, the speed gain is noticeable, while
for translated curves and projected intersections, we improve the more the higher the degrees. On some curves of large degree(!), we improve by a 
factor up to 250 and more.

We also recommend \ca over \fastca since it gives full geometric information
at basically no additional cost; that is, for random instances, both are similarly efficient whereas, for non-random input, the winner is often determined by
the geometry of the curve. For instance, the projection step in \fastca is faster than that of \ca for random and interpolated curves. 
We cannot fully explain this observation, but we guess that the initial shearing results in a better separation of the resultant's
roots which makes the real root isolation cheaper. On the other hand, for
curves with many covertical critical points (e.g.~\textsf{challenge\_12}), shearing
yields a resultant which decomposes into less but more complex factors, which implies much higher
cost to isolate the roots of the resultant polynomial. In addition, 
we have to consider more $x$-critical fibers, and \fastlift also has to deal with larger bitlengths. In summary, we propose to not consider a shearing because, from our experiments, we can say that the increased cost are higher than the cost for the few needed runs of \slowlift, when \ca analyses the curve in the original coordinate system.
s
Unlike existing algorithms, \ca exhibits a very robust behavior on singular
inputs. In contrast, it often performs even better on singular instances
than on non-singular curves which have the same input size. This behavior
can be read off in detail from Table \ref{tbl:app:special} in \ref{asec:ca}.
where we compare curves of same degree
without and with singularities. For large instances, 
\ca noticeably outperforms the other contestants and
actually even benefits from singularities. We suspect that this behavior is
due to the fact that the resultant splits into many simple factors. Namely,
in this case, root isolation of the resultant becomes less costly than in
the non-singular case, where the resultant does not yield such a strong
factorization.

The drastically improved analyses of algebraic curves has also some impact
on the performance for analyzing algebraic surfaces. The approach 
in~\cite{bks-efficient} is crucially based on the analysis of the projected silhouette curve of the surface $f(x,y,z)=0$ (i.e.~$\operatorname{res}(f,f_z;z)=0$). The latter analysis turns out to be the main bottleneck using \cgal's algebraic kernel (\textsc{AK\_2}; see column~3 in 
Table~\ref{tbl:cana}). In particular, for projected intersection curves of two surfaces, \ca behaves drastically (typically by a factor 100 and more) better than \textsc{AK\_2}. Hence, we claim that the maximal reasonable degree
of surfaces that can be analyzed using the approach from~\cite{bks-efficient} grows from approximately $5-6$ to $8-10$.

\subsection{Computing arrangements}
\label{ssec:implex:arr}

For arrangements of algebraic curves, we compare two implementations:
\begin{compactenum}[(A)]
 \item \textsc{Ak\_2} is \cgal{}'s bivariate algebraic kernel shipped with
\cgal~3.7 but with GPU-supported resultants and gcds. 
 \item \kernelnt is the same but uses \ca to analyze single algebraic
curves. For the curve pair analyses, \kernelnt exploits \textsc{Ak\_2}'s
functionality whenever subresultant computations are not needed (i.e.~a
unique transversal intersection of two curves along a critical event line).
For more difficult situations (i.e.~two covertical intersections or a
tangential intersection), the curve pair analysis uses \bs as explained in
Section~\ref{sec:arr}.
\end{compactenum}
Our testbed consists of sets of curves from different families:
\begin{inparaenum}
\item[(F)] random rational functions of various degree
\item[(C)] random circles
\item[(E)] random ellipses
\item[(R)] random curves of various degree and coefficient bit-length
\item[(P)] sets of projected intersection curves of algebraic surfaces, and, finally,
\item[(X)] combinations of ``special'' curves.
\end{inparaenum}

\begin{table}[hbtp]
\sffamily\small
\centering
\begin{tabularx}{0.6\linewidth}{ |>{\hsize=1.0\hsize}R| |>{}R| |>{\hsize=1.0\hsize\bfseries}R| } 

\hline
\multicolumn{3}{|l|}{\multirow{2}{*}{(P)~increasing number of projected surface intersections}}\\
\multicolumn{3}{|c|}{} \\
\hline
\#resultants & \normalfont\textsc{Ak\_2} & \normalfont\kernelnt\\
\hline
 2   &0.49 & 0.21 \\
 3   &0.93 & 0.48 \\
 4   &1.64 & 1.03 \\
 5   &3.92 & 2.44 \\
 6   &7.84 & 5.14 \\
 7   &21.70&13.65 \\
 8   &35.77&22.69 \\
 9   &67.00&41.53 \\
10   &91.84&58.37 \\

\hline
\multicolumn{3}{|l|}{\multirow{2}{*}{(X)~combinations of special curves}}\\
\multicolumn{3}{|c|}{} \\
\hline
\#curves & \normalfont\textsc{Ak\_2} & \normalfont\kernelnt\\
\hline
 2   &81.93  &9.2    \\
 3   &148.46 &25.18  \\
 4   &730.57 &248.87 \\
 5   &836.43 &323.42 \\
 6   &3030.27&689.39 \\
 7   &3313.27&757.94 \\
 8   &timeout&1129.98\\
 9   &timeout&1166.17\\
10   &timeout&1201.34\\
11   &timeout&2696.15\\

\hline
\end{tabularx}
\caption{Running times (in sec) for computing arrangements of algebraic curves; \textbf{timeout:} algorithm timed out ($>$~4000~sec)}
\label{tbl:arr}
\end{table}

We skip the tables for rational functions, circles, ellipses and
random curves because the performance of both
contestants are more or less equal: The \emph{linearly} many 
curve-analyses are simple and, for the \emph{quadratic} number of curve-pair
analyses, there are typically no multiple intersections along a fiber,
that is, \bs is not triggered. Thus, the execution paths of
both implementations are almost identical, but only as we enhanced 
\textsc{Ak\_2} with GPU-enabled resultants and $\gcd$s. In addition, we also do not expect the need of a shear for such curves, thus, the behavior is
anticipated. The picture changes for projected intersection curves of 
surfaces and combinations of special curves whose running times are 
reported in~Table~\ref{tbl:arr}. The \textsc{Ak\_2} requires for both sets 
expensive subresultants to analyze single curves and to compute covertical 
intersections, while \kernelnt's performance is crucially less affected in such situations.




\section{Summary and Outlook}
\label{sec:conclusion}

We presented new algorithms to exactly compute with algebraic curves. By combining methods from different fields, we have been able to considerably reduce the amount of purely symbolic operations, and to outsource the remaining ones to graphics hardware. The majority of all computation steps is exclusively based on certified approximate arithmetic. As a result, our new algorithms are not only faster than existing methods but also capable to handle geometric difficult instances at least as fast as seemingly easy ones. We believe that, with respect to efficiency, there is a good chance that exact and complete methods can compete with purely numerical approaches which do not come with any additional guarantee. The presented experiments seem to affirm this claim.

We are confident that our new approach will also have some positive impacts in the following respect: There exist several non-certified (or non-complete) approaches either based on subdivision~\cite{mourrain-sub,bcgy-complete,mp-subdivision, PV04, Snyder92, SnyderK92} or homotopy methods~\cite{homotopy07}. They show very good 
behavior for most inputs. However, in order to guarantee exactness for all possible inputs (e.g.~singular curves), additional certification steps (e.g.~worst case separation bounds for subdivision methods) have to be considered, an approach which has not shown to be effective in practice so far.
An advantage of the latter methods, compared to elimination approaches, 
is that they are local and do not need (global) algebraic operations. It seems reasonable that combining our algorithm with a subdivision or homotopy 
approach eventually leads to a certified and \emph{complete} method which shows excellent ``local'' behavior as well. 

We further see numerous applications of our methods, in particular, when computing arrangements of surfaces. The actual implementation~\cite{bks-efficient} for surface triangulation is crucially based on planar arrangement computations of singular curves. Thus, we are confident that its efficiency can be considerably improved by using the new algorithm for planar arrangement computation.
In addition, it would be interesting to extend 
our algorithm \bs to the task of solving a polynomial system 
of higher dimensions.

The bit complexity analysis of \bs as presented in~\cite{es-bisolvecomplexity-11} hints to the fact that the total cost of \bs 
is dominated by the root isolation step for the elimination polynomial, and, for many instances, our experiments also confirm the latter claim. We aim to provide a proof for this behavior by means of a bit complexity analysis for \ca as well.

Finally, we remark that \textsc{Ak\_2} has been integrated into a webdemo~\cite{webdemo-08} which has already been used by numerous parties of interest. Certainly, we aim to update this webdemo by integrating the new algorithms from \kernelnt sinstead  

\section*{Acknowledgments}
\label{sec:acks}
Without Michael Kerber's careful implementation of the bivariate 
kernel in \cgal~\cite{cgal:wfzh-a2-11b}, this work would not have been 
realizable in a reasonable time. We would like to use the opportunity to 
thank Michael for his excellent work. Additionally, his comments on 
prior versions of the work were highly appreciated.
A special thank goes to all anonymous reviewers of the underlying 
conference submissions for their constructive and detailed criticism 
that have helped to improve the quality and exposition 
of this contribution.\newpage

{
\bibliographystyle{elsarticle-num}
\bibliography{bib,bib2,num,cgal_3.9}
}

\newpage

\begin{appendix}

\section{Numerical Solver with Certificate}
\label{asec:numerical}

\newcommand{\abs}[1]{\ensuremath{\left|#1\right|}}%
\newcommand{\PN}{\ensuremath{\operatorname{\text{\textit{PN}}}}}%
\newcommand{\RN}{\ensuremath{\operatorname{\text{\textit{RN}}}}}%
In \fastlift (see Section~\ref{sssec:ca:lift:fastlift}), we deploy a certified numerical solver for a fiber polynomial to find regions certified to contain its complex roots.
Bini and Fiorentino presented a highly efficient solution to this problem in their \textsc{MPSolve} package \cite{BF00-design}.
However, the interface of \textsc{MPSolve} only allows root isolation for polynomials with arbitrary, but fixed, precision coefficients.
Our solver adapts their approach in a way suited to also handle the case where the coefficients are not known a priori, but rather in an intermediate representation which can be evaluated to any arbitrary finite precision.
In particular, this applies in the setting of \fastlift, where the input features algebraic coefficients, represented as refineable isolating intervals of integer polynomials.

The description given in this section is rather high-level, and chosen to cover the specific application \fastlift.
For the details of an efficient implementation, we refer the reader to \cite{Kobel11}.
Let $g(z) := f(\alpha, z) = \sum_{i=0}^n g_i z^i \in \RR[z]$ be a fiber polynomial at an $x$-critical value $\alpha$ and $V(g) = \{ \zeta_i \}$, $i = 1, \dots, n,$ its complex roots. Thus, $g(z) = g_n \allowbreak \prod_{i=1}^n \allowbreak (z - \zeta_i)$.

Our numerical solver is based on the Aberth-Ehrlich iteration for simultaneous root finding.
Starting from arbitrary distinct root guesses $(z_i)_{i=1,\dots,n}$, it is given by the component-wise iteration rule $z'_i = z_i$ if $g(z_i) = 0$, and
\begin{equation*}
  z'_i = z_i - \frac{g(z_i) / g'(z_i)}{1 - g(z_i) / g'(z_i) \cdot \sum_{j \ne i} \frac{1}{z_i - z_j}}
\end{equation*}
otherwise.
As soon as the approximation vector $(z_i)_i$ lies in a sufficiently small neighborhood of some permutation of the actual roots $(\zeta_i)_i$ of $g$, this iteration converges with cubic order \cite{Tilli98} to simple roots.
For roots of higher multiplicity or clustered roots, we use a variant of Newton's method to achieve quadratic convergence as an intermediate step between the Aberth-Ehrlich iterations.
In practice, this combination shows excellent performance even if started with an arbitrary configuration of initial root guesses far away from the solutions.

A straight-forward implementation of the Aberth-Ehrlich method in arbitrary-precision arithmetic requires the coefficients $g_i$ of $g$ to be known up to some relative precision $p$, that is, the input is a polynomial $\tilde{g} = \sum \tilde{g_i} x^i$ whose floating point coefficients satisfy $\abs{\tilde{g_i}-g_i} \le 2^{-p} \abs{g_i}$.
In particular, this requirement implies that we have to decide in advance whether a coefficient vanishes.
However, in our application, a critical $x$-coordinate $\alpha$ of a fiber polynomial is not necessarily rational, and so are the coefficients of $g$.
Thus, the restriction on the coefficients translates to expensive symbolic $\gcd$ computations of the resultant and the coefficients of the defining polynomial $f$ of the curve, considered as a univariate polynomial in $\Z[y][x]$.

Instead, we work on a \emph{Bitstream interval representation} $[g]^\mu$ of $g$ (see \cite{ek+-descartes,Kobel11}).
Its coefficients are interval approximations of the coefficients of $g$, where we require the width $|g_i^+ - g_i^-|$ of each coefficient $[g]^\mu_i = [g_i^-, g_i^+]$ to be $\le \mu$ for a certain \emph{absolute} precision $\mu.$
Thus, in contrast to earlier implementations, we have to decide whether $g_i = 0$ for the leading coefficient only.
$[g]^\mu$ represents the set $\{ \tilde{g} : \tilde{g_i} \in [g]^\mu_i \}$ of polynomials in a \emph{$\mu$-polynomial neighborhood} of $g$; in particular, $g$ itself is contained in $[g]^\mu$.
Naturally, for the interval boundaries, we consider dyadic floating point numbers (\emph{bigfloats}).
Note that we can easily compute arbitrarily good Bitstream representations of $f(\alpha, z)$ by approximating $\alpha$ to an arbitrary small error, for example using the quadratic interval refinement technique \cite{abbott-qir-06}.

Starting with some precision (say, $\mu = 2^{-53}$) and a vector of initial approximations, we perform Aberth's iteration on some representative $\tilde{g} \in [g]^\mu$.
The natural choice is the \emph{median polynomial} with $\tilde{g_i} = (g_i^- + g_i^+)/2$, but we take the liberty to select other candidates in case of numerical singularities in Aberth's rule (most notably, if $\tilde{g}'(z_i) = 0$
in some iteration).

After a finite number of iterations (depending on the degree of $g$), we interrupt the iteration and check whether the current approximation state already captures the structure of $V(g)$.
We use the following result by Neumaier and Rump \cite{Rump03}, founded in the conceptually similar Wei\-er\-stra\ss-Du\-rand-Ker\-ner simultaneous root iteration:
\begin{lemma}[Neumaier]
  Let $g(z) = g_n \prod_{i=1}^n (z - \zeta_i) \in \C[z]$, $g_n \ne 0$.
  Let $z_i \in \C$ for $i = 1, \dots, n$ be pairwise distinct root approximations.
  Then, all roots of $g$ belong to the union $\mathcal{D}$ of the discs
  \begin{gather*}
    D_i := D(z_i - r_i, |r_i|),\\
    \text{where } r_i := \frac{n}{2} \cdot \frac{\omega_i}{g_n} \text{ and } \omega_i := \frac{g(z_i)}{\prod\nolimits_{j \ne i} (z_i - z_j)}.
  \end{gather*}
  Moreover, every connected component $C$ of $\mathcal{D}$ consisting of $m$ discs contains \emph{exactly} $m$ zeros of $g$, counted with multiplicity.
\end{lemma}
The above lemma applied to $[g]^\mu$ using conservative interval arithmetic yields a superset $\mathcal{C} = \{ C_1, \dots, C_m \}$ of regions and corresponding multiplicities $\lambda_1, \dots, \lambda_m$ such that, for each $C_k \in \mathcal{C}$, all polynomials $\tilde{g} \in [g]^\mu$ (and, in particular, $g$) have exactly $\lambda_k$ roots in $C_k$ counted with multiplicities.
Furthermore, once the quality of the approximations $(z_i)_i$ and $[g]^\mu$ is sufficiently high, $\mathcal{C}$ converges to $V(g)$.

In \fastlift, where we aim to isolate the roots of $g:=f(\alpha,y)$, we check whether $m=m_\alpha = n_{\alpha}^{+}$.
If the latter equality holds, we are guaranteed that the regions $C_k \in \mathcal{C}$ are isolating for the roots of $g$, and we stop.
Otherwise, we repeat Aberth's iteration after checking whether $0 \in [g]^\mu(z_i)$.
Informally, if this holds the quality of the root guess is not distinguishable from any (possibly better) guess within the current interval approximation of $g$, and we double the precision ($\mu' = \mu^2$) for the next stage.

Aberth's iteration lacks a proof for convergence in the general case and, thus, cannot be considered complete.
However, we feel this is a purely theoretical issue: to the best of our knowledge, only artificially constructed, highly degenerate configurations of initial approximations render the algorithm to fail.
In our extensive experiments, this situation never occurred.
From a theoretical point of view, it is possible to enhance the Aberth-Ehrlich method by a complete complex solver as a fallback method to ensure convergence of the root isolation.
E.g., the \textsc{CEval} subdivision solver \cite{SY09,Kamath10} can be extended to handle bitstream coefficients by employing perturbation bound techniques \cite{s-bitstream-mcs11}.

We note that regardless of this restriction, the regions $C_k \in \mathcal{C}$ are certified to comprise the roots of $g$ at any stage of the algorithm by Neumaier's lemma and the rigorous use of interval arithmetic.
In particular, the correctness of \fastlift and, thus, the completeness of the filtered curve analysis \ca is not affected.

\newpage
\section{Description of Special Curves}
\label{asec:instances}

\begin{table}[H]
  \centering\sffamily\small
  \begin{minipage}[t]{\linewidth}
    \centering
    \begin{tabularx}{\linewidth}{|l|r|L|}
      \hline
      Single curve & $\deg_y$ & Description\\
      \hline
      13\_sings\_9      &     9 & large coefficients; high-curvature points\\
      FTT\_5\_4\_4*     &    40 & many non-rational singularities\\
      L4\_circles       &    16 & 4 circles w.r.t.\ L4-norm; clustered solutions\\
      L6\_circles       &    32 & 4 circles w.r.t.\ L6-norm; clustered solutions\\
      SA\_2\_4\_eps*    &    16 & singular points with high tangencies, displaced\\
      SA\_4\_4\_eps*    &    33 & singular points with high tangencies, displaced\\ 
      challenge\_12*    &    30 & many candidate solutions to check\\
      challenge\_12\_1* &    40 & many candidates to be check\\
      compact\_surf     &    18 & silhouette of an algebraic surface; many singularities and isolated solutions\\
      cov\_sol\_20      &    20 & covertical solutions\\
      curve24           &    24 & curvature of degree 8 curve; many singularities\\
      curve\_issac      &    15 & isolated points, high-curvature points \cite{LGP-09}\\
      cusps\_and\_flexes&     9 & high-curvature points\\
      degree\_7\_surf   &    42 & silhouette of an algebraic surface; covertical solutions in $x$ and $y$\\ 
      dfold\_10\_6*     &    30 & many half-branches\\
      grid\_deg\_10     &    10 & large coefficients; curve in generic position\\
      huge\_cusp        &     8 & large coefficients; high-curvature points\\
      mignotte\_xy      &    42 & a product of $x$/$y$-Mignotte polynomials, displaced; many clustered solutions\\
      spider            &    12 & degenerate curve; many clustered solutions\\
      swinnerton\_dyer       &    25 & covertical solutions in $x$ and $y$\\ 
      ten\_circles      &    20 & set of 10 random circles multiplied together; rational solutions\\
\hline
\hline
      Pairs of curves & $\deg_y$ & Description\\
\hline
      deg18\_7\_curves  & 18, 7 & higher-order singularities on both curves\\
      hard\_one         &  27, 6 & vertical lines as components of one curve; many candidates to check\\
      large\_curves     &   24, 19    & large number of solutions\\
      spiral29\_24      & 29, 24 & Taylor expansion of a spiral intersecting a curve with many branches;\newline many candidates to check\\
      tryme             &   24, 34    & covertical solutions; many candidates to check\\
      vert\_lines       & 16, 6 & high-order singularity on one curve, many intersections\\
%
      \hline
    \end{tabularx}
  \end{minipage}
\caption{Description of the curves used in the first part of experiments. 
In case only a single curve given, the second curve is taken to be the 
first derivative w.r.t.\ $y$-variable. Curves marked with a star (*) are given in~\cite{labs_10}.}\
\label{tbl:app:special-desc}
\end{table}

\newpage
\section{Further experiments for bivariate system solving}
\label{asec:bs}

\begin{table}[H]
\centering\footnotesize\sffamily
\begin{minipage}[t]{\linewidth}
\centering
\begin{tabularx}{\linewidth}{ |>{\hsize=1.6\hsize}R| 
|>{\hsize=0.70\hsize}R 
|>{\hsize=0.9\hsize}R
|>{\hsize=1.5\hsize}R|
|>{\hsize=0.9\hsize\bfseries}R
|>{\hsize=0.9\hsize}R|
|>{\hsize=0.75\hsize}R
|>{\hsize=0.75\hsize}R|}
\hline
\multicolumn{8}{|l|}{\multirow{2}{*}{(I)~sets of five interpolated curves through points on a grid}}\\
\multicolumn{8}{|c|}{} \\
\hline
 & \textsc{BS} & \textsc{BS{\scriptsize +bstr}} & \textsc{BS{\scriptsize +bstr+comb}} & \textsc{BS{\scriptsize +all}} & \textsc{BS{\scriptsize +all}} & \textsc{Isolate} & \textsc{LGP} \\
 degree & \multicolumn{3}{c||}{GPU} & GPU & CPU & Maple & Maple \\
\hline
5 & 0.29 & 0.17 & 0.32 & 0.27 & 0.38 & 0.59 & 0.51\\
6 & 1.20 & 0.50 & 0.67 & 0.59 & 0.71 & 1.07 & 1.12\\
7 & 4.52 & 1.79 & 1.35 & 1.16 & 1.37 & 2.08 & 2.32\\
8 & 14.86 & 3.63 & 2.55 & 1.98 & 2.51 & 3.82 & 4.20\\
9 & 63.46 & 7.33 & 5.19 & 3.70 & 4.50 & 7.17 & 7.99\\
10 & 194.04 & 13.14 & 8.96 & 5.46 & 6.71 & 12.44 & 13.76\\
11 & timeout & 25.11 & 19.59 & 10.94 & 12.31 & 24.82 & 28.61\\
12 & timeout & 44.84 & 41.88 & 23.09 & 25.23 & 50.54 & 55.56\\
13 & timeout & 80.44 & 84.29 & 45.54 & 49.92 & 98.92 & 110.02\\
14 & timeout & 138.13 & 191.25 & 101.96 & 103.91 & 182.72 & 205.26\\
15 & timeout & 225.39 & 376.17 & 214.54 & 219.39 & 371.25 & 399.64\\
16 & timeout & 367.85 & timeout & 410.46 & 427.50 & timeout & timeout\\
\hline
\multicolumn{8}{|l|}{\multirow{2}{*}{(S)~sets of five parameterized curves on a sphere with 16bit-coefficients}}\\
\multicolumn{8}{|c|}{} \\
\hline
 & \textsc{BS} & \textsc{BS{\scriptsize +bstr}} & \textsc{BS{\scriptsize +bstr+comb}} & \textsc{BS{\scriptsize +all}} & \textsc{BS{\scriptsize +all}} & \textsc{Isolate} & \textsc{LGP} \\
 degree & \multicolumn{3}{c||}{GPU} & GPU & CPU & Maple & Maple \\
\hline
1 & 0.06 & 0.05 & 0.1 & 0.09 & 0.12 & 0.14 & 0.13\\
2 & 0.23 & 0.48 & 0.24 & 0.21 & 0.36 & 0.47 & 0.40\\
3 & 3.28 & 1.94 & 0.53 & 0.39 & 0.66 & 0.92 & 0.87\\
4 & 26.62 & 9.21 & 1.38 & 1.03 & 2.07 & 2.81 & 2.65\\
5 & 241.74 & 23.74 & 3.22 & 1.93 & 4.24 & 6.92 & 6.05\\
6 & timeout & 65.23 & 6.26 & 3.00 & 6.21 & 10.81 & 10.01\\
7 & timeout & 136.56 & 19.81 & 11.52 & 21.33 & 52.11 & 50.37\\
8 & timeout & 221.74 & 38.8 & 22.52 & 35.77 & 107.27 & 107.84\\
9 & timeout & 569.67 & 66.19 & 30.87 & 50.00 & 170.10 & 169.87\\
10 & timeout & timeout & 117.21 & 46.32 & 69.99 & 280.90 & 277.94\\
\hline
\multicolumn{8}{|l|}{\multirow{2}{*}{(T)~sets of five curves with a vertically translated copy}}\\
\multicolumn{8}{|c|}{} \\
\hline
 & \textsc{BS} & \textsc{BS{\scriptsize +bstr}} & \textsc{BS{\scriptsize +bstr+comb}} & \textsc{BS{\scriptsize +all}} & \textsc{BS{\scriptsize +all}} & \textsc{Isolate} & \textsc{LGP} \\
 degree & \multicolumn{3}{c||}{GPU} & GPU & CPU & Maple & Maple \\
\hline
5 & 23.29 & 1.38 & 1.8 & 0.93 & 2.07 & 2.02 & 1.68\\
6 & 123.54 & 3.31 & 3.5 & 1.32 & 2.89 & 3.17 & 2.64\\
7 & 506.96 & 7.73 & 6.62 & 2.15 & 4.22 & 4.43 & 4.18\\
8 & timeout & 13.32 & 12.66 & 2.84 & 5.68 & 6.42 & 6.47\\
9 & timeout & 25.95 & 22.4 & 5.05 & 10.28 & 11.09 & 12.15\\
10 & timeout & 41.67 & 38.12 & 5.19 & 10.77 & 12.28 & 13.40\\
\hline
\multicolumn{8}{|l|}{\multirow{2}{*}{(P)~projected~intersection~curve of
surfaces with 8bit-coefficients}}\\
\multicolumn{8}{|c|}{} \\
\hline
 & \textsc{BS} & \textsc{BS{\scriptsize +bstr}} & \textsc{BS{\scriptsize +bstr+comb}} & \textsc{BS{\scriptsize +all}} & \textsc{BS{\scriptsize +all}} & \textsc{Isolate} & \textsc{LGP} \\
 degrees & \multicolumn{3}{c||}{GPU} & GPU & CPU & Maple & Maple \\
\hline
$3 \cdot 3$ & 0.10 & 0.11 & 0.11 & 0.08 & 0.14 & 0.18 & 0.14\\
$4 \cdot 4$ & 0.72 & 0.46 & 0.21 & 0.07 & 0.15 & 0.18 & 0.16\\
$5 \cdot 5$ & 98.16 & 27.09 & 1.92 & 1.00 & 2.36 & 3.25 & 3.19\\
$6 \cdot 6$ & timeout & 48.52 & 9.98 & 1.40 & 2.50 & 3.17 & 3.60\\
$7 \cdot 7$ & timeout & timeout & 94.75 & 19.90 & 27.73 & 29.38 & 29.53\\
$8 \cdot 8$ & timeout & timeout & 377.85 & 21.86 & 32.75 & 46.02 & 74.17\\
\hline
\end{tabularx}
\end{minipage}
\caption{Running times (in sec) for solving families of bivariate systems $f = f_y = 0$; \textbf{timeout:} algorithm timed out ($>$~600~sec)}
\label{tbl:app:bs}
\end{table}

\newpage
\section{Further experiments for analysing curves}
\label{asec:ca}

\begin{table}[H]
\sffamily\small
\centering
\begin{tabularx}{\linewidth}{ |>{\hsize=1.4\hsize}R| |>{\hsize=0.8\hsize}R| |>{\hsize=0.8\hsize}R||>{\hsize=1\hsize}R|>{\hsize=1\hsize}R| |>{\hsize=1\hsize\bfseries}R| } 

\hline
\multicolumn{6}{|l|}{\multirow{2}{*}{(R)~sets of five random dense curves}}\\
\multicolumn{6}{|c|}{} \\
\hline
degree, bits & \normalfont\textsc{BS{\scriptsize +all}} & \normalfont\textsc{Ak\_2} & \normalfont\slowca & \normalfont\fastca & \normalfont\ca \\
\hline
06, \hphantom{00}10 & 0.20 & 0.37 & 0.71 & 0.07 & 0.14\\
06, \hphantom{0}128 & 0.26 & 0.35 & 0.62 & 0.10 & 0.15\\
06, \hphantom{0}512 & 0.43 & 0.56 & 1.15 & 0.17 & 0.29\\
06, 2048 & 1.50 & 1.74 & 4.25 & 0.47 & 0.98\\
\hline
09, \hphantom{00}10 & 0.36 & 0.66 & 1.50 & 0.29 & 0.23\\
09, \hphantom{0}128 & 0.45 & 0.58 & 1.21 & 0.23 & 0.29\\
09, \hphantom{0}512 & 0.88 & 1.00 & 2.38 & 0.60 & 0.57\\
09, 2048 & 3.75 & 3.48 & 10.61 & 2.03 & 2.16\\
\hline
12, \hphantom{00}10 & 1.07 & 1.74 & 4.54 & 0.62 & 0.65\\
12, \hphantom{0}128 & 1.32 & 1.45 & 3.51 & 0.66 & 0.82\\
12, \hphantom{0}512 & 2.45 & 2.52 & 7.37 & 1.13 & 1.49\\
12, 2048 & 11.14 & 10.01 & 33.72 & 3.83 & 6.95\\
\hline
15, \hphantom{00}10 & 1.55 & 2.15 & 5.81 & 0.96 & 0.92\\
15, \hphantom{0}128 & 2.01 & 1.94 & 4.92 & 1.27 & 1.20\\
15, \hphantom{0}512 & 3.95 & 3.53 & 11.16 & 1.91 & 2.46\\
15, 2048 & 19.89 & 16.86 & 54.58 & 7.74 & 13.24\\
\hline
\hline
\multicolumn{6}{|l|}{\multirow{2}{*}{(R)~sets of five random sparse curves}}\\
\multicolumn{6}{|c|}{} \\
\hline
degree, bits & \normalfont\textsc{BS{\scriptsize +all}} & \normalfont\textsc{Ak\_2} & \normalfont\slowca & \normalfont\fastca & \normalfont\ca \\
\hline
06, \hphantom{00}10 & 0.10 & 0.22 & 0.25 & 0.06 & 0.07\\
06, \hphantom{0}128 & 0.11 & 0.23 & 0.26 & 0.08 & 0.08\\
06, \hphantom{0}512 & 0.20 & 0.34 & 0.42 & 0.12 & 0.13\\
06, 2048 & 0.58 & 1.07 & 1.39 & 0.42 & 0.36\\
\hline
09, \hphantom{00}10 & 0.18 & 1.05 & 0.54 & 0.20 & 0.11\\
09, \hphantom{0}128 & 0.22 & 1.00 & 0.48 & 0.27 & 0.13\\
09, \hphantom{0}512 & 0.37 & 1.30 & 0.78 & 0.39 & 0.20\\
09, 2048 & 1.21 & 4.46 & 2.79 & 1.38 & 0.68\\
\hline
12, \hphantom{00}10 & 0.28 & 1.62 & 0.88 & 0.21 & 0.17\\
12, \hphantom{0}128 & 0.36 & 1.62 & 0.93 & 0.25 & 0.22\\
12, \hphantom{0}512 & 0.66 & 2.45 & 1.73 & 0.47 & 0.42\\
12, 2048 & 2.70 & 8.49 & 7.23 & 1.89 & 1.94\\
\hline
15, \hphantom{00}10 & 1.00 & 3.37 & 3.03 & 0.71 & 0.59\\
15, \hphantom{0}128 & 1.25 & 3.87 & 3.10 & 0.99 & 0.63\\
15, \hphantom{0}512 & 2.40 & 5.65 & 5.88 & 1.59 & 1.22\\
15, 2048 & 10.97 & 22.78 & 24.85 & 5.47 & 5.46\\

\hline
\end{tabularx}
\caption{Running times (in sec) for analyses of random algebraic curves}
\label{tbl:app:random}
\end{table}

\newpage

\begin{table}[H]
\sffamily\small
\centering
\begin{tabularx}{\linewidth}{ |>{\hsize=1.4\hsize}R| |>{\hsize=0.8\hsize}R| |>{\hsize=0.8\hsize}R||>{\hsize=1\hsize}R|>{\hsize=1\hsize}R| |>{\hsize=1\hsize\bfseries}R| } 

\hline
\multicolumn{6}{|l|}{\multirow{2}{*}{(I)~sets of five interpolated curves through points on a grid}}\\
\multicolumn{6}{|c|}{} \\
\hline
degree & \normalfont\textsc{BS{\scriptsize +all}} & \normalfont\textsc{Ak\_2} & \normalfont\slowca & \normalfont\fastca & \normalfont\ca \\
\hline

 5 & 0.27 & 0.51 & 0.79 & 0.18 & 0.20\\
 6 & 0.59 & 0.87 & 1.53 & 0.31 & 0.37\\
 7 & 1.16 & 1.69 & 2.98 & 0.49 & 0.73\\
 8 & 1.98 & 2.88 & 5.39 & 1.09 & 1.19\\
 9 & 3.70 & 4.98 & 9.49 & 1.59 & 2.37\\
10 & 5.46 & 7.62 & 15.89 & 3.36 & 3.22\\
11 & 10.94 & 13.52 & 28.99 & 5.51 & 6.57\\
12 & 23.09 & 27.56 & 57.91 & 12.37 & 13.61\\
13 & 45.54 & 46.90 & 113.87 & 18.20 & 26.26\\
14 & 101.96 & 88.76 & 219.89 & 43.99 & 56.47\\
15 & 214.54 & 160.36 & 451.29 & 69.20 & 114.63\\
16 & 410.46 & 312.27 & timeout & 69.65 & 236.39\\

\hline
\multicolumn{6}{|l|}{\multirow{2}{*}{(S)~sets of five parameterized curves on a sphere with 16bit-coefficients}}\\
\multicolumn{6}{|c|}{} \\
\hline
degree & \normalfont\textsc{BS{\scriptsize +all}} & \normalfont\textsc{Ak\_2} & \normalfont\slowca & \normalfont\fastca & \normalfont\ca \\
\hline
 1 & 0.09 & 0.11 & 0.18 & 0.03 & 0.07\\
 2 & 0.21 & 0.34 & 0.68 & 0.08 & 0.17\\
 3 & 0.39 & 0.70 & 1.51 & 0.29 & 0.26\\
 4 & 1.03 & 2.43 & 4.73 & 0.59 & 0.71\\
 5 & 1.93 & 5.99 & 10.17 & 0.98 & 1.33\\
 6 & 3.00 & 12.62 & 16.12 & 1.97 & 1.98\\
 7 & 11.52 & 16.35 & 49.50 & 12.95 & 7.42\\
 8 & 22.52 & 28.28 & 84.85 & 22.87 & 14.04\\
 9 & 30.87 & 39.74 & 119.61 & 27.49 & 21.37\\
10 & 46.32 & 53.28 & 154.56 & 27.91 & 28.16\\

\hline
\multicolumn{6}{|l|}{\multirow{2}{*}{(T)~sets of five curves with a vertically translated copy}}\\
\multicolumn{6}{|c|}{} \\
\hline
degree & \normalfont\textsc{BS{\scriptsize +all}} & \normalfont\textsc{Ak\_2} & \normalfont\slowca & \normalfont\fastca & \normalfont\ca \\
\hline
 5 & 0.93 & 5.72 & 5.85 & 0.55 & 0.53\\
 6 & 1.32 & 12.69 & 8.59 & 0.77 & 0.67\\
 7 & 2.15 & 29.40 & 13.27 & 1.22 & 1.07\\
 8 & 2.84 & 66.13 & 16.74 & 2.03 & 1.27\\
 9 & 5.05 & 134.75 & 27.93 & 5.39 & 2.23\\
10 & 5.19 & 286.69 & 29.27 & 5.71 & 2.30\\

\hline
\multicolumn{6}{|l|}{\multirow{2}{*}{(P)~projected~intersection~curve of
surfaces with 8bit-coefficients}}\\
\multicolumn{6}{|c|}{} \\
\hline
degree(s) & \normalfont\textsc{BS{\scriptsize +all}} & \normalfont\textsc{Ak\_2} & \normalfont\slowca & \normalfont\fastca & \normalfont\ca \\
\hline

$3\cdot3$ & 0.08 & 0.15 & 0.36 & 0.05 & 0.06\\
$4\cdot4$ & 0.21 & 0.67 & 1.81 & 0.35 & 0.12\\
$5\cdot5$ & 1.00 & 3.94 & 6.87 & 1.33 & 0.55\\
$6\cdot6$ & 1.40 & 220.02 & 383.45 & 2.57 & 0.68\\
$7\cdot7$ & 19.90 & timeout & 84.74 & 7.11 & 3.70\\
$8\cdot8$ & 21.86 & timeout & 117.57 & 19.56 & 6.17\\

\hline
\end{tabularx}
\caption{Running times (in sec) for analyses of algebraic curves of various families; \textbf{timeout:} algorithm timed out ($>$~600~sec)}
\label{tbl:app:misc}
\end{table}

\begin{table}[H]
\sffamily\small
\centering
\begin{tabularx}{\linewidth}{ |>{\hsize=1.6\hsize}R| |>{\hsize=0.7\hsize}R| |>{\hsize=0.7\hsize}R||>{\hsize=1\hsize}R|>{\hsize=1\hsize}R| |>{\hsize=1\hsize\bfseries}R| } 

\hline
\multicolumn{6}{|l|}{\multirow{2}{*}{(G)~random~singular~and~non-singular~curves}}\\
\multicolumn{6}{|c|}{} \\
\hline
type \hfill degree, bits & \normalfont\textsc{BS{\scriptsize +all}} & \normalfont\textsc{Ak\_2} & \normalfont\slowca & \normalfont\fastca & \normalfont\ca \\
\hline
non-singular            \hfill 20,	160	&2.76	&2.15	&6.47	&0.84	&1.27\\              
singular-$f$-$f_z$      \hfill 20,	161	&4.82	&109.31	&16.59	&1.34	&1.43\\          
singular-$f$-$g$        \hfill 20,	160	&4.56	&115.96	&16.17	&2.36	&1.11\\          
\hline                                                                                
non-singular            \hfill 30,	199	&19.26	&12.51	&45.09	&5.08	&9.30\\           
singular-$f$-$f_z$      \hfill 30,	201	&20.34	&timeout	&60.45	&9.39	&5.32\\        
singular-$f$-$g$        \hfill 30,	198	&29.89	&timeout	&90.79	&12.22	&5.38\\      
\hline                                                        
non-singular            \hfill 42,	237	&56.57	&40.66	&133.12	&23.27	&35.80\\        
singular-$f$-$f_z$      \hfill 42,	238	&64.24	&timeout	&372.99	&52.27	&25.50\\      
singular-$f$-$g$        \hfill 42,	237	&122.20	&timeout	&419.16	&39.55	&18.77\\     
\hline                                                        
non-singular            \hfill 56,	284	&367.99	&161.68	&timeout	&timeout	&129.88\\
singular-$f$-$f_z$      \hfill 56,	290	&214.05	&timeout	&timeout	&187.64	&121.79\\
singular-$f$-$g$        \hfill 56,	280	&timeout	&timeout	&timeout	&136.64	&77.53\\ 

\hline
\multicolumn{6}{|l|}{\multirow{2}{*}{(X)~special~curves (see Table \ref{tbl:app:special-desc} in~\ref{asec:instances} for descriptions)}}\\
\multicolumn{6}{|c|}{} \\
\hline
curve & \normalfont\textsc{BS{\scriptsize +all}} & \normalfont\textsc{Ak\_2} & \normalfont\slowca & \normalfont\fastca & \normalfont\ca \\
\hline

13\_sings\_9 & 0.97 & 2.66 & 3.74 & 0.22 & 0.61\\
FTT\_5\_4\_4 & 20.51 & timeout & 32.07 & 95.03 & 27.81\\
L4\_circles & 0.74 & 6.63 & 12.41 & 0.64 & 0.45\\
L6\_circles & 2.60 & 171.86 & 108.46 & 1.61 & 1.62\\
SA\_2\_4\_eps & 0.44 & 53.96 & 2.35 & 1.17 & 0.29\\
SA\_4\_4\_eps & 2.01 & 122.30 & 11.96 & 3.92 & 2.00\\
challenge\_12 & 7.35 & timeout & 16.11 & 64.75 & 12.50\\
challenge\_12\_1 & 19.17 & timeout & 48.95 & 185.55 & 35.65\\
compact\_surf & 4.06 & 81.56 & 19.66 & 7.43 & 5.31\\
cov\_sol\_20 & 5.77 & 43.40 & 14.06 & 4.22 & 2.41\\
curve24 & 8.22 & 38.22 & 27.58 & 8.36 & 3.54\\
curve\_issac & 0.88 & 2.63 & 5.46 & 0.33 & 0.37\\
cusps\_and\_flexes & 0.63 & 2.09 & 2.97 & 0.57 & 0.44\\
degree\_7\_surf & 7.74 & timeout & 57.41 & 6.23 & 4.19\\
dfold\_10\_6 & 1.55 & 35.40 & 10.74 & 8.97 & 0.90\\
grid\_deg\_10 & 1.20 & 1.55 & 3.19 & 1.18 & 0.73\\
huge\_cusp & 6.44 & 17.88 & 19.09 & 3.34 & 4.82\\
mignotte\_xy & 243.16 & timeout & 276.89 & 199.59 & 128.05\\
spider & 46.47 & timeout & 200.61 & 22.34 & 21.03\\
swinnerton\_dyer & 5.28 & 347.28 & 43.78 & 13.04 & 6.97\\
ten\_circles & 1.33 & 22.77 & 11.84 & 4.26 & 0.86\\

\hline
\end{tabularx}
\caption{Running times (in sec) for analyses of generated and special algebraic curves; \textbf{timeout:} algorithm timed out ($>$~600~sec)}
\label{tbl:app:special}
\end{table}

\end{appendix}

\end{document}

%% file: fig.pdf_t
\begin{picture}(0,0)%
\includegraphics{fig.pdf}%
\end{picture}%
\setlength{\unitlength}{4144sp}%
\begingroup\makeatletter\ifx\SetFigFontNFSS\undefined%
\gdef\SetFigFontNFSS#1#2#3#4#5{%
  \reset@font\fontsize{#1}{#2pt}%
  \fontfamily{#3}\fontseries{#4}\fontshape{#5}%
  \selectfont}%
\fi\endgroup%
\begin{picture}(4970,2182)(5209,-3626)
\put(6661,-3391){\makebox(0,0)[b]{\smash{{\SetFigFontNFSS{8}{9.6}{\familydefault}{\mddefault}{\updefault}{\color[rgb]{0,0,0}$\alpha_2$}%
}}}}
\put(7381,-3391){\makebox(0,0)[b]{\smash{{\SetFigFontNFSS{8}{9.6}{\familydefault}{\mddefault}{\updefault}{\color[rgb]{0,0,0}$\alpha_3$}%
}}}}
\put(5562,-3538){\makebox(0,0)[b]{\smash{{\SetFigFontNFSS{8}{9.6}{\familydefault}{\mddefault}{\updefault}{\color[rgb]{0,0,0}g($\alpha$,y)}%
}}}}
\put(6076,-3391){\makebox(0,0)[b]{\smash{{\SetFigFontNFSS{8}{9.6}{\familydefault}{\mddefault}{\updefault}{\color[rgb]{0,0,0}$\alpha_1$}%
}}}}
\put(5568,-3406){\makebox(0,0)[b]{\smash{{\SetFigFontNFSS{8}{9.6}{\familydefault}{\mddefault}{\updefault}{\color[rgb]{0,0,0}f($\alpha$,y)}%
}}}}
\put(10164,-3186){\makebox(0,0)[b]{\smash{{\SetFigFontNFSS{8}{9.6}{\familydefault}{\mddefault}{\updefault}{\color[rgb]{0,0,0}f(x,$\beta$)}%
}}}}
\put(10075,-3017){\makebox(0,0)[b]{\smash{{\SetFigFontNFSS{8}{9.6}{\familydefault}{\mddefault}{\updefault}{\color[rgb]{0,0,0}$\beta_1$}%
}}}}
\put(10074,-2582){\makebox(0,0)[b]{\smash{{\SetFigFontNFSS{8}{9.6}{\familydefault}{\mddefault}{\updefault}{\color[rgb]{0,0,0}$\beta_2$}%
}}}}
\put(10077,-2215){\makebox(0,0)[b]{\smash{{\SetFigFontNFSS{8}{9.6}{\familydefault}{\mddefault}{\updefault}{\color[rgb]{0,0,0}$\beta_3$}%
}}}}
\put(10078,-1715){\makebox(0,0)[b]{\smash{{\SetFigFontNFSS{8}{9.6}{\familydefault}{\mddefault}{\updefault}{\color[rgb]{0,0,0}$\beta_4$}%
}}}}
\put(10159,-3316){\makebox(0,0)[b]{\smash{{\SetFigFontNFSS{8}{9.6}{\familydefault}{\mddefault}{\updefault}{\color[rgb]{0,0,0}g(x,$\beta$)}%
}}}}
\end{picture}%